\documentclass[11pt]{article}

\makeatletter
\renewcommand{\thefootnote}{}

\usepackage{epsfig,amsmath,graphicx,amssymb,overpic}

\usepackage{setspace}
\usepackage{graphicx}% Include figure files
\usepackage{dcolumn}% Align table columns on decimal point
\usepackage{bm}% bold math
\usepackage{graphicx}
\usepackage{subfigure}
\usepackage{epsfig,amsmath,graphicx,amssymb,overpic,cite}

\usepackage{graphicx}% Include figure files
\usepackage{dcolumn}% Align table columns on decimal point
\usepackage{bm}% bold math
\usepackage{graphicx}
\usepackage{subfigure}
\usepackage{amsthm}

\usepackage{pdfsync}%pdf反向搜索用到的包
\usepackage{cite}
\usepackage{amscd}
\usepackage{amsmath}
\usepackage{latexsym}
\usepackage{amsfonts}
\usepackage{amssymb}
\usepackage{amsthm}
\usepackage{graphicx}
\usepackage{indentfirst}
\usepackage{enumerate}
\usepackage{amstext}
\usepackage[dvipdfm,
            pdfstartview=FitH,
            CJKbookmarks=true,
            bookmarksnumbered=true,
            bookmarksopen=true,
            colorlinks=true, %注释掉此项则交叉引用为彩色边框(将colorlinks和pdfborder 同时注释掉)
            pdfborder=001,   %注释掉此项则交叉引用为彩色边框
            citecolor=magenta,% magenta , cyan
            linkcolor=blue,
            ]{hyperref}

 \allowdisplaybreaks[4]

\newtheorem{remark}{Remark}

\usepackage{tikz}
\usetikzlibrary{shapes,arrows,positioning}
\usepackage{xcolor} %为了实现不同的颜色

\newtheorem{theorem}{Theorem}

\newtheorem{proposition}{Proposition}

\def\be{\begin{equation}}
\def\ee{\end{equation}}
\def\bee{\begin{eqnarray}}
\def\ene{\end{eqnarray}}
\def\bes{\begin{subequations}}
\def\ees{\end{subequations}}

\def\d{\displaystyle}
\def\v{\vspace{0.05in}}
\def\no{{\nonumber}}

\newtheorem{RH}{Riemann-Hilbert Problem}

\newcommand{\x}{\hat{x}}

\begin{document}

\baselineskip=13pt
\renewcommand {\thefootnote}{\dag}
\renewcommand {\thefootnote}{\ddag}
\renewcommand {\thefootnote}{ }

\pagestyle{plain}

\begin{center}
\baselineskip=16pt \leftline{} \vspace{-.3in} {\Large \bf Rigorous analysis of large-space and long-time asymptotics for the short-pulse soliton gases} \\[0.2in]
\end{center}

\begin{center} \small
{\bf Guoqiang Zhang$^{1}$, Weifang Weng$^2$, and Zhenya Yan}$^{1,3,*}$\footnote{$^{*}${\it Email address}: zyyan@mmrc.iss.ac.cn (Corresponding author)}  \\[0.1in]
{$^1${\footnotesize KLMM,  Academy of Mathematics and Systems Science,  Chinese Academy of Sciences, Beijing 100190, China}}\\
{$^2${\footnotesize School of Mathematical Sciences,University of Electronic Science and Technology of China, Chengdu 611731, China}}\\
{$^3${\footnotesize School of Mathematical Sciences, University of Chinese Academy of Sciences, Beijing 100049, China}} \\[0.18in]
\end{center}

%\noindent \rule[0.25\baselineskip]{\textwidth}{0.6pt}

\noindent {\bf Abstract:}\, {\small
We rigorously analyze the asymptotics of soliton gases to the short-pulse equation. The soliton gas is formulated in terms of a Riemann-Hilbert problem, which is derived from the Riemann-Hilbert problems of the $N$-soliton solutions in the limit as $N \to \infty$. Building on prior work in the study of the KdV soliton gas and orthogonal polynomials with Jacobi-type weights, we extend the reflection coefficient to two generalized forms on the interval $\left[\eta_1, \eta_2\right]$:
\begin{itemize}
    \item $r_0(\lambda) = \left(\lambda - \eta_1\right)^{\beta_1}\left(\eta_2 - \lambda\right)^{\beta_2}|\lambda - \eta_0|^{\beta_0}\gamma(\lambda)$,
    \item $r_c(\lambda) = \left(\lambda - \eta_1\right)^{\beta_1}\left(\eta_2 - \lambda\right)^{\beta_2}\chi_c(\lambda)\gamma(\lambda)$,
\end{itemize}
where $0 < \eta_1 < \eta_0 < \eta_2$ and $\beta_j > -1$ for $j = 0, 1, 2$. The function $\gamma(\lambda)$ is continuous and positive on $\left[\eta_1, \eta_2\right]$, with an analytic extension to a neighborhood of this interval. The step function $\chi_c(\lambda)$ is defined as $\chi_c(\lambda) = 1$ for $\lambda \in \left[\eta_1, \eta_0\right)$ and $\chi_c(\lambda) = c^2$ for $\lambda \in \left(\eta_0, \eta_2\right]$, where $c$ is a positive constant with $c \neq 1$.
The asymptotic analysis is performed using the steepest descent method introduced by Deift and Zhou. A key aspect of the analysis is the construction of the $g$-function, which plays a critical role in the conjugation procedure. Notably, as a negative flow of the WKI system, the short pulse equation requires a distinct $g$-function construction compared to other integrable systems. To address the singularity at the origin, we introduce an innovative piecewise definition of the $g$-function.
In order to establish the order of the error term, we construct local parametrices near the endpoints $\eta_j$ for $j = 1, 2$, and the singularity at $\eta_0$. At the endpoints, we employ the Airy parametrix and the first type of modified Bessel parametrix. At the singularity $\eta_0$, we use the second type of modified Bessel parametrix for $r_0$ and the confluent hypergeometric parametrix for $r_c(\lambda)$.}
%\end{abstract}

%%Graphical abstract  %可选项

%\item  The KdV soliton gases are characterized by the Riemann-Hilbert problem with two types of generalized reflection coefficients.
%\item The rigorous long-time asymptotics of KdV soliton gases are derived by the Deift-Zhou steepest descent technique for the Riemann-Hilbert problem.
%\item For the first type of generalized reflection coefficient, the local parametrix near $\eta_0$ is constructed using modified Bessel functions, whereas for the second type, it is constructed using confluent hypergeometric functions.

\vspace{0.1in}
{\small \noindent {\it MSC 2020:} 35Q15; 35Q51; 35P25; 37K40

\vspace{0.05in}
  }

%\end{keyword}

%\end{frontmatter}

%\vspace{0.1in} \noindent {\small {\it MSC:} \, 35P25; 35Q15; 37K40; 35B40 }

%\vspace{0.1in} \noindent {\small {\it Keywords:}\, Modified Camassa-Holm equation with linear dispersion; Inverse spectral method;
%$N_\infty$-soliton asympototics; Riemann-Hiblert problem; Discrete spectral region}

%\noindent \rule[0.25\baselineskip]{\textwidth}{0.6pt}

\vspace{0.2in}

\begin{spacing}{1.05}
\tableofcontents
\end{spacing}

%\vspace{0.1in} \noindent {\bf Mathematics Subject Classification}\, 35B30; 35G25; 35B44; 35Q35

\baselineskip=15pt


\section{Introduction and statement of results}
\label{}

\subsection{Introduction}

Since the concept of `soliton' was coined by Zabusky and Kruskal~\cite{soliton} in 1965, multi-solitons has found been been found in many integrable nonlinear partial differential equations via some powerful approaches, such as the inverse scattering transform (IST)~\cite{GGKM}, Hirota bilinear method~\cite{35}, Darboux transform~\cite{DT, Gu}, and etc.. In particular, the study of interactions of multi-soliton solutions play a significant role in the filed of soliton theory and mathematical physics~\cite{soliton-book91}. In 1971, Zakharov~\cite{1} first proposed the concept of  {\it soliton gas} (i.e., the $N \to \infty$ limit behavior of the $N$-soliton solutions), which pertains to an infinite collection of weakly interacting solitons within the framework of the Korteweg-de Vries (KdV) equation, modeled through a kinetic equation governing the spectral distribution functions of this ensemble. Zakharov's initial formulation, based on a rarefied soliton gas, was later expanded to account for denser soliton gases in the KdV framework \cite{2}. This generalization involved developing a spectral theory based on the thermodynamic limit of finite-gap solutions and extended to models of soliton gas behavior in both the KdV equation \cite{3} and breather gases for the focusing nonlinear Schr\"odinger (NLS) equation \cite{4}. More recent work also applied these methods to bidirectional dispersive hydrodynamics within the defocusing NLS framework \cite{5}. Now soliton gas becomes a highly significant field of mathematical physics and integrable systems.

In addition to spectral theory, the study of soliton gases has spurred significant interest in other integrable nonlinear wave systems, including their integrable reductions \cite{6}, hydrodynamic reductions \cite{7}, and minimal energy configurations \cite{8}, as well as investigations into classical integrability in hydrodynamic models \cite{9}. Numerical studies have further suggested that the nonlinear phase of spontaneous modulational instability (MI)~\cite{10} and the emergence of rogue waves \cite{11,12} can be fundamentally attributed to soliton gas dynamics.

In recent years, asymptotic analyses of soliton gases have gained increasing momentum, particularly through the use of Riemann-Hilbert (RH) technique. The RH problem was originally introduced by Zakharov and Shabat~\cite{z1973} in the study of solitons of the NLS equation. The dressing method with a corresponding RH problem has been employed to describe the KdV soliton gas, notably involving two reflection coefficients in the jump conditions~\cite{13}. The asymptotics of this model was thoroughly examined when only one reflection coefficient was considered via the Deift-Zhou steepest descent method~\cite{14,15,16,17}. This idea was extended to study the interaction between a soliton gas and a large soliton in the modified KdV equation \cite{18}, yielding key insights into local wave properties, including phase shifts, soliton peak locations, and the soliton’s average velocity as governed by a kinetic equation. Moreover, soliton gases have been extended beyond real spectral domains to the focusing NLS equation in bounded domains, uncovering the intriguing phenomenon of soliton shielding \cite{19}. Despite this progress, there remain significant open problems in the field, as highlighted by a recent review paper \cite{20}. One such problem is the rigorous asymptotic analysis of soliton gases in other integrable systems, such as Camassa-Holm type equations, which serves as the primary motivation for this paper.

Inspired by the advancements in \cite{18,21}, this paper aims to tackle this problem.
The focus of our study is the nonlinear short pulse (SP) equation, given by
\begin{gather} \label{short-pulse}
u_{xt}=u+\frac{1}{6}\left(u^3\right)_{xx},\qquad u=u(x,t)\in\mathbb{R}[x,t],
\end{gather}
which was presented to describe the wave propagation of ultra-short  light pulses in silica optical fibers~\cite{sp}, and differs from the usual NLS model modelling the pulse propagation in optical fibers~\cite{Hase}.
The SP equation (\ref{short-pulse}) can be regarded as the short-wave approximation $\epsilon\to 0$ with the scaling transforms ($x\to \epsilon x,\, t\to \epsilon^{-1}t,\, u\to \epsilon u,$) of the modified Camassa-Holm equation~\cite{mch-1,mch-2,mch-3}
\begin{equation}
m_t+u_x+\left((u^2-u^2_x)m \right) _x=0,\hspace{0.5cm} m=u-u_{xx}.
\end{equation}
Many properties of the SP equation (\ref{short-pulse}) have been studied. For example, the SP equation (\ref{short-pulse}) was shown to possess a Wadati-Konno-Ichikawa (WKI) type Lax pair~\cite{lax-SS}:
\bee \label{lax}
\left\{\begin{array}{l}
\Psi_x=U^X\Psi, \qquad U^X=\lambda \left(\sigma_3+u_x\sigma_1\right), \v\\
\Psi_t=U^T\Psi, \qquad U^T=\dfrac{\sigma_3}{4\lambda}+\dfrac{1}{2}\left(u^2U^X+u\sigma_1\sigma_3\right).
 \end{array}\right.
 \ene
where $\Psi=\Psi(x,t;\lambda)$ is a $2\times 2$ matrix-valued function, $\lambda\in\mathbb{C}$ is an iso-spectral parameter, and the three Pauli matrices are defined as
\begin{gather}\label{sigma}
\sigma_1 =
\begin{pmatrix}
0 & 1 \\
1 & 0
\end{pmatrix}, \quad
\sigma_2 =
\begin{pmatrix}
0 & -\mathrm{i} \\
\mathrm{i} & 0
\end{pmatrix}, \quad
\sigma_3 =
\begin{pmatrix}
1 & 0 \\
0 & -1
\end{pmatrix}.
\end{gather}
The bi-Hamiltonian structures and local and nonlocal conversation laws were found for a hierarchy of the SP equation (\ref{short-pulse})~\cite{Brun1,Brun2}. The SP equation (\ref{short-pulse}) was shown to be changed into the sine-Gordon equation
\bee
 u_{\xi\tau}=\sin u,
\ene
 and the coupled integrable dispersionless equation~\cite{ko}
\bee
 u_{ys}=2\rho u,\quad \rho_y+2uu_y=0
\ene
via the hodograph transform and other transforms~\cite{lax-SS,SS-1}. Moreover, some types of exact soliton solutions were found for the SP equation (\ref{short-pulse}), including multi-soliton, multi-loop, multi-breather, and periodic solutions~\cite{SS-1, Mat-1,Mat-2,tang24}. The well-posedness of the Cauchy problem and wave breaking were investigated for the SP equation (\ref{short-pulse})~\cite{cau-1,cau-2,wb}.  The RH problem was presented for the SP equation (\ref{short-pulse}) to study the
long-time asymptotic behaviors of its initial value problem (the initial value decays sufficiently fast as $|x|\to \infty$)~\cite{22} via the modified Deift-Zhou nonlinear steepest descent method~\cite{lenells17,14}. Recently, the long-time asymptotic behaviors of the SP equation (\ref{short-pulse}) with other types of initial values belonging to the Schwartz space~\cite{xu18} or the weighted Sobolev space~\cite{fan21,li23} via the Deift-Zhou method~\cite{14} or the $\bar\partial$ steepest descent method~\cite{dbar1,dbar2}. The B\"acklund transform was presented for the  SP equation (\ref{short-pulse})~\cite{BT}.
Some integrable extensions of the SP equation (\ref{short-pulse}) were also presented, such as the complex SP equation, positive flow SP equation, nonlocal SP equation, multi-component SP equations, and discretization version~\cite{gsp,gsp1,gsp2,gsp3,gsp4,gsp5,gsp6,gsp7,gsp8,dis}.

Using the WKI-type Lax pair~\eqref{lax}, a sequence of $N$-soliton solutions of the  SP equation (\ref{short-pulse}) can be derived via the inverse scattering transform, linked to a RH problem~\cite{22}.

\begin{RH}\label{RH-1} Let $M^N=M^N(\lambda; \hat{x}, t)$ denotes the solution to the following RH problem:

\begin{itemize}

\item{} The function $M^N$ is analytic except for discrete spectral points $\pm \lambda_j$, for $j=1, 2, \cdots, N$, where $\lambda_j\in\left(\eta_1, \eta_2\right)$.

\item{} The normalization at infinity is given by $M^N\to \mathbb{I}_2$ as $\lambda\to\infty$.

 \item{} The discrete spectra $\pm \lambda_j$ are simple poles, and the following residue conditions hold
\begin{gather}
\begin{aligned}
\mathop\mathrm{Res}\limits_{\lambda=\lambda_j}M^N(\lambda; \hat{x}, t)&=\lim_{\lambda\to\lambda_j}M^N(\lambda; \hat{x}, t)\mathcal{L}^{t\theta}\left[c_j\right],\\
\mathop\mathrm{Res}\limits_{\lambda=-\lambda_j}M^N(\lambda; \hat{x}, t)&=\lim_{\lambda\to -\lambda_j}M^N(\lambda; \hat{x}, t)\mathcal{U}^{t\theta}\left[c_j\right],
\end{aligned}
\end{gather}
where $c_j$ is a negative real number and the terms $\mathcal{L}^{t\theta}\left[c_j\right]$ and $\mathcal{U}^{t\theta}\left[c_j\right] $ are defined in \eqref{LU}, and the phase $\theta=\theta(\lambda; \hat{x}, t)$ is given by
\begin{gather}
\theta=\frac{\xi\lambda+\lambda^{-1}}{4}, \quad \xi=\frac{4\hat{x}}{t}
\end{gather}
with $\hat{x}$ being related to $x$ through the reciprocal transformation
\begin{gather}
\hat{x}=x-\int_x^{+\infty} \left(\sqrt{u_y^2(y, t)+1}-1\right)\mathrm{d} y.
\end{gather}
\end{itemize}
\end{RH}

Therefore, the $N$-soliton solution sequence $u_N(x, t)$ of the  SP equation (\ref{short-pulse}) can be recovered from $M^N$ by the following relations
\bee\label{uN}
\begin{array}{rl}
u_N(x, t)= &\lim\limits_{\lambda\to 0}\lambda^{-1}\left(M^N(0; \x, t)^{-1}M^N(\lambda; \x, t)\right)_{1, 2}, \v\\
x=& \x+\lim\limits_{\lambda\to 0}\lambda^{-1}\left(\left(M^N(0; \x, t)^{-1}M^N(\lambda; \x, t)\right)_{1, 1}-1\right).
\end{array}
\ene

In the work of Girotti {\it et al}~\cite{21}, certain assumptions were imposed on the reflection coefficient to guarantee well-defined asymptotic behavior in the solution of the associated RH problem related to the KdV soliton gases. Specifically, the reflection coefficient was assumed to: 
\begin{itemize}

\item [(i)] be continuous and strictly positive on the interval $\lambda \in [\eta_1, \eta_2]$; 

\item[(ii)] admit an analytic extension to a neighborhood of this interval; 

\item[(iii)] exhibit symmetry, taking the same values over $[\eta_1, \eta_2]$ and $[-\eta_2, -\eta_1]$.
 \end{itemize}
 These conditions lead to a solution
$Y$ for the RH problem, characterized by logarithmic singularities at the endpoints $\eta_1$ and $\eta_2$, allowing the construction of local parametrices near these points in terms of modified Bessel functions of index 0.

A different scenario was explored for the mKdV soliton gases in \cite{18}, where the reflection coefficient near the endpoints behaves as $\left|\lambda\mp\eta_j\right|^{\pm1/2}$. In this case, no local parametrix constructions are required near the endpoints $\lambda=\pm\eta_j$, as the outer parametrix alone suffices to capture the solution's behavior. The outer parametrix, originally developed to analyze long-time asymptotics for the KdV equation with step-like initial data \cite{23}, was first introduced in \cite{24} for the asymptotic study of orthogonal polynomials associated with Hermitian matrix models. This approach has since been adapted to various integrable systems, such as the long-time asymptotics of the NLS shock problem \cite{25}, the modified KdV equation \cite{26}, and the NLS equation with nonzero boundary conditions \cite{27,28}.

Note that the reflection coefficients play an important role in the RH problems related to the integrable systems, and are similar to the Jacobi weights of orthogonal polynomials~\cite{op1,op2,op3}. In this paper, based on these ideas, we would like to consider two generalized reflection coefficients for the SP soliton gases, and analyze their associated asymptotic behaviors. In the first case, the reflection coefficient is generalized as
\begin{gather}\label{r0}
r_0(\lambda)=\left(\lambda-\eta_1\right)^{\beta_1}\left(\eta_2-\lambda\right)^{\beta_2}\left|\lambda-\eta_0\right|^{\beta_0}\gamma\left(\lambda\right),
\end{gather}
where $\gamma(\lambda)$, consistent with the setting in \cite{21}, is continuous and strictly positive for $\left[\eta_1, \eta_2\right]$ and analytic in a neighborhood of this interval. Unlike the reflection coefficient in \cite{21}, which remains positive at the endpoints, the generalized coefficient $r_0$ exhibits zeros and singularities at both $\eta_1$ and $\eta_2$,  and also at an interior point $\eta_0\in\left(\eta_1, \eta_2\right)$. This modification introduces more intricate behaviors due to the singularity at $\eta_0$, which is absent in the original formulation. The absolute value ensures positivity in the subintervals $\left(\eta_1, \eta_0\right)$ and $\left(\eta_0, \eta_2\right)$.
The second generalized reflection coefficient is defined as
\begin{gather}\label{rc}
r_c(\lambda)=\left(\lambda-\eta_1\right)^{\beta_1}\left(\eta_2-\lambda\right)^{\beta_2}\chi_c\left(\lambda\right)\gamma\left(\lambda\right),
\end{gather}
where $r_c$  retains the singular behaviors at $\eta_1$, $\eta_2$, and $\eta_0$, but introduces a jump discontinuity at $\eta_0$ via the step function  $\chi_c(\lambda)$. Specifically, $\chi_c\left(\lambda\right)=1$ for $\lambda\in\left[\eta_1, \eta_0\right)$ and $\chi_c\left(\lambda\right)=c^2$ for $\lambda\in\left(\eta_0, \eta_2\right]$, where $c$ is a positive constant with $c\ne1$. The exponents $\beta_1, \beta_2$ and $\beta_0$  are constrained by
 $\beta_1, \beta_2, \beta_0>-1$, a condition that ensures the convergence of certain integrals, as discussed later.

By employing an adaptation of the interpolation method that transforms poles into jump conditions, as introduced by Deift {\it et al}~\cite{29} in the context of the Toda rarefaction problem, we construct a sequence of $N$-soliton solutions, denoted by $u_N$, with $N=N_1+N_2$.

\begin{RH} \label{RH-2}
For simplicity, we use $M^N$ to represent the solution to the following deformed Riemann-Hilbert problem:

\begin{itemize}

\item{} The matrix-valued function $M^N$is analytic in $\lambda$ except on two non-intersecting closed contours, $\Gamma_+$ and $\Gamma_-$, which encircle the intervals $\left[\eta_1, \eta_2\right]$ and $\left[-\eta_2, -\eta_1\right]$ respectively, in a counterclockwise direction.

\item{} At infinity, $M^N_\pm$ normalizes to the identity matrix: $M^N\to\mathbb{I}_2$ as $\lambda\to \infty$.

 \item{} The jump conditions on the contours $\Gamma_\pm$ are given by
\begin{gather}\label{MN}
M^N_+=M^N_-
\begin{cases}
\mathcal{L}^{t\theta}\left[-\d\sum_{j=1}^{N_1}\dfrac{\left(\eta_0-\eta_1\right)\,r\left(\lambda_{1, j}\right)}{2N_1\pi\,\left(\lambda-\lambda_{1, j}\right)}
-\d\sum_{j=1}^{N_2}\dfrac{\left(\eta_2-\eta_0\right)\,r\left(\lambda_{2, j}\right)}{2N_2\pi\,\left(\lambda-\lambda_{2, j}\right)}
\right], & \lambda\in\Gamma_+, \\[2em]
\mathcal{U}^{t\theta}\left[-\d\sum_{j=1}^{N_1}\dfrac{\left(\eta_0-\eta_1\right)\,r\left(-\lambda_{1, j}\right)}{2N_1\pi\,\left(\lambda+\lambda_{1, j}\right)}
-\d\sum_{j=1}^{N_2}\dfrac{\left(\eta_2-\eta_0\right)\,r\left(-\lambda_{2, j}\right)}{2N_2\pi\,\left(\lambda+\lambda_{2, j}\right)}\right],  & \lambda\in\Gamma_-.
\end{cases}
\end{gather}
\end{itemize}
\end{RH}
Here, the discrete spectral points $\lambda_{1, j}\in\left(\eta_1, \eta_0\right)$ and $\lambda_{2, j}\in\left(\eta_0, \eta_2\right)$ are equally spaced as  $\lambda_{1, j}=\eta_1+j(\eta_0-\eta_1)/(N_1+1)$ for $j=1, 2, \cdots N_1$ and $\lambda_{2, j}=\eta_0+j(\eta_2-\eta_0)/(N_2+1)$ for $j=1, 2, \cdots N_2$.
This arrangement leads to a distribution of discrete spectral points within the interval $\left(\eta_1, \eta_2\right)$, where the reflection coefficient may exhibit singularities of order $\beta_j$ at the endpoints.
In contrast to \cite{21}, where the singularity at $\eta_0$ was absent, the modified setup here allows more complex local behaviors near $\eta_0$, including the order $\beta_0$ singularity in $r_0$ and the jump discontinuity in $r_c$. For both generalized reflection coefficients $r_0$ and $r_c$, the values on the interval $\left(-\eta_2, -\eta_1\right)$ are determined by the symmetry $r(\lambda)=r(-\lambda)$.

The limit process has proven to be a highly effective method for deriving novel nonlinear wave solutions that are otherwise inaccessible through conventional direct approaches within the theory of integrable systems. For instance, an infinite-order rogue wave solution was constructed in \cite{30} by applying the limit to a sequence of Riemann-Hilbert problems corresponding to $N$th-order rogue waves, utilizing the inverse scattering transform framework \cite{31}. This technique succeeds in situations where methods such as the Darboux transformation \cite{32,33,34}, the Hirota bilinear method \cite{35,36}, and the inverse scattering transform itself \cite{37,38} fail to yield explicit solutions. A similar limit procedure has been employed to derive infinite-order soliton solutions for the focusing NLS equation through the RH analysis \cite{39}. In this context, a suitable rescaling transform plays a pivotal role in ensuring the convergence of the jump matrices sequence, and it also allows the study of large-$N$ asymptotics of $N$-soliton solutions with initial data of the form $N\mathrm{sech}(x)$ to be recast as a semiclassical limit problem \cite{40}. By contrast, in \cite{21}, the key step involves carefully selecting norming constants, enabling the limit to converge to a well-defined Riemann integral.

In the present paper, we adopt a similar approach, where the norming constants are taken as discretizations of generalized reflection coefficients $r$. The limit, in this case, converges not only to definite integrals but also to improper integrals. Alternatively, this limit procedure may be viewed as replacing the reflection coefficient with its semiclassical approximation, a perspective aligned with the results in \cite{41,42}.
By taking the limits $N_1\to\infty$ and $N_2\to\infty$, we derive a RH problem characterizing the soliton gas.

\begin{RH} \label{RH-3} The resulting $2\times 2$ matrix solution, denoted by $M^\infty=M^\infty(\lambda; \x, t)$, satisfies the following properties:

\begin{itemize}

\item{} $M^\infty$ is analytic in $\lambda$ for $\lambda\in\mathbb{C}\setminus\left(\Gamma_+\cup\Gamma_-\right)$;

 \item{} $M^\infty$ normalizes as $M^\infty\to \mathbb{I}2$ as $\lambda\to\infty$;

 \item{} The continuous boundary values of $M^\infty$ are related by the jump conditions
\begin{gather}\label{Minfty}
M^\infty_+=M^\infty_-
\begin{cases}
\mathcal{L}^{t\theta}\left[\mathrm{i}\left(\mathcal{P}_1+\mathcal{P}_2\right)
\right], &\mathrm{for}\,\,\, \lambda\in\Gamma_+,\\[0.5em]
\mathcal{U}^{t\theta}\left[\mathrm{i}\left(\mathcal{P}_{-1}+\mathcal{P}_{-2}\right)\right],  &\mathrm{for}\,\,\, \lambda\in\Gamma_-,
\end{cases}
\end{gather}
where the terms $\mathcal{P}_1, \mathcal{P}_2, \mathcal{P}_{-1}, \mathcal{P}_{-2}$ are defined as
$\mathcal{P}_1=\int_{\eta_1}^{\eta_0}\frac{r(s)}{s-\lambda}\mathrm{d}s$, $\mathcal{P}_2=\int_{\eta_0}^{\eta_2}\frac{r(s)}{s-\lambda}\mathrm{d}s$, $\mathcal{P}_{-1}=\int_{-\eta_0}^{-\eta_1}\frac{r(s)}{s-\lambda}\mathrm{d}s$, and $\mathcal{P}_{-2}=\int_{-\eta_2}^{-\eta_0}\frac{r(s)}{s-\lambda}\mathrm{d}s$.
\end{itemize}
\end{RH}

Following the strategy of Zhou's vanishing lemma, as implemented in \cite{18} for RH problem 7, the existence and uniqueness of the solution $M^\infty$ can be similarly established. Moreover, the soliton gas solution $u(x, t)$ of the SP equation can be recovered from $M^\infty$ through the formulas:
\begin{gather}
\begin{array}{rl}
u(x, t)=& \lim\limits_{\lambda\to 0}\lambda^{-1}\left(M^\infty(0; \x, t)^{-1}M^\infty(\lambda; \x, t)\right)_{1, 2},\\[1em]
x=& \x+\lim\limits_{\lambda\to 0}\lambda^{-1}\left(\left(M^\infty(0; \x, t)^{-1}M^\infty(\lambda; \x, t)\right)_{1, 1}-1\right).
\end{array}
\end{gather}
In this limit process, it is crucial to ensure the validity of the jump matrices as we transition from \eqref{MN} to \eqref{Minfty}. This holds if the parameters $\beta_1$, $\beta_2$, and $\beta_0$ associated with the reflection coefficient satisfy the conditions $\beta_1, \beta_2, \beta_0 > -1$. To clarify this, two cases are considered. First, when $\beta_1, \beta_2, \beta_0 \geq 0$, the validity follows immediately from the definition of the definite Riemann integral. Second, if any of $\beta_1$, $\beta_2$, or $\beta_0$ lies in the interval $(-1, 0)$, improper integrals of the second kind arise due to the presence of singularities, and hence the convergence cannot be deduced from the standard theory of definite integrals. Nevertheless, it can still be shown that the Riemann sum converges to the corresponding improper integrals. This proof relies on basic calculus and can be carried out using the principles of monotonicity and uniform continuity.

\subsection{Statement  of main results}

A series of transformations, denoted as $Y \mapsto T \mapsto S \mapsto E$, are performed to ensure that the matrix $E$ is normalized to the identity matrix $\mathbb{I}_2$ at infinity, and that its jump matrices exhibit uniform exponential decay towards the identity. These transformations are based on the approach developed by Deift, Kriecherbauer, McLaughlin, Venakides, and Zhou~\cite{24,43}, who employed this technique to study the asymptotics of orthogonal polynomials with exponential weights, following the RH formulation introduced by Fokas, Its, and Kitaev \cite{44,45}.
The Deift-Zhou steepest descent method for orthogonal polynomials has been applied to a wide range of weight functions, including logarithmic weights \cite{46,47}, Freud weights \cite{48}, and Laguerre polynomials \cite{49,50,op3,52}. Applications have also extended to measures supported on the plane \cite{53}, Jacobi weights \cite{54}, modified Jacobi weights \cite{op2,56}, and discontinuous Gaussian weights \cite{57}, among others. Beyond orthogonal polynomials, this technique has found utility in various related problems, such as analyzing the distribution of the length of the longest increasing subsequence in random permutations \cite{15}, and investigating the asymptotics of discrete holomorphic maps \cite{58}.

In this paper, we would like to use the Deift-Zhou nonlinear steepest descent method to derive the rigorous asymptotics for soliton gases in the SP equation, using two types of generalized reflection coefficients, $r_0$ and $r_c$, across different regions of interest: $\left(-\infty, \xi_\mathrm{crit}\right)$, $\left(\xi_\mathrm{crit}, \xi_0\right)$, $\left(\xi_0, -\eta_2^{-2}\right)$, and $\left(-\eta_2^{-2}, +\infty\right)$. The critical values $\xi_\mathrm{crit}$ and $\xi_0$ are uniquely determined by the following expressions:
\begin{gather}
\xi_\mathrm{crit} = -\eta_2^{-2} \, W\left(\frac{\eta_1}{\eta_2}\right), \quad \xi_0 = -\eta_2^{-2} \, W\left(\frac{\eta_0}{\eta_2}\right),
\end{gather}
where $W: m \mapsto W(m)$ is defined as:
\begin{gather}
W(m) = \frac{eE(m) / eK(m)}{m \left(m^2 - 1 + eE(m) / eK(m)\right)},
\end{gather}
with $eE$ and $eK$ representing the complete elliptic integrals, defined in \eqref{elliptic12}.

The function $W$ plays a key role in the construction of the so-called $g$-function, which is used in the conjugation process to ensure exponential decay along the lenses. As a result, the Airy parametrix provides an accurate local approximation around $\lambda = \alpha$. The value of $\alpha$ is determined through the Whitham evolution equation \cite{59}, which is given by:
\begin{gather}\label{Whitham}
\xi = -\eta_2^{-2} \, W\left(\frac{\alpha}{\eta_2}\right).
\end{gather}

The main results of this paper are summarized as follows.

\begin{theorem}[Large-$\x$ asymptotics for the initial value $u(x, 0)$ of the soliton gas]\label{large-x}
For these two types of generalized reflection coefficients $r=r_0$ and $r=r_c$, the large-$\x$ asymptotics for the initial value
$u(x, 0)$ of  the soliton gas are established as follows.

\begin{enumerate}[0.]
\item[\rm{\textbullet}]  For $\beta_2, \beta_1, \beta_0\ge 0$, in the limit of $\x\to+\infty$, there exists a positive constant $\mu_0$ such that
\begin{gather}\label{large-x-right}
\begin{array}{rl}
u(x, 0)=&\mathcal{O}\left(\mathrm{e}^{-\mu_0 \x}\right),\\[0.5em]
x=&\x+\mathcal{O}\left(\mathrm{e}^{-\mu_0 \x}\right).
\end{array}
\end{gather}
\item[\rm{\textbullet}] For $\beta_2, \beta_1, \beta_0>-1$, in the limit of $\x\to-\infty$,  we have
\bee\label{large-x-left}
\begin{array}{rl}
u(x, 0)=&\mathrm{e}^{\Delta_1^0}\left(\Psi_1^0\tilde{\Phi}_1^0+\Phi_1^0\tilde{\Psi}_1^0\right)\Theta_1^0
+\mathcal{O}\left(\dfrac{1}{\left|\x\right|}\right),        \\[1em]
x=&\x+\left(\Psi_1^0\tilde{\Psi}_1^0-\Phi_1^0\tilde{\Phi}_1^0\right)\Theta_1^0+\tilde{\mathcal{X}}_1-\mathcal{X}_1^0
+\mathcal{O}\left(\dfrac{1}{\left|\x\right|}\right),
\end{array}
\ene
where
\begin{gather}
\Psi_1^0=\frac{\vartheta_3\left(\frac{\tau_1}{2}-\frac{\Delta_1^0}{2\pi \mathrm{i}}; \tau_1\right)}{\vartheta_3\left(\frac{\tau_1}{2}; \tau_1\right)}, \quad
\Theta_1^0=\frac{\vartheta_3^2\left(0; \tau_1\right)}{\vartheta_3^2\left(\frac{\Delta_1^0}{2\pi\mathrm{i}}; \tau_1\right)},
\\[0.5em]
\tilde{\Psi}_1^0=\frac{1}{4\eta_1 eK(m_1)}\left(
\frac{\vartheta_3\left(\frac{\tau_1}{2}-\frac{\Delta_1^0}{2\pi \mathrm{i}}; \tau_1\right)\vartheta_3'\left(\frac{\tau_1}{2}; \tau_1\right)}{\vartheta_3^2\left(\frac{\tau_1}{2}; \tau_1\right)}-\frac{\vartheta_3'\left(\frac{\tau_1}{2}-\frac{\Delta_1^0}{2\pi \mathrm{i}}; \tau_1\right)}{\vartheta_3\left(\frac{\tau_1}{2}; \tau_1\right)}
\right),   \\[0.5em]
\Phi_1^0=2\left(m_1-1\right)eK(m_1)\frac{\vartheta_3\left(\frac{1+\tau_1}{2}+\frac{\Delta_1^0}{2\pi\mathrm{i}}; \tau_1\right)}{\vartheta_3'\left(\frac{1+\tau_1}{2}; \tau_1\right)},        \\[0.5em]
\tilde{\Phi}_1^0=\frac{m_1-1}{2\eta_1}\left(
\frac{\vartheta_3''\left(\frac{1+\tau_1}{2}; \tau_1\right)\vartheta_3\left(\frac{1+\tau_1}{2}+\frac{\Delta_1^0}{2\pi\mathrm{i}}; \tau_1\right)}{\vartheta_3'\left(\frac{1+\tau_1}{2}; \tau_1\right)^2}-\frac{\vartheta_3'\left(\frac{1+\tau_1}{2}+\frac{\Delta_1^0}{2\pi\mathrm{i}}; \tau_1\right)}{\vartheta_3'\left(\frac{1+\tau_1}{2}; \tau_1\right)}
\right),                            \\[0.5em]
\mathcal{X}_1^0=\left(1+\frac{1}{m_1}\left(\frac{eE(m_1)}{eK(m_1)}-1\right)\right) \x, \\[0.5em]
\tilde{\mathcal{X}}_1=\frac{\Omega_1\phi_1}{\pi \mathrm{i}\eta_1}\left(eE(m_1)-eK(m_1)\right)+\frac{\eta_1\eta_2}{\pi}\int_{\eta_1}^{\eta_2}\frac{\log r(s)}{s^2\sqrt{\left(\eta_2^2-s^2\right)\left(s^2-\eta_1^2\right)}}\mathrm{d}s,
\end{gather}
with
$m_1=\eta_1/\eta_2$,
$\tau_1=\mathrm{i}\,eK\left(\sqrt{1-m_1^2}\right)\left/\right.2eK\left(m_1\right)$,
$\Omega_1=-\pi \mathrm{i}\eta_2\left/\right.eK\left(m_1\right)$,
\begin{gather}
\Delta^0_1=\Omega_1(\x+\phi_1), \quad
\phi_1=-\int_{\eta_1}^{\eta_2}\frac{\log r(s)}{\sqrt{\left(s^2-\eta_1^2\right)\left(\eta_2^2-s^2\right)}}\frac{\mathrm{d}s}{\pi}.
\end{gather}

\end{enumerate}
\end{theorem}

\begin{theorem}[Long-time asymptotics for the soliton gas $u(x, t)$]
For these two types of generalized reflection coefficients $r=r_0$ and $r=r_c$, the long-time asymptotics for the short pulse soliton gas $u(x, t)$ are  established as follows.
\begin{enumerate}[0.]
\item[\rm{\textbullet}] For $\xi>-\eta_2^{-2}$ with $\beta_1, \beta_2, \beta_0\ge 0$, there exists a positive constant $\mu=\mu(\xi)$ such that
\begin{gather}\label{large-right}
\begin{array}{rl}
u(x, t)=&\mathcal{O}\left(\mathrm{e}^{-\mu t}\right),\\[0.5em]
x=& \x+\mathcal{O}\left(\mathrm{e}^{-\mu t}\right).
\end{array}
\end{gather}

\item[\rm{\textbullet}] For $\xi\in\left(\xi_0, -\eta_2^{-2}\right)$ with  $\beta_0\ge 0, \beta_1\ge 0, \beta_2>-1$,
and for  $\xi\in\left(\xi_{\mathrm{crit}}, \xi_0\right)$ with $\beta_0>-1, \beta_1\ge 0, \beta_2>-1$,
\begin{gather}\label{large-middle}
\begin{array}{rl}
u(x, t)=&\mathrm{e}^{\Delta^\alpha}\left(\Psi^\alpha \tilde{\Phi}^\alpha+\Phi^\alpha \tilde{\Psi}^\alpha \right)\Theta^\alpha
+\mathcal{O}\left(\dfrac{1}{t}\right),                                                                 \\[1em]
x=&\x+\left(\Psi^\alpha \tilde{\Psi}^\alpha-\Phi^\alpha \tilde{\Phi}^\alpha \right)\Theta^\alpha+\tilde{\mathcal{X}}^\alpha-\mathcal{X}^\alpha
+\mathcal{O}\left(\dfrac{1}{t}\right),
\end{array}
\end{gather}
where
\begin{gather}
\Psi^\alpha=\frac{\vartheta_3\left(\frac{\tau_\alpha}{2}-\frac{\Delta^\alpha}{2\pi \mathrm{i}}; \tau_\alpha \right)}{\vartheta_3\left(\frac{\tau_\alpha}{2}; \tau_\alpha \right)}, \quad
\Theta^\alpha=\frac{\vartheta_3^2\left(0; \tau_\alpha \right)}{\vartheta_3^2\left(\frac{\Delta^\alpha }{2\pi\mathrm{i}}; \tau_\alpha \right)}, \\[0.5em]
\tilde{\Psi}^\alpha=\frac{1}{4\alpha eK(m_\alpha)}\left(
\frac{\vartheta_3\left(\frac{\tau_\alpha}{2}-\frac{\Delta^\alpha}{2\pi \mathrm{i}}; \tau_\alpha \right)\vartheta_3'\left(\frac{\tau_\alpha}{2}; \tau_\alpha \right)}{\vartheta_3^2\left(\frac{\tau_\alpha}{2}; \tau_\alpha \right)}-\frac{\vartheta_3'\left(\frac{\tau_\alpha}{2}-\frac{\Delta^\alpha}{2\pi \mathrm{i}}; \tau_\alpha \right)}{\vartheta_3\left(\frac{\tau_\alpha}{2}; \tau_\alpha \right)}
\right),   \\[0.5em]
\Phi_1=2\left(m_\alpha-1\right)eK(m_\alpha)\frac{\vartheta_3\left(\frac{1+\tau_\alpha}{2}+\frac{\Delta^\alpha}{2\pi\mathrm{i}}; \tau_\alpha \right)}{\vartheta_3'\left(\frac{1+\tau_\alpha}{2}; \tau_\alpha \right)},        \\[0.5em]
\tilde{\Phi}_1=\frac{m_\alpha-1}{2\alpha}\left(
\frac{\vartheta_3''\left(\frac{1+\tau_\alpha}{2}; \tau_\alpha \right)\vartheta_3\left(\frac{1+\tau_\alpha}{2}+\frac{\Delta^\alpha}{2\pi\mathrm{i}}; \tau_\alpha \right)}{\vartheta_3'\left(\frac{1+\tau_\alpha}{2}; \tau_\alpha \right)^2}-\frac{\vartheta_3'\left(\frac{1+\tau_\alpha}{2}+\frac{\Delta^\alpha}{2\pi\mathrm{i}}; \tau_\alpha \right)}{\vartheta_3'\left(\frac{1+\tau_\alpha}{2}; \tau_\alpha \right)}
\right),                            \\[0.5em]
\mathcal{X}^\alpha=\left(1+\frac{1}{m_\alpha}\left(\frac{eE(m_\alpha)}{eK(m_\alpha)}-1\right)\right)\x
+\left(\frac{eE(m_\alpha)}{4\alpha^2 eK(m_\alpha)}+\frac{1}{8\eta_2^2}-\frac{1}{8\alpha^2}\right)t,   \\[0.5em]
\tilde{\mathcal{X}}^\alpha=\frac{\Omega^\alpha \phi^\alpha}{\pi \mathrm{i} \alpha}\left(eE(m_\alpha)-eK(m_\alpha)\right)+\frac{\alpha \eta_2}{\pi}\int_{\alpha}^{\eta_2}\frac{\log r(s)}{s^2\sqrt{\left(\eta_2^2-s^2\right)\left(s^2-\alpha^2\right)}}\mathrm{d}s,
\end{gather}
where
$m_\alpha=\alpha/\eta_2$,
$\tau_\alpha=\mathrm{i}\,eK\left(\sqrt{1-m_\alpha^2}\right)\left/\right.2eK\left(m_\alpha\right)$,
$\Omega^\alpha=-\pi \mathrm{i}\eta_2\left/\right.eK\left(m_\alpha\right)$,
\begin{gather}
\Delta^\alpha=\Omega^\alpha\left(\x+\frac{t}{4\alpha \eta_2}+\phi^\alpha \right),  \quad
\phi^\alpha=-\int_{\alpha}^{\eta_2}\frac{\log r(s)}{\sqrt{\left(s^2-\alpha^2\right)\left(\eta_2^2-s^2\right)}}\frac{\mathrm{d}s}{\pi}.
\end{gather}
Specially, in the case of $r=r_0$ with $\beta_0=0$, \eqref{large-middle} holds for $\xi\in\left(\xi_\mathrm{crit}, -\eta_2^{-2}\right)$.

\item[\rm{\textbullet}] For $\xi<\xi_\mathrm{crit}$  with $\beta_1, \beta_2, \beta_0>-1$,
\begin{gather}\label{large-left}
\begin{aligned}
&u(x, t)=\mathrm{e}^{\Delta_1}\left(\Psi_1\tilde{\Phi}_1+\Phi_1\tilde{\Psi}_1\right)\Theta_1
+\mathcal{O}\left(\frac{1}{t}\right),                                                                 \\
&x=\x+\left(\Psi_1\tilde{\Psi}_1-\Phi_1\tilde{\Phi}_1\right)\Theta_1+\tilde{\mathcal{X}}_1-\mathcal{X}_1
+\mathcal{O}\left(\frac{1}{t}\right),
\end{aligned}
\end{gather}
where
\begin{gather}
\Psi_1=\frac{\vartheta_3\left(\frac{\tau_1}{2}-\frac{\Delta_1}{2\pi \mathrm{i}}; \tau_1\right)}{\vartheta_3\left(\frac{\tau_1}{2}; \tau_1\right)}, \quad
\Theta_1=\frac{\vartheta_3^2\left(0; \tau_1\right)}{\vartheta_3^2\left(\frac{\Delta_1}{2\pi\mathrm{i}}; \tau_1\right)},
\\[0.5em]
\tilde{\Psi}_1=\frac{1}{4\eta_1 eK(m_1)}\left(
\frac{\vartheta_3\left(\frac{\tau_1}{2}-\frac{\Delta_1}{2\pi \mathrm{i}}; \tau_1\right)\vartheta_3'\left(\frac{\tau_1}{2}; \tau_1\right)}{\vartheta_3^2\left(\frac{\tau_1}{2}; \tau_1\right)}-\frac{\vartheta_3'\left(\frac{\tau_1}{2}-\frac{\Delta_1}{2\pi \mathrm{i}}; \tau_1\right)}{\vartheta_3\left(\frac{\tau_1}{2}; \tau_1\right)}
\right),   \\[0.5em]
\Phi_1=2\left(m_1-1\right)eK(m_1)\frac{\vartheta_3\left(\frac{1+\tau_1}{2}+\frac{\Delta_1}{2\pi\mathrm{i}}; \tau_1\right)}{\vartheta_3'\left(\frac{1+\tau_1}{2}; \tau_1\right)},        \\[0.5em]
\tilde{\Phi}_1=\frac{m_1-1}{2\eta_1}\left(
\frac{\vartheta_3''\left(\frac{1+\tau_1}{2}; \tau_1\right)\vartheta_3\left(\frac{1+\tau_1}{2}+\frac{\Delta_1}{2\pi\mathrm{i}}; \tau_1\right)}{\vartheta_3'\left(\frac{1+\tau_1}{2}; \tau_1\right)^2}-\frac{\vartheta_3'\left(\frac{1+\tau_1}{2}+\frac{\Delta_1}{2\pi\mathrm{i}}; \tau_1\right)}{\vartheta_3'\left(\frac{1+\tau_1}{2}; \tau_1\right)}
\right),                            \\[0.5em]
\mathcal{X}_1=\left(1+\frac{1}{m_1}\left(\frac{eE(m_1)}{eK(m_1)}-1\right)\right)\x
+\left(\frac{eE(m_1)}{4\eta_1^2 eK(m_1)}+\frac{1}{8\eta_2^2}-\frac{1}{8\eta_1^2}\right)t,
\end{gather}
where $\Delta_1$ is denoted as
\begin{gather}
\Delta_1=\Omega_1\left(\x+\frac{t}{4\eta_1\eta_2}+\phi_1\right).
\end{gather}

\end{enumerate}

\end{theorem}

\begin{remark} In the above-mentioned study, we focus on a single singularity $\eta_0$ associated with the generalized reflection coefficients $r_0$ and $r_c$, but it is feasible to extend the analysis to accommodate an arbitrary number \(n\) of singularities. For the first generalized reflection coefficient, we can express it as
\begin{gather} \label{general-0}
r_0 = (\lambda - \eta_1)^{\beta_1}(\eta_2 - \lambda)^{\beta_2}\left(\prod_{j=1}^n\left|\lambda - \eta_{0, j}\right|^{\beta_{0, j}}\right)\gamma(\lambda),
\end{gather}
where the conditions \(\eta_1 < \eta_{0, 1} < \eta_{0, 2} < \cdots < \eta_{0, n} < \eta_2\) and \(\beta_{0, j} \in (-1, 0) \cup (0, +\infty)\) hold true. For the second reflection coefficient, we consider
\begin{gather} \label{general-1}
r_c = (\lambda - \eta_1)^{\beta_1}(\eta_2 - \lambda)^{\beta_2}\left(\prod_{j=1}^n\chi_j(\lambda)\right)\gamma(\lambda),
\end{gather}
where \(\chi_j\) is defined as a piecewise function given by \(\chi_j(\lambda) = 1\) for \(\lambda \in [\eta_{0, j-1}, \eta_{0, j})\) and \(\chi_j(\lambda) = c_j^2\) for \(\lambda \in (\eta_{0, j}, \eta_{0, j+1}]\), with \(c_j \neq 0\). Here, we have \(\eta_1 = \eta_{0, 0} < \eta_{0, 1} < \eta_{0, 2} < \cdots < \eta_{0, n} < \eta_{0, n+1} = \eta_2\).

The long-time asymptotic behavior of SP soliton gases defined by these extended forms in Theorems 1 and can be determined using the methods outlined in this paper, with the added construction of local parametrices around each $\eta_{0, j}$. For the first case \eqref{general-0}, the local parametrix around $\eta_{0, j}$ can be derived using the second type of modified Bessel parametrix, while for the second case \eqref{general-1}, confluent hypergeometric parametrix are employed for the construction of these local solutions.

\end{remark}

\vspace{1em}
\noindent \textbf{Notations.}\,\,
We conclude this introduction by outlining the notational conventions used throughout the paper. Subscripts \( + \) and \( - \) represent non-tangential boundary values taken from the left and right, respectively, along a jump contour in the context of a Riemann-Hilbert problem. For brevity in expressing jump matrices, we introduce the notation \( \mathcal{L}^{\lambda_1}_{\lambda_2}[\lambda_0] \) and \( \mathcal{U}^{\lambda_1}_{\lambda_2}[\lambda_0] \), defined as follows:
\begin{gather}\label{LU}
\begin{array}{l}
\mathcal{L}^{\lambda_1}_{\lambda_2}[\lambda_0] = \mathrm{e}^{\lambda_1 \sigma_3} \lambda_2^{\sigma_3} \mathcal{L}[\lambda_0] \lambda_2^{-\sigma_3} \mathrm{e}^{-\lambda_1 \sigma_3}, \\[1em]
\mathcal{U}^{\lambda_1}_{\lambda_2}[\lambda_0] = \mathrm{e}^{\lambda_1 \sigma_3} \lambda_2^{\sigma_3} \mathcal{U}[\lambda_0] \lambda_2^{-\sigma_3} \mathrm{e}^{-\lambda_1 \sigma_3},
\end{array}
\end{gather}
where \( \mathcal{L}[\lambda_0] \) and \( \mathcal{U}[\lambda_0] \) are lower and upper triangular matrices, respectively, with diagonal entries all equal to 1:
\begin{gather}\label{LU0}
\mathcal{L}[\lambda_0] =
\begin{pmatrix}
1 & 0 \\
\lambda_0 & 1
\end{pmatrix}, \qquad
\mathcal{U}[\lambda_0] =
\begin{pmatrix}
1 & \lambda_0 \\
0 & 1
\end{pmatrix}.
\end{gather}
Two special cases of equation \eqref{LU} are worth noting. First, the notation \( \mathcal{L}^{\lambda_1} \) and \( \mathcal{U}^{\lambda_1} \) corresponds to \eqref{LU} with \( \lambda_2 = 1 \). Second, \( \mathcal{L}_{\lambda_2} \) and \( \mathcal{U}_{\lambda_2} \) represent the case where \( \lambda_1 = 0 \).
We also define the constant matrices \( C \), \( C_0 \), and \( C_1 \), which are used throughout the paper, as follows:
\begin{gather}
C = \frac{1}{\sqrt{2}}
\begin{pmatrix}
1 & \mathrm{i} \\
\mathrm{i} & 1
\end{pmatrix}, \qquad
C_0 = -\sqrt{2\pi}
\begin{pmatrix}
1 & 0 \\
0 & \mathrm{i}
\end{pmatrix}, \qquad
C_1 = \frac{1}{\sqrt{2}}
\begin{pmatrix}
1 & -1 \\
1 & 1
\end{pmatrix}.
\end{gather}
We denote by \( B(\lambda_0) \) a neighborhood centered at \( \lambda = \lambda_0 \) in the complex \( \lambda \)-plane, and by \( B^\zeta(\zeta_0) \) a neighborhood centered at \( \zeta = \zeta_0 \) in the \( \zeta \)-plane.
Finally, the symbols \( eK \) and \( eE \) are used to denote the complete elliptic integrals of the first and second kind, respectively, given by:
\begin{gather}\label{elliptic12}
eK(\lambda)=\int_0^{\pi/2}\frac{\mathrm{d}y}{\sqrt{1-\lambda^2\sin^2 y}}, \qquad  eE(\lambda)=\int_0^{\pi/2}\sqrt{1-\lambda^2\sin^2 y}\,\mathrm{d}y.
\end{gather}

\section{Some solvable Riemann-Hilbert models}

We here give some solvable Riemann-Hilbert models, which are used in the other sections.

\subsection{The Airy parametrix $M^{\mathrm{Ai}}(\zeta)$}

The matrix \( M^{\mathrm{Ai}}(\zeta) \) is constructed using the Airy function of the first kind, \( \mathrm{Ai}(\zeta) \), which satisfies Airy's ordinary differential equation:
\begin{gather}
\frac{\mathrm{d}^2y}{\mathrm{d}\zeta^2}-\zeta y=0.
\end{gather}
The Airy parametrix, originally introduced by Deift, Kriecherbauer, McLaughlin, Venakides, and Zhou to study the asymptotics of orthogonal polynomials with exponential weights \cite{43}, is constructed using the steepest descent method developed by Deift and Zhou \cite{45}. In this paper, we adapt the parametrix slightly, as shown in Figure \ref{BesselAiry}(Left), and formulate the matrix \( M^{\mathrm{Ai}}(\zeta) \) as follows:
\begin{gather}
M^{\mathrm{Ai}}\left(\zeta\right)=
\left\{
\begin{array}{ll}
\zeta_0^{-\sigma_3/4}C_0
\begin{pmatrix}
\mathrm{Ai}'\left(\zeta_0\zeta\right)  & -\mathrm{e}^{2\pi \mathrm{i}/3} \mathrm{Ai}'\left(\mathrm{e}^{-2\pi\mathrm{i}/3}\zeta_0\zeta\right) \\[0.5em]
\mathrm{Ai}\left(\zeta_0\zeta\right)  & -\mathrm{e}^{-2\pi\mathrm{i}/3} \mathrm{Ai}\left(\mathrm{e}^{-2\pi\mathrm{i}/3}\zeta_0\zeta\right)
\end{pmatrix}, & \zeta \in \mathrm{D}_1^\zeta;  \\[2em]
\zeta_0^{-\sigma_3/4}C_0
\begin{pmatrix}
-\mathrm{e}^{2\pi \mathrm{i}/3}\mathrm{Ai}'\left(\mathrm{e}^{-4\pi \mathrm{i}/3}\zeta_0\zeta\right)  & -\mathrm{e}^{-4\pi \mathrm{i}/3} \mathrm{Ai}'\left(\mathrm{e}^{2\pi \mathrm{i}/3}\zeta_0\zeta\right) \\[0.5em]
-\mathrm{e}^{-4\pi \mathrm{i}/3}\mathrm{Ai}\left(\mathrm{e}^{-4\pi \mathrm{i}/3}\zeta_0\zeta\right)  & -\mathrm{e}^{2\pi \mathrm{i}/3} \mathrm{Ai}\left(\mathrm{e}^{2\pi \mathrm{i}/3}\zeta_0\zeta\right)
\end{pmatrix},& \zeta \in \mathrm{D}_2^\zeta; \\[2em]
\zeta_0^{-\sigma_3/4}C_0
\begin{pmatrix}
-\mathrm{e}^{2\pi \mathrm{i}/3}\mathrm{Ai}'\left(\mathrm{e}^{4\pi \mathrm{i}/3}\zeta_0\zeta\right)  & \mathrm{e}^{4\pi \mathrm{i}/3} \mathrm{Ai}'\left(\mathrm{e}^{2\pi \mathrm{i}/3}\zeta_0\zeta\right) \\[0.5em]
-\mathrm{e}^{4\pi \mathrm{i}/3}\mathrm{Ai}\left(\mathrm{e}^{4\pi \mathrm{i}/3}\zeta_0\zeta\right)  & \mathrm{e}^{2\pi \mathrm{i}/3} \mathrm{Ai}\left(\mathrm{e}^{2\pi \mathrm{i}/3}\zeta_0\zeta\right)
\end{pmatrix}, & \zeta \in \mathrm{D}_3^\zeta; \\[2em]
\zeta_0^{-\sigma_3/4}C_0
\begin{pmatrix}
\mathrm{Ai}'\left(\zeta_0\zeta\right)  & \mathrm{e}^{4\pi \mathrm{i}/3} \mathrm{Ai}'\left(\mathrm{e}^{2\pi \mathrm{i}/3}\zeta_0\zeta\right) \\[0.5em]
\mathrm{Ai}\left(\zeta_0\zeta\right)  & \mathrm{e}^{2\pi \mathrm{i}/3} \mathrm{Ai}\left(\mathrm{e}^{2\pi \mathrm{i}/3}\zeta_0\zeta\right)
\end{pmatrix}, & \zeta \in \mathrm{D}_4^\zeta.
\end{array}\right.
\end{gather}
In all of the above expressions, the constant \( \zeta_0 = \left(2/3\right)^{-2/3} \). The matrix \( M^{\mathrm{Ai}}(\zeta) \) solves a Riemann-Hilbert problem with the following properties. The matrix \( M^{\mathrm{Ai}}(\zeta) \) is analytic for \( \zeta \in \mathbb{C} \setminus (\cup_{j=1}^4 \Sigma_j) \) and, as \( \zeta \to \infty \), it satisfies the asymptotic condition:
\begin{gather}
M^{\mathrm{Ai}}\left(\zeta\right)=\zeta^{\sigma_3/4}C^{-1}\left(\mathbb{I}_2 + \mathcal{O}\left(\frac{1}{\zeta^{3/2}}\right)\right)\mathrm{e}^{-\zeta^{3/2} \sigma_3}.
\end{gather}
Along the contour \( \Sigma_j^0 = \Sigma_j \setminus \{0\}, j=1, 2, 3, 4 \), the matrix \( M^{\mathrm{Ai}}(\zeta) \) admits continuous boundary values that satisfy the following jump conditions:
\begin{gather}
M^{\mathrm{Ai}}_+\left(\zeta\right)=M^{\mathrm{Ai}}_-\left(\zeta\right)
\begin{cases}
\mathcal{L}\left[1\right], &\mathrm{for}\,\,\,  \zeta \in \Sigma_1^0 \cup \Sigma_3^0, \\[0.5em]
\mathrm{i} \sigma_2, &\mathrm{for}\,\,\,\zeta \in \Sigma_2^0, \\[0.5em]
\mathcal{U}\left[1\right], &\mathrm{for}\,\,\,\zeta \in \Sigma_4^0,
\end{cases}
\end{gather}
Near the intersection point \( \zeta = 0 \), the local behavior of \( M^{\mathrm{Ai}}(\zeta) \) is characterized by:
\begin{gather}
M^{\mathrm{Ai}}\left(\zeta\right)=
\mathcal{O}
\begin{pmatrix}
1 & 1 \\
1 & 1
\end{pmatrix}, \quad \text{as}\,\,\, \zeta \to 0.
\end{gather}

\begin{figure}[!t]
\centering
\includegraphics[scale=0.26]{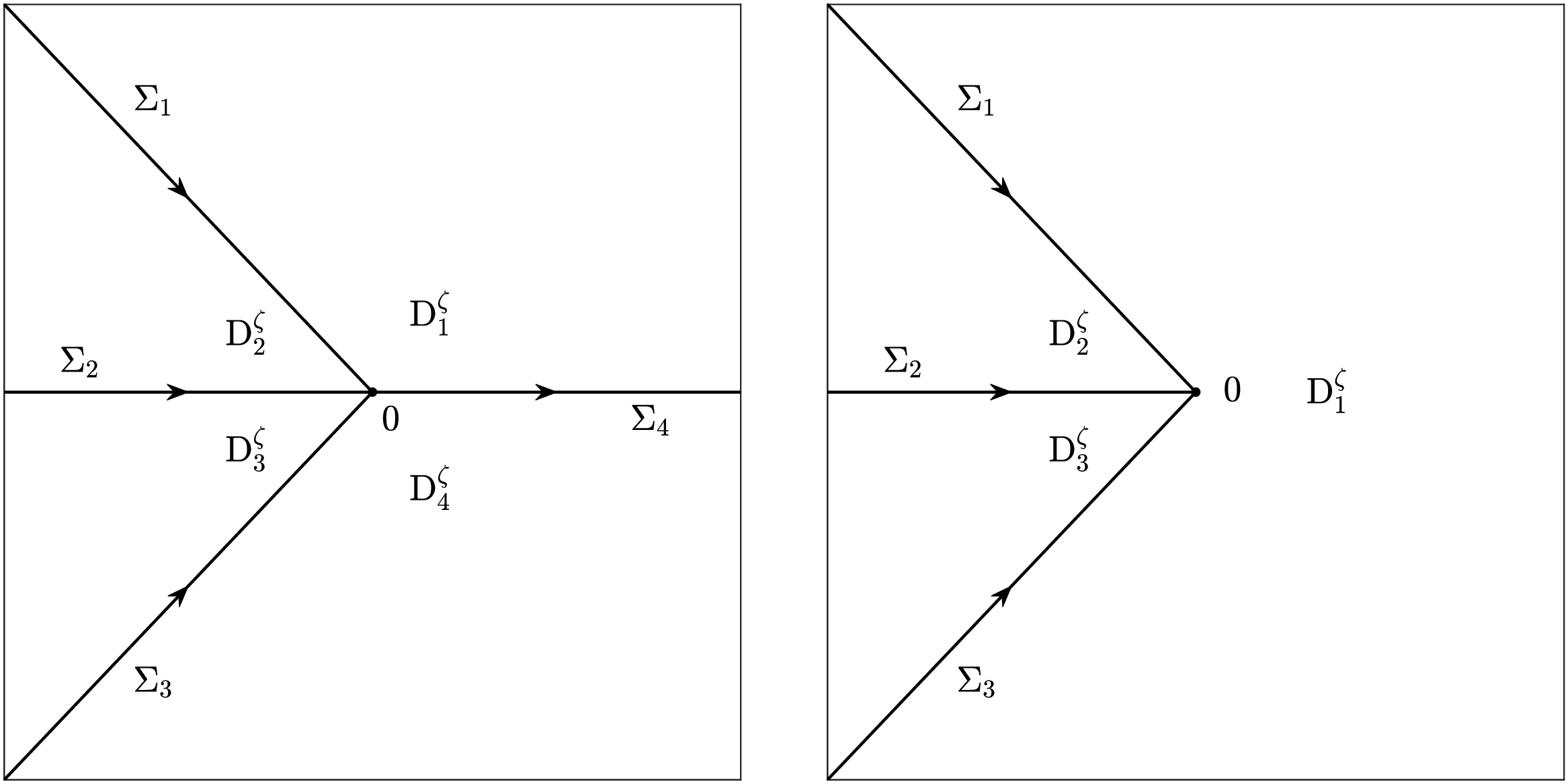}
\caption{Left: Jump contours for Airyparametrix $M^{\mathrm{mB}}$; Right: Jump contours for the first type of modified Bessel parametrix $M^{\mathrm{Ai}}$.}
\label{BesselAiry}
\end{figure}

\subsection{The first type of modified Bessel parametrix $M^{\mathrm{mB}}\left(\zeta; \beta\right)$}

The matrix \( M^{\mathrm{mB}}\left(\zeta; \beta\right) \) is constructed using the modified Bessel functions of the first and second kinds, \( I_{\beta}(\zeta) \) and \( K_{\beta}(\zeta) \), where the index \( \beta \in (-1, +\infty) \). These functions solve the modified Bessel's differential equation:
\begin{gather}
\zeta^2 \frac{\mathrm{d}^2 y}{\mathrm{d}\zeta^2} + \zeta \frac{\mathrm{d}y}{\mathrm{d}\zeta} - (\zeta^2 + \beta^2) y = 0.
\end{gather}
The function \( I_{\beta}(\zeta) \) is expressed as:
\[
I_{\beta}(\zeta) = \left(\zeta/2\right)^\beta \sum_{n=0}^\infty \frac{(\zeta/2)^{2n}}{n! \, \Gamma(\beta+n+1)},
\]
while the asymptotic behavior of \( K_{\beta}(\zeta) \) as \( \zeta \to \infty \) for \( \mathrm{arg}\zeta \in (-3\pi/2, 3\pi/2) \) is:
\[
K_{\beta}(\zeta) \sim \sqrt{\frac{\pi}{2\zeta}} \, \mathrm{e}^{-\zeta}.
\]
The modified Bessel parametrix was introduced by Kuijlaars, McLaughlin, Assche, and Vanlessen in their study of the strong asymptotics of orthogonal polynomials with modified Jacobi weights \cite{op2}. In this work, with a slight adjustment as shown in Figure \ref{BesselAiry}(Right), we define the matrix \( M^{\mathrm{mB}}(\zeta; \beta) \) as follows:
%For \( \lambda \in \mathrm{D}_1^\zeta \), the matrix \( M^{\mathrm{mB}}\left(\zeta; \beta\right) \) is given by:
\begin{gather} \label{modified-Bessel-1}
M^{\mathrm{mB}}\left(\zeta; \beta\right) = -\mathrm{i} \sqrt{\pi}
\begin{pmatrix}
\mathrm{i}\sqrt{\zeta} & 0 \\[0.5em]
0 & 1
\end{pmatrix}
\begin{pmatrix}
I'_{\beta}\left(\sqrt{\zeta}\right) & \mathrm{i} K'_{\beta}\left(\sqrt{\zeta}\right) / \pi \\[0.5em]
I_{\beta}\left(\sqrt{\zeta}\right) & \mathrm{i} K_{\beta}\left(\sqrt{\zeta}\right) / \pi
\end{pmatrix},\quad \zeta \in \mathrm{D}_1^\zeta;
\end{gather}
%For \( \lambda \in \mathrm{D}_2^\zeta \), it takes the following form:
\begin{gather}\label{modified-Bessel-2}
M^{\mathrm{mB}}\left(\zeta; \beta\right) = \frac{1}{\sqrt{\pi}}
\begin{pmatrix}
\mathrm{i}\sqrt{\zeta} & 0 \\[0.5em]
0 & 1
\end{pmatrix}
\begin{pmatrix}
K'_{\beta}\left(\sqrt{\zeta}\, \mathrm{e}^{-\pi \mathrm{i}}\right) & K'_{\beta}\left(\sqrt{\zeta}\right) \\[0.5em]
-K_{\beta}\left(\sqrt{\zeta}\, \mathrm{e}^{-\pi \mathrm{i}}\right) & K_{\beta}\left(\sqrt{\zeta}\right)
\end{pmatrix},\quad \zeta \in \mathrm{D}_2^\zeta;
\end{gather}
and
\begin{gather}\label{modified-Bessel-3}
M^{\mathrm{mB}}\left(\zeta; \beta\right) = \frac{1}{\sqrt{\pi}}
\begin{pmatrix}
\mathrm{i}\sqrt{\zeta} & 0 \\[0.5em]
0 & 1
\end{pmatrix}
\begin{pmatrix}
-K'_{\beta}\left(\sqrt{\zeta}\, \mathrm{e}^{\pi \mathrm{i}}\right) & K'_{\beta}\left(\sqrt{\zeta}\right) \\[0.5em]
K_{\beta}\left(\sqrt{\zeta}\, \mathrm{e}^{\pi \mathrm{i}}\right) & K_{\beta}\left(\sqrt{\zeta}\right)
\end{pmatrix},\quad \zeta \in \mathrm{D}_3^\zeta.
\end{gather}
The matrix \( M^{\mathrm{mB}}\left(\zeta; \beta\right) \) solves a Riemann-Hilbert problem with specific properties. It is analytic in \( \zeta \) for \( \zeta \in \mathbb{C} \setminus (\Sigma_1 \cup \Sigma_2 \cup \Sigma_3) \), and as \( \zeta \to \infty \), it satisfies the normalization condition:
\begin{gather}
M^{\mathrm{mB}}\left(\zeta; \beta\right) = \zeta^{\sigma_3/4} C^{-1} \left( \mathbb{I}_2 + \mathcal{O}\left(\frac{1}{\sqrt{\zeta}}\right) \right) \mathrm{e}^{\sqrt{\zeta} \, \sigma_3}.
\end{gather}
For \( \zeta \in \Sigma_1^0 \cup \Sigma_2^0 \cup \Sigma_3^0 \), the matrix \( M^{\mathrm{mB}}(\zeta; \beta) \) has continuous boundary values, with the following jump conditions:
\begin{gather}
M^{\mathrm{mB}}_+(\zeta; \beta) = M^{\mathrm{mB}}_-(\zeta; \beta)
\begin{cases}
\mathcal{L}\left[\mathrm{e}^{\beta \pi \mathrm{i}}\right], &\mathrm{for}\,\,\, \zeta \in \Sigma_1^0, \\[0.5em]
\mathrm{i} \sigma_2, &\mathrm{for}\,\,\, \zeta \in \Sigma_2^0, \\[0.5em]
\mathcal{L}\left[\mathrm{e}^{-\beta \pi \mathrm{i}}\right], &\mathrm{for}\,\,\, \zeta \in \Sigma_3^0,
\end{cases}
\end{gather}
where \( \Sigma_j^0 = \Sigma_j \setminus \{0\}, j=1, 2, 3 \).
The local behavior of \( M^{\mathrm{mB}}(\zeta; \beta) \) near the origin depends on the value of \( \beta \). When \( \beta > 0 \), the matrix behaves as:
\begin{gather}
M^{\mathrm{mB}}(\zeta; \beta) =
\begin{cases}
\mathcal{O}
\begin{pmatrix}
|\zeta|^{\beta/2} & |\zeta|^{-\beta/2}  \\[0.5em]
|\zeta|^{\beta/2} & |\zeta|^{-\beta/2}
\end{pmatrix}, &\mathrm{as}\,\,\, \zeta \in \mathrm{D}^\zeta_1 \to 0, \\[2em]
\mathcal{O}
\begin{pmatrix}
|\zeta|^{-\beta/2} & |\zeta|^{-\beta/2}  \\[0.5em]
|\zeta|^{-\beta/2} & |\zeta|^{-\beta/2}
\end{pmatrix}, &\mathrm{as}\,\,\, \zeta \in \mathrm{D}^\zeta_2 \cup \mathrm{D}^\zeta_3 \to 0.
\end{cases}
\end{gather}
For \( \beta = 0 \), the matrix exhibits logarithmic behavior near the origin:
\begin{gather}
M^{\mathrm{mB}}(\zeta; \beta) =
\begin{cases}
\mathcal{O}
\begin{pmatrix}
1 & \log|\zeta|  \\[0.5em]
1 & \log|\zeta|
\end{pmatrix}, &\mathrm{as}\,\,\, \zeta \in \mathrm{D}^\zeta_1 \to 0, \\[2em]
\mathcal{O}
\begin{pmatrix}
\log|\zeta| & \log|\zeta|  \\[0.5em]
\log|\zeta| & \log|\zeta|
\end{pmatrix}, &\mathrm{as}\,\,\, \zeta \in \mathrm{D}^\zeta_2 \cup \mathrm{D}^\zeta_3 \to 0.
\end{cases}
\end{gather}
Finally, for \( -1 < \beta < 0 \), the behavior is:
\begin{gather}
M^{\mathrm{mB}}(\zeta; \beta) = \mathcal{O}
\begin{pmatrix}
|\zeta|^{\beta/2} & |\zeta|^{\beta/2}  \\[0.5em]
|\zeta|^{\beta/2} & |\zeta|^{\beta/2}
\end{pmatrix}, \quad \mathrm{as}\,\,\,\zeta \in \mathrm{D}^\zeta_1 \cup \mathrm{D}^\zeta_2 \cup \mathrm{D}^\zeta_3 \to 0.
\end{gather}

\subsection{The second type of modified Bessel parametrix $M^{\mathrm{mb}}(\zeta; \beta)$}

The matrix \( M^{\mathrm{mb}}(\zeta; \beta) \) is constructed using modified Bessel functions of the first and second kinds, with indices \( (\beta \pm 1)/2 \) and \( \beta \in (-1, 0) \cup (0, +\infty) \). It is essential to note that \( M^{\mathrm{mb}} \) is distinct from \( M^{\mathrm{mB}} \), and the different notation is used to clearly differentiate the two. An earlier form of the modified Bessel parametrix, analogous to \( M^{\mathrm{mb}} \), was proposed in \cite{60} in the study of orthogonal polynomials with generalized Jacobi weights. In this paper, we adopt a modified version of this parametrix, as shown in Figure \ref{BesselCH}(Left), and define \( M^{\mathrm{mb}}(\zeta; \beta) \) as follows:
%For \( \lambda \in \mathrm{D}_1^\zeta \), the matrix \( M^{\mathrm{mb}}(\zeta; \beta) \) is given by:
\begin{gather}
M^{\mathrm{mb}}(\zeta; \beta) = C_1
\begin{pmatrix}
-G^+(\zeta) & G^+(\mathrm{e}^{-\pi\mathrm{i}} \zeta) \\[0.5em]
-G^-(\zeta) & G^-(\mathrm{e}^{-\pi\mathrm{i}} \zeta)
\end{pmatrix}
\mathrm{e}^{-\beta \pi \mathrm{i} \sigma_3 / 4},\quad \zeta \in \mathrm{D}_1^\zeta;
\end{gather}
%For \( \lambda \in \mathrm{D}_2^\zeta \), it takes the form:
\begin{gather}
M^{\mathrm{mb}}(\zeta; \beta) = C_1
\begin{pmatrix}
-H^+(\mathrm{e}^{-\pi\mathrm{i}} \zeta) & G^+(\mathrm{e}^{-\pi\mathrm{i}} \zeta) \\[0.5em]
-H^-(\mathrm{e}^{-\pi\mathrm{i}} \zeta) & -G^-(\mathrm{e}^{-\pi\mathrm{i}} \zeta)
\end{pmatrix}
\mathrm{e}^{-\beta \pi \mathrm{i} \sigma_3 / 4},\quad \zeta \in \mathrm{D}_2^\zeta;
\end{gather}
%For \( \lambda \in \mathrm{D}_3^\zeta \), it is expressed as:
\begin{gather}
M^{\mathrm{mb}}(\zeta; \beta) = C_1
\begin{pmatrix}
-H^+(\mathrm{e}^{-\pi\mathrm{i}} \zeta) & G^+(\mathrm{e}^{-\pi\mathrm{i}} \zeta) \\[0.5em]
-H^-(\mathrm{e}^{-\pi\mathrm{i}} \zeta) & -G^-(\mathrm{e}^{-\pi\mathrm{i}} \zeta)
\end{pmatrix}
\mathrm{e}^{\beta \pi \mathrm{i} \sigma_3 / 4},\quad \zeta \in \mathrm{D}_3^\zeta;
\end{gather}
%For \( \lambda \in \mathrm{D}_4^\zeta \), it is given by:
\begin{gather}
M^{\mathrm{mb}}(\zeta; \beta) = C_1
\begin{pmatrix}
-G^+(\mathrm{e}^{-2\pi\mathrm{i}} \zeta) & G^+(\mathrm{e}^{-\pi\mathrm{i}} \zeta) \\[0.5em]
-G^-(\mathrm{e}^{-2\pi\mathrm{i}} \zeta) & -G^-(\mathrm{e}^{-\pi\mathrm{i}} \zeta)
\end{pmatrix}
\mathrm{e}^{\beta \pi \mathrm{i} \sigma_3 / 4},\quad \zeta \in \mathrm{D}_4^\zeta;
\end{gather}
%For \( \lambda \in \mathrm{D}_5^\zeta \), the matrix is:
\begin{gather}
M^{\mathrm{mb}}(\zeta; \beta) = C_1
\begin{pmatrix}
G^+(\mathrm{e}^{\pi\mathrm{i}} \zeta) & G^+(\zeta) \\[0.5em]
-G^-(\mathrm{e}^{\pi\mathrm{i}} \zeta) & G^-(\zeta)
\end{pmatrix}
\mathrm{e}^{-\beta \pi \mathrm{i} \sigma_3 / 4},\quad \zeta \in \mathrm{D}_5^\zeta;
\end{gather}
%For \( \lambda \in \mathrm{D}_6^\zeta \), it takes the form:
\begin{gather}
M^{\mathrm{mb}}(\zeta; \beta) = C_1
\begin{pmatrix}
H^+(\zeta) & G^+(\zeta) \\[0.5em]
-H^-(\zeta) & G^-(\zeta)
\end{pmatrix}
\mathrm{e}^{-\beta \pi \mathrm{i} \sigma_3 / 4},\quad \zeta \in \mathrm{D}_6^\zeta;
\end{gather}
%For \( \lambda \in \mathrm{D}_7^\zeta \), the matrix is:
\begin{gather}
M^{\mathrm{mb}}(\zeta; \beta) = C_1
\begin{pmatrix}
H^+(\zeta) & G^+(\zeta) \\[0.5em]
-H^-(\zeta) & G^-(\zeta)
\end{pmatrix}
\mathrm{e}^{\beta \pi \mathrm{i} \sigma_3 / 4},\quad \zeta \in \mathrm{D}_7^\zeta;
\end{gather}
and %For \( \lambda \in \mathrm{D}_8^\zeta \), the matrix takes the form:
\begin{gather}
M^{\mathrm{mb}}(\zeta; \beta) = C_1
\begin{pmatrix}
G^+(\mathrm{e}^{-\pi\mathrm{i}} \zeta) & G^+(\zeta) \\[0.5em]
-G^-(\mathrm{e}^{-\pi\mathrm{i}} \zeta) & G^-(\zeta)
\end{pmatrix}
\mathrm{e}^{\beta \pi \mathrm{i} \sigma_3 / 4}, ,\quad \zeta \in \mathrm{D}_8^\zeta,
\end{gather}
where \( G^\pm(\zeta) = \sqrt{\zeta/\pi} \, K_{(\beta\pm1)/2}(\zeta) \) and \( H^\pm(\zeta) = \sqrt{\pi\zeta} \, I_{(\beta\pm1)/2}(\zeta) \), with the argument of \( \zeta \) constrained to \( (-\pi/2, 3\pi/2) \).
The matrix \( M^{\mathrm{mb}}(\zeta; \beta) \) solves a \( 2 \times 2 \) Riemann-Hilbert problem with the following properties. It is analytic in \( \zeta \) for \( \zeta \in \mathbb{C} \setminus (\cup_{j=1}^8 \Sigma_j) \) and satisfies the following asymptotic normalization conditions:
\begin{gather}
M^{\mathrm{mb}}\left(\zeta; \beta\right)=
\begin{cases}
\left(\mathbb{I}_2+\mathcal{O}\left(\zeta^{-1}\right)\right)\mathrm{i}\sigma_2\,\mathrm{e}^{-\beta\pi\mathrm{i}\sigma_3/4}\mathrm{e}^{-\zeta\sigma_3}, &\mathrm{as}\,\,\, \zeta\in\mathrm{D}^\zeta_1\cup\mathrm{D}^\zeta_2\to\infty, \\[0.5em]
\left(\mathbb{I}_2+\mathcal{O}\left(\zeta^{-1}\right)\right)\mathrm{i}\sigma_2\,\mathrm{e}^{\beta\pi\mathrm{i}\sigma_3/4}\mathrm{e}^{-\zeta\sigma_3}, & \mathrm{as}\,\,\, \zeta\in\mathrm{D}^\zeta_3\cup\mathrm{D}^\zeta_4\to\infty, \\[0.5em]
\left(\mathbb{I}_2+\mathcal{O}\left(\zeta^{-1}\right)\right)\mathrm{e}^{-\beta\pi\mathrm{i}\sigma_3/4}\mathrm{e}^{\zeta\sigma_3}, &\mathrm{as}\,\,\, \zeta\in\mathrm{D}^\zeta_5\cup\mathrm{D}^\zeta_6\to\infty, \\[0.5em]
\left(\mathbb{I}_2+\mathcal{O}\left(\zeta^{-1}\right)\right)\mathrm{e}^{\beta\pi\mathrm{i}\sigma_3/4}\mathrm{e}^{\zeta\sigma_3}, & \mathrm{as}\,\,\, \zeta\in\mathrm{D}^\zeta_7\cup\mathrm{D}^\zeta_8\to\infty.
\end{cases}
\end{gather}
For \( \zeta \in \cup_{j=1}^8 \Sigma_j^0 \), the matrix \( M^{\mathrm{mb}}(\zeta; \beta) \) has continuous boundary values, related through the following jump conditions:
\begin{gather}
M^{\mathrm{mb}}_+(\zeta; \beta) = M^{\mathrm{mb}}_-(\zeta; \beta)
\begin{cases}
\mathcal{L}\left[\mathrm{e}^{-\beta \pi \mathrm{i}}\right], &\mathrm{for}\,\,\, \zeta \in \Sigma_1^0 \cup \Sigma_5^0, \\[0.5em]
\mathrm{e}^{\beta \pi \mathrm{i} \sigma_3 / 2}, &\mathrm{for}\,\,\, \zeta \in \Sigma_2^0 \cup \Sigma_6^0, \\[0.5em]
\mathcal{L}\left[\mathrm{e}^{\beta \pi \mathrm{i}}\right], &\mathrm{for}\,\,\, \zeta \in \Sigma_3^0 \cup \Sigma_7^0, \\[0.5em]
\mathrm{i} \sigma_2, &\mathrm{for}\,\,\, \zeta \in \Sigma_4^0 \cup \Sigma_8^0,
\end{cases}
\end{gather}
where \( \Sigma_j^0 = \Sigma_j \setminus \{0\}, j = 1, 2, \ldots, 8 \).
Near the origin, for \( \beta > -1 \) and \( \beta \neq 0 \), the local behavior of \( M^{\mathrm{mb}}(\zeta; \beta) \) is as follows:
\begin{gather}
M^{\mathrm{mb}}\left(\zeta; \beta\right)=
\begin{cases}
\mathcal{O}
\begin{pmatrix}
\left|\zeta\right|^{\beta/2} & \left|\zeta\right|^{-\left|\beta\right|/2}  \\[0.5em]
\left|\zeta\right|^{\beta/2} & \left|\zeta\right|^{-\left|\beta\right|/2}
\end{pmatrix},
& \mathrm{as}\,\,\, \zeta\in \mathrm{D}^\zeta_2\cup\mathrm{D}^\zeta_3\cup\mathrm{D}^\zeta_6\cup\mathrm{D}^\zeta_7\to 0, \\[2em]
\mathcal{O}
\begin{pmatrix}
\left|\zeta\right|^{-\left|\beta\right|/2} & \left|\zeta\right|^{-\left|\beta\right|/2}  \\[0.5em]
\left|\zeta\right|^{-\left|\beta\right|/2} & \left|\zeta\right|^{-\left|\beta\right|/2}
\end{pmatrix},
& \mathrm{as}\,\,\, \zeta\in \mathrm{D}^\zeta_1\cup\mathrm{D}^\zeta_4\cup\mathrm{D}^\zeta_5\cup\mathrm{D}^\zeta_8\to 0.
\end{cases}
\end{gather}

\begin{figure}[!t]
\centering
\includegraphics[scale=0.28]{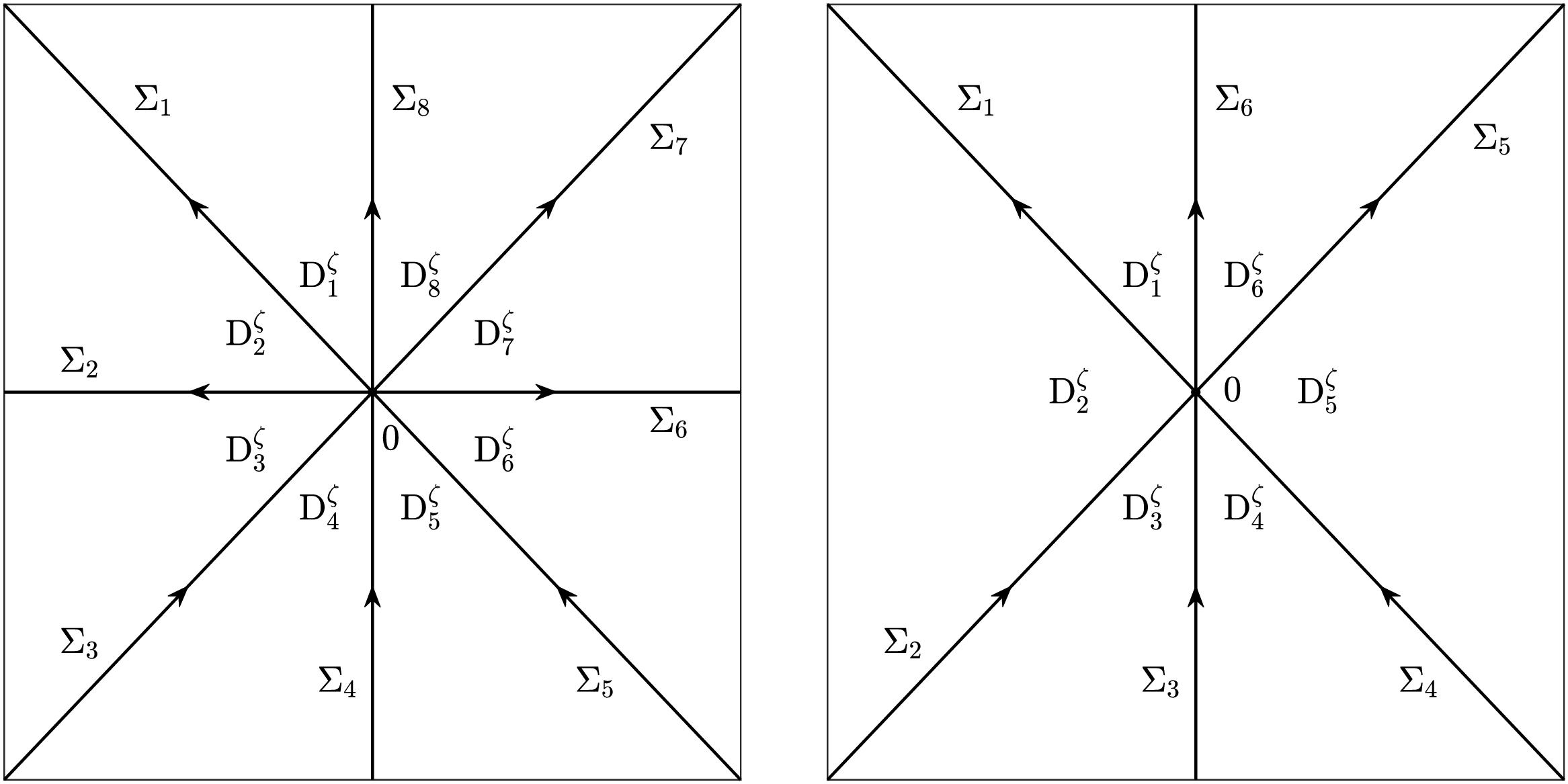}
\caption{Left: Jump contour for modified Bessel parametrix $M^{\mathrm{mb}}$; Right: Jump contour for Confluent Hypergeometric parametrix $M^{\mathrm{CH}}$.}
\label{BesselCH}
\end{figure}

\subsection{Confluent hypergeometric parametrix $M^{\mathrm{CH}}(\zeta; \kappa)$}

The confluent hypergeometric parametrix is built using the confluent hypergeometric functions \( M(\zeta; \kappa) \) and \( U(\zeta; \kappa) \), which are two standard solutions of Kummer’s equation:
\begin{gather}
\zeta \frac{\mathrm{d}^2 y}{\mathrm{d}\zeta^2} + (1 - \zeta) \frac{\mathrm{d}y}{\mathrm{d}\zeta} - \kappa y = 0,
\end{gather}
where \( M(\zeta; \kappa) \) is an entire function represented by the series:
\[
M(\zeta; \kappa) = \sum_{n=1}^\infty \frac{\Gamma(\kappa+n)}{\Gamma(\kappa) n!} \zeta^n,
\]
and \( U(\zeta; \kappa) \) has a branch point at \( \zeta = 0 \), with the following asymptotic behavior:
\[
U(\zeta; \kappa) \sim \zeta^{-\kappa},
\]
as \( \zeta \to \infty \) for \( \mathrm{arg} \, \zeta \in (-3\pi / 2, 3\pi / 2) \). The construction of the confluent hypergeometric parametrix was introduced in \cite{op1} for the analysis of asymptotic behavior of Hankel determinants and orthogonal polynomials with a jump in the Gaussian weight. In this paper, with slight modifications, we formulate the confluent hypergeometric parametrix \( M^{\mathrm{CH}}(\zeta; \kappa) \), where \( \kappa \in \mathrm{i} \mathbb{R} \), as depicted in Figure \ref{BesselCH}(Right). It is defined as follows:
%For \( \lambda \in \mathrm{D}_1^\zeta \), the matrix \( M^{\mathrm{CH}}(\zeta; \kappa) \) is given by:
\begin{gather}
M^{\mathrm{CH}}(\zeta; \kappa) =
\begin{pmatrix}
-\frac{\Gamma(1 + \kappa)}{\Gamma(-\kappa)} \mathrm{e}^{\kappa \pi \mathrm{i}} U(1 + \kappa, \zeta) & U(-\kappa, \mathrm{e}^{-\pi \mathrm{i}} \zeta) \\[0.5em]
-\mathrm{e}^{\kappa \pi \mathrm{i}} U(\kappa, \zeta) & \frac{\Gamma(1 - \kappa)}{\Gamma(\kappa)} U(1 - \kappa, \mathrm{e}^{-\pi \mathrm{i}} \zeta)
\end{pmatrix}
\mathrm{e}^{-\zeta \sigma_3 / 2},\quad \zeta \in \mathrm{D}_1^\zeta;
\end{gather}
%For \( \lambda \in \mathrm{D}_2^\zeta \), the matrix takes the form:
\begin{gather}
M^{\mathrm{CH}}(\zeta; \kappa) =
\begin{pmatrix}
\Gamma(1 + \kappa) M(-\kappa, \mathrm{e}^{-\pi \mathrm{i}} \zeta) & U(-\kappa, \mathrm{e}^{-\pi \mathrm{i}} \zeta) \\[0.5em]
-\Gamma(1 - \kappa) M(1 - \kappa, \mathrm{e}^{-\pi \mathrm{i}} \zeta) & \frac{\Gamma(1 - \kappa)}{\Gamma(\kappa)} U(1 - \kappa, \mathrm{e}^{-\pi \mathrm{i}} \zeta)
\end{pmatrix}
\mathrm{e}^{\zeta / 2},\quad \zeta \in \mathrm{D}_2^\zeta;
\end{gather}
%For \( \lambda \in \mathrm{D}_3^\zeta \), the matrix is expressed as:
\begin{gather}
M^{\mathrm{CH}}(\zeta; \kappa) =
\begin{pmatrix}
-\frac{\Gamma(1 + \kappa)}{\Gamma(-\kappa)} \mathrm{e}^{-\kappa \pi \mathrm{i}} U(1 + \kappa, \mathrm{e}^{-2\pi \mathrm{i}} \zeta) & U(-\kappa, \mathrm{e}^{-\pi \mathrm{i}} \zeta) \\[0.5em]
-\mathrm{e}^{-\kappa \pi \mathrm{i}} U(\kappa, \mathrm{e}^{-2\pi \mathrm{i}} \zeta) & \frac{\Gamma(1 - \kappa)}{\Gamma(\kappa)} U(1 - \kappa, \mathrm{e}^{-\pi \mathrm{i}} \zeta)
\end{pmatrix}
\mathrm{e}^{-\zeta \sigma_3 / 2},\quad \zeta \in \mathrm{D}_3^\zeta;
\end{gather}
%For \( \lambda \in \mathrm{D}_4^\zeta \), it is written as:
\begin{gather}
M^{\mathrm{CH}}(\zeta; \kappa) =
\begin{pmatrix}
\mathrm{e}^{-\kappa \pi \mathrm{i}} U(-\kappa, \mathrm{e}^{\pi \mathrm{i}} \zeta) & \frac{\Gamma(1 + \kappa)}{\Gamma(-\kappa)} U(1 + \kappa, \zeta) \\[0.5em]
\frac{\Gamma(1 - \kappa)}{\Gamma(\kappa)} \mathrm{e}^{-\kappa \pi \mathrm{i}} U(1 - \kappa, \mathrm{e}^{\pi \mathrm{i}} \zeta) & U(\kappa, \zeta)
\end{pmatrix}
\mathrm{e}^{\zeta \sigma_3 / 2},\quad \zeta \in \mathrm{D}_4^\zeta;
\end{gather}
%For \( \lambda \in \mathrm{D}_5^\zeta \), the matrix is:
\begin{gather}
M^{\mathrm{CH}}(\zeta; \kappa) =
\begin{pmatrix}
\Gamma(1 + \kappa) M(1 + \kappa, \zeta) & \frac{\Gamma(1 + \kappa)}{\Gamma(-\kappa)} U(1 + \kappa, \zeta) \\[0.5em]
-\Gamma(1 - \kappa) M(\kappa, \zeta) & U(\kappa, \zeta)
\end{pmatrix}
\mathrm{e}^{-\zeta / 2},\quad \zeta \in \mathrm{D}_5^\zeta;
\end{gather}
and %For \( \lambda \in \mathrm{D}_6^\zeta \), the matrix is expressed as:
\begin{gather}
M^{\mathrm{CH}}(\zeta; \kappa) =
\begin{pmatrix}
\mathrm{e}^{\kappa \pi \mathrm{i}} U(-\kappa, \mathrm{e}^{-\pi \mathrm{i}} \zeta) & \frac{\Gamma(1 + \kappa)}{\Gamma(-\kappa)} U(1 + \kappa, \zeta) \\[0.5em]
\frac{\Gamma(1 - \kappa)}{\Gamma(\kappa)} \mathrm{e}^{\kappa \pi \mathrm{i}} U(1 - \kappa, \mathrm{e}^{-\pi \mathrm{i}} \zeta) & U(\kappa, \zeta)
\end{pmatrix}
\mathrm{e}^{\zeta \sigma_3 / 2},\quad \zeta \in \mathrm{D}_6^\zeta.
\end{gather}
The matrix \( M^{\mathrm{CH}}(\zeta; \kappa) \) satisfies a Riemann-Hilbert problem with specific properties. It is analytic in \( \zeta \) for \( \zeta \in \mathbb{C} \setminus (\cup_{j=1}^6 \Sigma_j) \), and its asymptotic behavior is given by:
\begin{gather}
M^{\mathrm{CH}}(\zeta; \kappa) =
\begin{cases}
\left(\mathbb{I}_2 + \mathcal{O}(\zeta^{-1})\right) \zeta^{\kappa \sigma_3} \mathrm{e}^{\zeta \sigma_3 / 2}, & \mathrm{as}\,\,\,\zeta \in \mathrm{D}^\zeta_4 \cup \mathrm{D}^\zeta_5 \cup \mathrm{D}^\zeta_6 \to \infty, \\[0.5em]
\left(\mathbb{I}_2 + \mathcal{O}(\zeta^{-1})\right) \mathrm{i} \sigma_2 \mathrm{e}^{\kappa \pi \mathrm{i} \sigma_3} \zeta^{-\kappa \sigma_3} \mathrm{e}^{-\zeta \sigma_3}, &\mathrm{as}\,\,\, \zeta \in \mathrm{D}^\zeta_1 \cup \mathrm{D}^\zeta_2 \cup \mathrm{D}^\zeta_3 \to \infty.
\end{cases}
\end{gather}
For \( \zeta \in \cup_{j=1}^6 \Sigma_j^0 \), the matrix \( M^{\mathrm{CH}}(\zeta; \kappa) \) admits continuous boundary values, denoted by \( M^{\mathrm{CH}}_+(\zeta; \kappa) \) and \( M^{\mathrm{CH}}_-(\zeta; \kappa) \), which satisfy the following jump conditions:
\begin{gather}
M^{\mathrm{CH}}_+(\zeta; \kappa) = M^{\mathrm{CH}}_-(\zeta; \kappa)
\begin{cases}
\mathcal{L}\left[\mathrm{e}^{\kappa \pi \mathrm{i}}\right], &\mathrm{for}\,\,\, \zeta \in \Sigma_1^0 \cup \Sigma_5^0, \\[0.5em]
\mathcal{L}\left[\mathrm{e}^{-\kappa \pi \mathrm{i}}\right], &\mathrm{for}\,\,\, \zeta \in \Sigma_2^0 \cup \Sigma_4^0, \\[0.5em]
\mathrm{i} \sigma_2 \mathrm{e}^{-\kappa \pi \mathrm{i} \sigma_3}, &\mathrm{for}\,\,\, \zeta \in \Sigma_3^0, \\[0.5em]
\mathrm{i} \sigma_2 \mathrm{e}^{\kappa \pi \mathrm{i} \sigma_3}, &\mathrm{for}\,\,\, \zeta \in \Sigma_6^0,
\end{cases}
\end{gather}
where \( \Sigma_j^0 = \Sigma_j \setminus \{0\}, j=1, 2, \ldots, 6 \).
The local behavior of \( M^{\mathrm{CH}}(\zeta; \kappa) \) near the origin is as follows:
\begin{gather}
M^{\mathrm{CH}}(\zeta; \kappa) =
\begin{cases}
\mathcal{O}
\begin{pmatrix}
1 & \log|\zeta| \\[0.5em]
1 & \log|\zeta|
\end{pmatrix}, &\mathrm{as}\,\,\, \zeta \in \mathrm{D}^\zeta_2 \cup \mathrm{D}^\zeta_5 \to 0, \\[2em]
\mathcal{O}
\begin{pmatrix}
\log|\zeta| & \log|\zeta| \\[0.5em]
\log|\zeta| & \log|\zeta|
\end{pmatrix}, & \mathrm{as}\,\,\,\zeta \in \mathrm{D}^\zeta_1 \cup \mathrm{D}^\zeta_3 \cup \mathrm{D}^\zeta_4 \cup \mathrm{D}^\zeta_6 \to 0.
\end{cases}
\end{gather}

\section{A generalized $g$-function construction}

The so-called \( g \)-function technique was pioneered by Deift, Venakides, and Zhou during their analysis of the zero-dispersion limit of the KdV equation \cite{16}. One of the most prominent uses of the \( g \)-function is to achieve normalization at infinity in Riemann-Hilbert problems, as seen in the theory of orthogonal polynomials \cite{45}. In particular, this is accomplished by constructing the function from the equilibrium measure, a method originally applied in the study of strong asymptotics of orthogonal polynomials with exponential weights \cite{23,43}. Since its inception, the Deift-Zhou steepest descent method, which heavily relies on the \( g \)-function, has been employed to solve numerous asymptotic problems. For example, in the case of orthogonal polynomials with modified Jacobi weights, the \( g \)-function is obtained through a conformal mapping from \( \mathbb{C} \setminus [-1, 1] \) to the exterior of the unit circle \cite{op2}. Similarly, in the context of the large-time asymptotics of infinite-order rogue waves, the \( g \)-function is derived from the associated spectral curve \cite{30}.

In the case of soliton gas dynamics, for the large-\( \x \) asymptotics of the initial condition \( u(x, 0) \), it is necessary to construct a \( g_0 \)-function for \( \x < 0 \). This construction is analogous to that described in \cite{21}, where the \( g_0 \)-function, denoted \( g_0 = g_0(\lambda) \), is given by:
\begin{gather}\label{g0}
g_0 = \lambda - \int_{\eta_2}^{\lambda} \frac{\zeta^2 - \rho}{R_0(\zeta)} \, \mathrm{d}\zeta,
\end{gather}
where \( \rho = \eta_2^2 \left( 1 - \frac{eE(\eta_1 / \eta_2)}{eK(\eta_1 / \eta_2)} \right) \) and \( R_0(\zeta) \) is the branch of \( \sqrt{(\zeta^2 - \eta_2^2)(\zeta^2 - \eta_1^2)} \) with branch cuts along \( (\eta_1, \eta_2) \cup (-\eta_2, -\eta_1) \), behaving as \( R_0(\zeta) = \zeta^2 + \mathcal{O}(1) \) as \( \zeta \to \infty \). More details about this type of \( g \)-function can be found in works such as \cite{23,26,62,63}.

\begin{RH}\label{RH-g}
The \( g_0 \)-function defined in \eqref{g0} satisfies the  following RH problem:

 \begin{itemize}

 \item{} The \( g_0 \)-function is analytic in \( \lambda \) for \( \lambda \in \mathbb{C} \setminus [-\eta_2, \eta_2] \);

  \item{} It is normalized as \( g_0 = \mathcal{O}(\lambda^{-1}) \) at infinity;

  \item{} The \( g_0 \)-function satisfies the following jump conditions:
\begin{gather}
\begin{aligned}
g_{0+}+g_{0-}&=2\lambda,  && \mathrm{for}\,\,\, \lambda\in\left(\eta_1, \eta_2\right)\cup\left(-\eta_2, -\eta_1\right)\\
g_{0+}-g_{0-}&=\Omega_1,  &&\mathrm{for}\,\,\,\lambda\in\left(-\eta_1, \eta_1\right).
\end{aligned}
\end{gather}
\end{itemize}
\end{RH}

For the long-time asymptotics of the soliton gas \( u(x, t) \) of the SP equation, an essential step involves determining the sign of the real part of the phase \( \theta \):
\begin{gather}
\Re(\theta) = \frac{\Re(\lambda)}{4} \left( \xi + \frac{1}{|\lambda|^2} \right),
\end{gather}
with the corresponding sign charts shown in Figure \ref{Sign}. It is clear that when \( \xi > -\eta_2^{-2} \), a small-norm argument suffices. However, for \( \xi < -\eta_2^{-2} \), a \( g \)-function must be constructed to achieve exponential decay.

\begin{figure}[!t]
\centering
\includegraphics[scale=0.3]{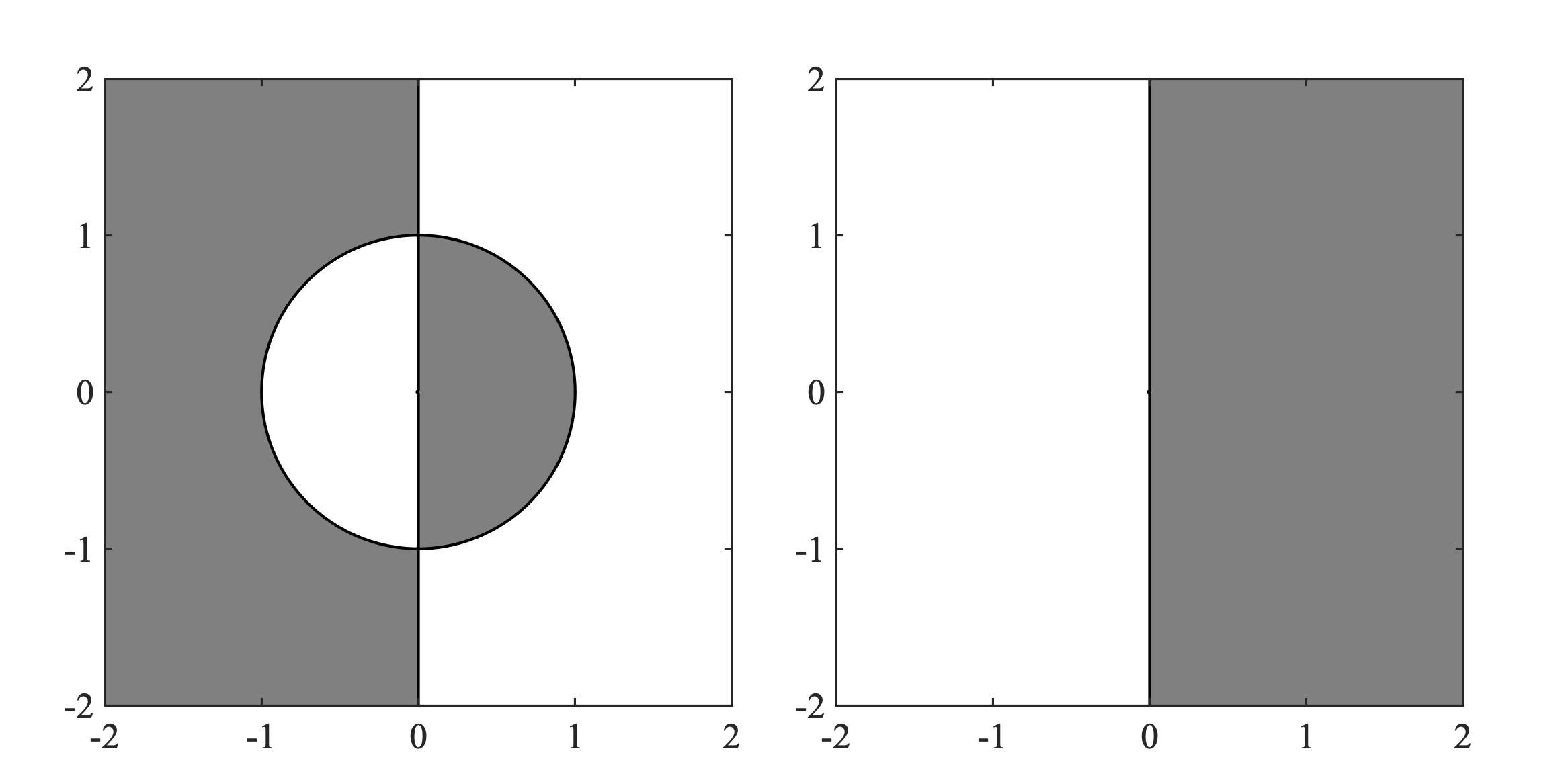}
\caption{Sign charts for \( \Re(\theta) \) with \( \xi = -1 \) (Left) and \( \xi = 0 \) (Right): \( \Re(\theta) > 0 \) in grey regions, and \( \Re(\theta) < 0 \) in white regions.}
\label{Sign}
\end{figure}

In this work, we construct a \( g \)-function \( g = g(\lambda; \xi) \) for the SP equation as follows:
\begin{gather}\label{g-function}
g = \theta - p,
\end{gather}
where the function \( p = p(\lambda; \xi) \) is defined as:
For \( \lambda \in \{ \lambda \mid \Re(\lambda) \ge 0 \} \setminus [0, \eta_2) \),
\begin{gather}\label{p-right}
p(\lambda; \xi) = \int_{\eta_2}^{\lambda} \frac{Q(y; \xi)}{4 y^2 R(y; \xi)} \, \mathrm{d}y,
\end{gather}
and for \( \lambda \in \{ \lambda \mid \Re(\lambda) \le 0 \} \setminus (-\eta_2, 0] \),
\begin{gather}
p(\lambda; \xi) = \int_{-\eta_2}^{\lambda} \frac{Q(y; \xi)}{4 y^2 R(y; \xi)} \, \mathrm{d}y,
\end{gather}
with the integration paths running from \( \pm \eta_2 \) to \( \lambda \). The functions \( R(y; \xi) \) and \( Q(y; \xi) \) depend on the parameter \( \xi \in (\xi_{\mathrm{crit}}, -\eta_2^{-2}) \cup (-\infty, \xi_{\mathrm{crit}}) \). Here, \( R(y; \xi) \) is the branch of \( \sqrt{(y^2 - \eta_2^2)(y^2 - \alpha^2)} \), with branch cuts on \( (\alpha, \eta_2) \cup (-\eta_2, -\alpha) \) and behavior \( R(y; \xi) = y^2 + \mathcal{O}(1) \) at infinity. The parameter \( \alpha \) is determined by the Whitham evolution equation \eqref{Whitham} for \( \xi \in (\xi_{\mathrm{crit}}, -\eta_2^{-2}) \), and \( \alpha = \eta_1 \) for \( \xi \in (-\infty, \xi_{\mathrm{crit}}) \).
For \( \xi \in (\xi_{\mathrm{crit}}, -\eta_2^{-2}) \), \( Q(y; \xi) \) is defined by:
\begin{gather}\label{Q1}
Q(y; \xi) = (y^2 - \alpha^2)\left( \xi y^2 - \frac{\eta_2}{\alpha} \right),
\end{gather}
and for \( \xi \in (-\infty, \xi_{\mathrm{crit}}) \), \( Q(y; \xi) \) is expressed as:
\begin{gather}\label{Q2}
Q(y; \xi) = \xi y^4 + \eta_2^2 \left( \xi + \frac{1}{\eta_1 \eta_2} \right) \left( \frac{eE(\eta_1 / \eta_2)}{eK(\eta_1 / \eta_2)} - 1 \right) y^2 + \eta_1 \eta_2.
\end{gather}

\begin{RH}\label{RH-g2}
The \( g \)-function defined in \eqref{g-function} satisfies the following RH problem:

\begin{itemize}

\item{} It is analytic for \( \lambda \in \mathbb{C} \setminus [-\eta_2, \eta_2] \);

 \item{} It normalizes to \( g = \mathcal{O}(\lambda^{-1}) \) as \( \lambda \to \infty \);

  \item{} For \( \lambda \in (-\eta_2, -\alpha) \cup (-\alpha, \alpha) \cup (\alpha, \eta_2) \), the \( g \)-function has continuous boundary values and satisfies the following jump conditions:
\begin{gather}
\begin{aligned}
g_{+}+g_{-}&=2\theta,  && \mathrm{for}\,\,\, \lambda\in\left(\alpha, \eta_2\right)\cup\left(-\eta_2, -\alpha\right),\\[0.5em]
g_{+}-g_{-}&=\frac{\Omega^\alpha}{4}\left(\xi+\frac{1}{\alpha \eta_2}\right),  &&\mathrm{for}\,\,\,\lambda\in\left(-\alpha, \alpha\right),
\end{aligned}
\end{gather}
where \( \Omega^\alpha = -\pi \mathrm{i} \eta_2 / eK(m_\alpha) \) for \( \xi \in (\xi_{\mathrm{crit}}, -\eta_2^{-2}) \), and \( \Omega^\alpha = \Omega_1 \) for \( \xi \in (-\infty, \xi_{\mathrm{crit}}) \). Near the points \( \pm \eta_2 \) and \( \pm \alpha \), the \( g \)-function exhibits the following local behavior:
\begin{gather}
g = \mathcal{O}(1), \quad \mathrm{as }\,\,\, \lambda \to \pm \eta_2, \pm \alpha.
\end{gather}

\end{itemize}
\end{RH}

\begin{proposition}\label{p-sign}
For $\xi\in\left(\xi_{\mathrm{crit}}, -\eta_2^{-2}\right)$, we have
\begin{gather}
\begin{aligned}
p_++p_->0, \quad &\mathrm{if}\,\,\, \lambda\in\left[\eta_1, \alpha \right);  \\[0.5em]
\Re\left(p_+\right)=0, \quad &\mathrm{if}\,\,\, \lambda\in\left[\alpha, \eta_2\right];  \\[0.5em]
\lim_{\Im(\lambda)\to 0^+}\frac{\partial \Re(p)}{\partial\Im(\lambda)}<0, \quad &\mathrm{if}\,\,\, \Re(\lambda)\in\left(\alpha, \eta_2\right).
\end{aligned}
\end{gather}
\end{proposition}
\begin{proof}
When \( \lambda \in [\eta_1, \alpha) \), we can express the sum of \( p_+ \) and \( p_- \) as:
\begin{gather*}
p_+ + p_- = \int_{\lambda}^{\alpha} \frac{Q(y; \xi)}{2 y^2 |R(y; \xi)|} \, \mathrm{d}y.
\end{gather*}
Given the definition of \( Q(y; \xi) \) in \eqref{Q1}, it follows that \( p_+ + p_- > 0 \).
For \( \lambda \in [\alpha, \eta_2] \), a straightforward calculation shows that:
\begin{gather*}
p_+ = \mathrm{i} \int_{\lambda}^{\eta_2} \frac{Q(y; \xi)}{4 y^2 |R_+(y; \xi)|} \, \mathrm{d}y \in \mathrm{i} \mathbb{R}.
\end{gather*}
Finally, for \( \Re(\lambda) \in (\alpha, \eta_2) \), from \eqref{p-right}, we obtain:
\begin{gather*}
\lim_{\Im(\lambda) \to 0^+} \frac{\partial \Re(p)}{\partial \Im(\lambda)} = \frac{Q(\Re(\lambda); \xi)}{4 \Re(\lambda)^2 |R_+(\Re(\lambda); \xi)|} < 0.
\end{gather*}
This completes the proof.
\end{proof}

\begin{proposition}\label{p-sign-1}
For $\xi\in\left(-\infty, \xi_{\mathrm{crit}}\right)$, we have
\begin{gather}
\begin{aligned}
\Re\left(p_+\right)=0, \quad &\mathrm{if}\,\,\, \lambda\in\left[\eta_1, \eta_2\right];  \\[0.5em]
\lim_{\Im(\lambda)\to 0^+}\frac{\partial \Re(p)}{\partial\Im(\lambda)}<0, \quad &\mathrm{if}\,\,\, \lambda\in\left(\eta_1, \eta_2\right).
\end{aligned}
\end{gather}
\end{proposition}
\begin{proof}
Using a similar approach as in Proposition \ref{p-sign}, we can show that \( p_+ \in \mathrm{i} \mathbb{R} \) for \( \lambda \in [\eta_1, \eta_2] \).
To demonstrate the second property, we examine the defining function \( Q \) from \eqref{Q2} and claim that it satisfies:
\begin{gather*}
Q(\lambda; \xi) < 0, \quad \text{for all } \lambda \in (\eta_1, \eta_2).
\end{gather*}
We observe that \( Q(0; \xi) = \eta_1 \eta_2 > 0 \) and note that \( Q(\lambda; \xi) \) is a quartic, even polynomial. Therefore, to establish \( Q(\lambda; \xi) < 0 \), we need to check that \( Q(\eta_1; \xi) < 0 \).
Differentiating \( Q(\eta_1; \xi) \) with respect to \( \xi \) gives:
\begin{gather*}
\frac{\partial Q(\eta_1; \xi)}{\partial \xi} = \eta_1^2 \eta_2^2 \left( \frac{eE(\eta_1 / \eta_2)}{eK(\eta_1 / \eta_2)} + \frac{\eta_1^2}{\eta_2^2} - 1 \right).
\end{gather*}
Recalling the inequality:
\begin{gather*}
1 - \zeta^2 < \frac{eE(\zeta)}{eK(\zeta)} < 1, \quad \text{for all } \zeta \in (0, 1),
\end{gather*}
we deduce that \( \frac{\partial Q(\eta_1; \xi)}{\partial \xi} > 0 \). Hence, it follows that \( Q(\eta_1; \xi) < 0 \), as \( Q(\eta_1; \xi_{\mathrm{crit}}) = 0 \).
This completes the proof.
\end{proof}

\section{Riemann-Hilbert problem for the short-pulse soliton gas}

The Riemann-Hilbert problem for the matrix \( M^\infty \) indicates that the associated jump matrices, as described in \eqref{Minfty}, consist of Cauchy integrals. To facilitate further analysis, we perform a contour deformation and introduce a new \( 2 \times 2 \) matrix-valued function, denoted as \( Y = Y(\lambda; \x, t) \), which is defined as follows:
\begin{gather}
Y(\lambda; \x, t)=\begin{cases}
M^\infty(\lambda; \x, t) \mathcal{L}^{t\theta}\left[\mathrm{i}\left(\mathcal{P}_1+\mathcal{P}_2\right)
\right], &\mathrm{for}\,\,\, \lambda\,\,\, \text{inside} \,\,\,\Gamma_+,  \\[0.5em]
M^\infty(\lambda; \x, t) \mathcal{U}^{t\theta}\left[\mathrm{i}\left(\mathcal{P}_{-1}+\mathcal{P}_{-2}\right)
\right], &\mathrm{for}\,\,\, \lambda \,\,\, \text{inside} \,\,\, \Gamma_-, \\[0.5em]
M^\infty(\lambda; \x, t), & \mathrm{for}\,\,\,\lambda\,\,\, \text{outside}\,\,\, \Gamma_+\cup\Gamma_-.
\end{cases}
\end{gather}

\begin{RH} \label{RH-7} $Y(\lambda; \x, t)$ satisfies the following RH problem:

\begin{itemize}

\item{} \( Y \) is analytic for \( \lambda \in \mathbb{C} \setminus \left( [-\eta_2, -\eta_1] \cup [\eta_1, \eta_2] \right) \);

\item{} $Y$ normalizes to the identity matrix \( \mathbb{I}_2 \) as \( \lambda \to \infty \);

\item{} When \( \lambda \in (\eta_1, \eta_0) \cup (\eta_0, \eta_2) \cup (-\eta_2, -\eta_0) \cup (-\eta_0, -\eta_1) \), the matrix \( Y \) admits continuous boundary values, denoted by \( Y_+ \) and \( Y_- \). These boundary values are related through the following jump conditions, derived using the Sokhotski-Plemelj formula:
\begin{gather}\label{Y-Jump}
Y_+(\lambda; \x, t)=Y_-(\lambda; \x, t)
\begin{cases}
\mathcal{L}^{t\theta}\left[\mathrm{i}r\right], &\mathrm{for}\,\,\, \lambda\in\left(\eta_1, \eta_0\right)\cup\left(\eta_0, \eta_2\right),\\[0.5em]
\mathcal{U}^{t\theta}\left[\mathrm{i}r\right],  &\mathrm{for}\,\,\, \lambda\in\left(-\eta_2, -\eta_0\right)\cup\left(-\eta_0, -\eta_1\right),
\end{cases}
\end{gather}
where the values of \( r \) on the interval \( (-\eta_2, -\eta_0) \cup (-\eta_0, -\eta_1) \) are determined by the symmetry \( r(-\lambda) = r(\lambda) \).
\end{itemize}
\end{RH}
For the case of the first type of generalized reflection coefficient, \( r = r_0 \), when \( \beta_0 = 0 \), the matrix \( Y \) has continuous boundary values on \( (\eta_1, \eta_2) \cup (-\eta_2, -\eta_1) \), which satisfy:
\begin{gather}\label{Y-Jump-Special}
Y_+(\lambda; \x, t)=Y_-(\lambda; \x, t)
\begin{cases}
\mathcal{L}^{t\theta}\left[\mathrm{i}r\right], &\mathrm{for}\,\,\, \lambda\in\left(\eta_1, \eta_2\right),\\[0.5em]
\mathcal{U}^{t\theta}\left[\mathrm{i}r\right],  &\mathrm{for}\,\,\, \lambda\in\left(-\eta_2, -\eta_1\right),
\end{cases}
\end{gather}
However, at the endpoints \( \pm \eta_1 \) and \( \pm \eta_2 \), the boundary values \( Y_\pm \) are not well-defined due to the singularities present in the reflection coefficients \( r_0 \) and \( r_c \) at \( \lambda = \pm \eta_0 \). As a result, the limits of \( Y \) do not exist at these points. To establish the uniqueness of the solution for \( Y \), it is essential to account for the local behavior of \( Y \) near both the endpoints and the singularities at \( \pm \eta_0 \).

Near each endpoint \( \eta_j \) (for \( j = 1, 2 \)), as \( \lambda \to \eta_j \), the local behavior of \( Y \) is given by:
\begin{gather}\label{local-Y-eta-1}
Y(\lambda; \x, t)=
\begin{cases}
\mathcal{O}
\begin{pmatrix}
1 &1  \\
1 &1
\end{pmatrix}
, & \mathrm{if} \,\,\, \beta_j\in\left(0, +\infty\right),  \\[1.5em]
\mathcal{O}
\begin{pmatrix}
\log \left|\lambda-\eta_j\right| &1  \\
\log \left|\lambda-\eta_j\right|  &1
\end{pmatrix}
, & \mathrm{if} \,\,\, \beta_j=0,  \\[1.5em]
\mathcal{O}
\begin{pmatrix}
\left|\lambda-\eta_j\right|^{\beta_j} &1\\
\left|\lambda-\eta_j\right|^{\beta_j} &1
\end{pmatrix}
,  & \mathrm{if} \,\,\, \beta_j\in\left(-1, 0\right),
\end{cases}
\end{gather}
with the \( \mathcal{O} \)-terms understood element-wise. Similarly, as \( \lambda \to -\eta_j \), the local behavior of \( Y \) near the endpoints \( -\eta_j \) (for \( j = 1, 2 \)) is:
\begin{gather}\label{local-Y-eta-2}
Y(\lambda; \x, t)=
\begin{cases}
\mathcal{O}
\begin{pmatrix}
1 &1  \\
1 &1
\end{pmatrix}
, & \mathrm{if} \,\,\, \beta_j\in\left(0, +\infty\right),  \\[1.5em]
\mathcal{O}
\begin{pmatrix}
1 & \log\left|\lambda+\eta_j\right|  \\
1 & \log\left|\lambda+\eta_j\right|
\end{pmatrix}
, & \mathrm{if} \,\,\, \beta_j=0,  \\[1.5em]
\mathcal{O}
\begin{pmatrix}
1& \left|\lambda+\eta_j\right|^{\beta_j}\\
1& \left|\lambda+\eta_j\right|^{\beta_j}
\end{pmatrix}
,  & \mathrm{if} \,\,\, \beta_j\in\left(-1, 0\right).
\end{cases}
\end{gather}

The local behavior of \( Y \) near the endpoints \( \pm\eta_j \) for \( j = 1, 2 \) remains consistent for both reflection coefficients \( r_0 \) and \( r_c \). However, near the singularities \( \pm\eta_0 \), the behavior of \( Y \) differs between the two cases.
For the first type of generalized reflection coefficient \( r_0 \), we focus on the case where \( \beta_0 \neq 0 \) due to the jump conditions in \eqref{Y-Jump-Special}. The local behavior of \( Y \) near \( \pm\eta_0 \) depends on the value of \( \beta_0 \):
When \( \beta_0 > 0 \),
\begin{gather}
Y=\mathcal{O}
\begin{pmatrix}
1 &1  \\
1 & 1
\end{pmatrix}, \quad
\mathrm{as}\,\,\, \lambda\to\pm\eta_0;
\end{gather}
When \( \beta_0 \in (-1, 0) \),
\begin{gather}\label{local-Y-eta-01}
Y(\lambda; \x, t)=
\begin{cases}
\mathcal{O}
\begin{pmatrix}
\left|\lambda-\eta_0\right|^{\beta_0} &1 \\
\left|\lambda-\eta_0\right|^{\beta_0} &1
\end{pmatrix}
, & \text{as} \,\,\, \lambda\to\eta_0,  \\[1.5em]
\mathcal{O}
\begin{pmatrix}
1& \left|\lambda+\eta_0\right|^{\beta_0} \\
1& \left|\lambda+\eta_0\right|^{\beta_0}
\end{pmatrix}
,  & \text{as} \,\,\, \lambda\to-\eta_0.
\end{cases}
\end{gather}
For the second type of reflection coefficient \( r_c \), the local behavior of \( Y \) near \( \pm\eta_0 \) is as follows:
\begin{gather}\label{local-Y-eta-02}
Y(\lambda; \x, t)=
\begin{cases}
\mathcal{O}
\begin{pmatrix}
\log\left|\lambda-\eta_0\right| &1\\
\log\left|\lambda-\eta_0\right| &1
\end{pmatrix}
, & \mathrm{as} \,\,\,\lambda\to\eta_0,\\[1.5em]
\mathcal{O}
\begin{pmatrix}
1 & \log\left|\lambda+\eta_0\right| \\
1 & \log\left|\lambda+\eta_0\right|
\end{pmatrix}
, & \mathrm{as} \,\,\,\lambda\to-\eta_0.
\end{cases}
\end{gather}
The solution for the soliton gas, denoted by \( u = u(x, t) \), can be reconstructed from the matrix \( Y(\lambda; \x, t) \) through the following relations:
\begin{gather}\label{potential-formula}
\begin{array}{rl}
u(x, t)=&\lim_{\lambda\to 0}\lambda^{-1}\left(Y(0; \x, t)^{-1}Y(\lambda; \x, t)\right)_{1, 2},\\[0.5em]
x=&\x+\lim_{\lambda\to 0}\lambda^{-1}\left(\left(Y(0; \x, t)^{-1}Y(\lambda; \x, t)\right)_{1, 1}-1\right).
\end{array}
\end{gather}
At \( t = 0 \), the matrix \( Y(\lambda; \x, 0) \) is directly related to the initial condition \( u(x, 0) \) for the soliton gas. The jump conditions for \( Y \) in this case are expressed as:
\begin{gather}\label{Y-Jump}
Y_+(\lambda; \x, 0)=Y_-(\lambda; \x, 0)
\begin{cases}
\mathcal{L}^{\lambda \x}\left[\mathrm{i}r\right], &\mathrm{for}\,\,\, \lambda\in\left(\eta_1, \eta_0\right)\cup\left(\eta_0, \eta_2\right),\\[0.5em]
\mathcal{U}^{\lambda \x}\left[\mathrm{i}r\right],  &\mathrm{for}\,\,\, \lambda\in\left(-\eta_2, -\eta_0\right)\cup\left(-\eta_0, -\eta_1\right),
\end{cases}
\end{gather}
where \( r = r_0 \) or \( r_c \) with \( \beta_0 \neq 0 \).

For the first type of reflection coefficient \( r_0 \) in the special case where \( \beta_0 = 0 \), the jump conditions simplify to:
\begin{gather}\label{Y-Jump-Special-0}
Y_+(\lambda; \x, 0)=Y_-(\lambda; \x, 0)
\begin{cases}
\mathcal{L}^{\lambda \x}\left[\mathrm{i}r\right], &\mathrm{for}\,\,\, \lambda\in\left(\eta_1, \eta_2\right),\\[0.5em]
\mathcal{U}^{\lambda \x}\left[\mathrm{i}r\right],  &\mathrm{for}\,\,\, \lambda\in\left(-\eta_2, -\eta_1\right),
\end{cases}
\end{gather}

\section{Large-$\x$ asymptotic for the initial value $u(x, 0)$ of the soliton gas}
For \( \x \to +\infty \), the standard small-norm method leads directly to \eqref{large-x-right}. In this section, we derive the asymptotic behavior for the regime where \( \x \to -\infty \).
Several transformations are performed on \( Y(\lambda; \x, 0) \), applying the sequence \( Y \to T \to S \to E \), in order to ensure that \( E(\lambda; \x, 0) \) normalizes to the identity matrix \( \mathbb{I}_2 \) as \( \lambda \to \infty \), and that the jump matrices decay uniformly and exponentially to \( \mathbb{I}_2 \). Following the Deift-Zhou steepest descent method, the next step involves triangular decomposition, facilitating contour deformation:
\begin{gather}
\begin{aligned}
\mathcal{L}^{\lambda \x}\left[\mathrm{i}r\right] &= \mathcal{U}^{\lambda \x}\left[-\mathrm{i}r^{-1}\right] \left(\mathcal{L}^{\lambda \x}\left[\mathrm{i}r\right] + \mathcal{U}^{\lambda \x}\left[\mathrm{i}r^{-1}\right] - 2\, \mathbb{I}_2\right) \mathcal{U}^{\lambda \x}\left[-\mathrm{i}r^{-1}\right], \\
\mathcal{U}^{\lambda \x}\left[\mathrm{i}r\right] &= \mathcal{L}^{\lambda \x}\left[-\mathrm{i}r^{-1}\right] \left(\mathcal{L}^{\lambda \x}\left[\mathrm{i}r^{-1}\right] + \mathcal{U}^{\lambda \x}\left[\mathrm{i}r\right] - 2\, \mathbb{I}_2\right) \mathcal{L}^{\lambda \x}\left[-\mathrm{i}r^{-1}\right].
\end{aligned}
\end{gather}
However, based on the sign of \( \Re(\lambda) \), exponential decay fails on the corresponding lenses. To address this, we apply a conjugation step before contour deformation by introducing an appropriate \( g_0 \)-function. In addition to the \( g_0 \)-function, a scalar \( f_0 \)-function is introduced to yield constant jump matrices.

\subsection{Riemann-Hilbert problem for $T(\lambda; \x, 0)$}
The \( f_0 \)-function is analytic in \( \lambda \) for \( \lambda \in \mathbb{C} \setminus [-\eta_2, \eta_2] \) and is normalized to 1 as \( \lambda \to \infty \). The continuous boundary values \( f_{0\pm} \) are related by the following conditions:
\begin{gather}
\begin{aligned}
&f_{0+}f_{0-} =r, &&\mathrm{for} \,\,\,\lambda\in\left(\eta_1, \eta_0\right)\cup\left(\eta_0, \eta_2\right);    \\[0.5em]
&f_{0+}f_{0-} =r^{-1},   && \mathrm{for} \,\,\,\lambda\in\left(-\eta_2, -\eta_0\right)\cup\left(-\eta_0, -\eta_1\right); \\[0.5em]
&f_{0+}^{-1}f_{0-}=\mathrm{e}^{\Omega_1\phi_1},  && \mathrm{for} \,\,\,\lambda\in\left(-\eta_1, \eta_1\right).
\end{aligned}
\end{gather}
For the first type of generalized reflection coefficient, \( r = r_0 \), in the case of \( \beta_0 = 0 \), the jump condition changes slightly to:
\begin{gather}
\begin{aligned}
&f_{0+}f_{0-} =r, &&\mathrm{for} \,\,\,\lambda\in\left(\eta_1, \eta_2\right);    \\[0.5em]
&f_{0+}f_{0-} =r^{-1},   && \mathrm{for} \,\,\,\lambda\in\left(-\eta_2,  -\eta_1\right); \\[0.5em]
&f_{0+}^{-1}f_{0-}=\mathrm{e}^{\Omega_1\phi_1},  && \mathrm{for} \,\,\,\lambda\in\left(-\eta_1, \eta_1\right).
\end{aligned}
\end{gather}
Using Plemelj's formula, the \( f_0 \)-function is derived by taking the logarithm, leading to:
\begin{gather}
f_0 = \exp\left\{\frac{R_0}{\pi \mathrm{i}} \left(\int_{\eta_1}^{\eta_2} \frac{\log r(s)}{R_{0+}(s)} \frac{\lambda}{s^2 - \lambda^2} \mathrm{d}s - \int_{0}^{\eta_1} \frac{\Omega_1 \phi_1}{R_0(s)} \frac{\lambda}{s^2 - \lambda^2} \mathrm{d}s \right)\right\}.
\end{gather}

With the \( g_0 \)- and \( f_0 \)-functions defined, the following conjugation is applied:
\begin{gather}\label{Conjugation-0}
T(\lambda; \x, 0) = Y(\lambda; \x, 0) \mathrm{e}^{\x g_0 \sigma_3} f_0^{-\sigma_3},
\end{gather}

\begin{RH}
This determines a \( 2 \times 2 \) matrix-valued function \( T(\lambda; \x, 0) \), which satisfies the following Riemann-Hilbert problem:

\begin{itemize}
\item{} \( T \) is analytic in \( \lambda \) for \( \lambda \in \mathbb{C} \setminus [-\eta_2, \eta_2] \);

 \item{} It normalizes to \( \mathbb{I}_2 \) as \( \lambda \to \infty \);

  \item{} For \( \lambda \in (-\eta_2, \eta_2) \setminus \{ \pm\eta_1, \pm\eta_0 \} \), the continuous boundary values of \( T \) are related by the jump conditions:
\begin{gather}
T_+(\lambda; \x, 0)=T_-(\lambda; \x, 0)\!\!\!
\begin{cases}
\mathcal{U}^{\x p_{0-}}_{f_{0-}}\left[-\mathrm{i}r^{-1}\right]\left(\mathrm{i}\sigma_1\right)\mathcal{U}^{\x p_{0+}}_{f_{0+}}\left[-\mathrm{i}r^{-1}\right], &\mathrm{for}\, \lambda\in\left(\eta_1, \eta_0\right)\cup\left(\eta_0, \eta_2\right),\\[0.5em]
\mathcal{L}^{\x p_{0-}}_{f_{0-}}\left[-\mathrm{i}r^{-1}\right]\left(\mathrm{i}\sigma_1\right)\mathcal{L}^{\x p_{0+}}_{f_{0+}}\left[-\mathrm{i}r^{-1}\right], &\mathrm{for}\, \lambda\in\left(-\eta_2, -\eta_0\right)\cup\left(-\eta_0, -\eta_1\right),\\[0.5em]
\mathrm{e}^{\Delta_1^0\sigma_3}, &\mathrm{for}\, \lambda\in\left(-\eta_1, \eta_1\right),
\end{cases}
\end{gather}
Here, \( p_0 \) is defined by \( p_0 = \lambda - g_0 \). For the reflection coefficient \( r_0 \) with \( \beta_0 = 0 \), the jump conditions slightly differ:
\begin{gather}
T_+(\lambda; \x, 0)=T_-(\lambda; \x, 0)
\begin{cases}
\mathcal{U}^{\x p_{0-}}_{f_{0-}}\left[-\mathrm{i}r^{-1}\right]\left(\mathrm{i}\sigma_1\right)\mathcal{U}^{\x p_{0+}}_{f_{0+}}\left[-\mathrm{i}r^{-1}\right], &\mathrm{for}\,\,\, \lambda\in\left(\eta_1, \eta_2\right),\\[0.5em]
\mathcal{L}^{\x p_{0-}}_{f_{0-}}\left[-\mathrm{i}r^{-1}\right]\left(\mathrm{i}\sigma_1\right)\mathcal{L}^{\x p_{0+}}_{f_{0+}}\left[-\mathrm{i}r^{-1}\right], &\mathrm{for}\,\,\, \lambda\in\left(-\eta_2, -\eta_1\right),\\[0.5em]
\mathrm{e}^{\Delta_1^0\sigma_3}, &\mathrm{for}\,\,\, \lambda\in\left(-\eta_1, \eta_1\right).
\end{cases}
\end{gather}
\end{itemize}
\end{RH}

Near each endpoint \( \eta_j \) for \( j = 1, 2 \), the matrix \( T(\lambda; \x, 0) \) behaves as follows:
\begin{gather}
T(\lambda; \x, 0)=
\begin{cases}
\mathcal{O}
\begin{pmatrix}
\left|\lambda-\eta_j\right|^{-\beta_j/2} &\left|\lambda-\eta_j\right|^{\beta_j/2}  \\[0.5em]
\left|\lambda-\eta_j\right|^{-\beta_j/2} &\left|\lambda-\eta_j\right|^{\beta_j/2}
\end{pmatrix}
, & \mathrm{if} \,\,\, \beta_j\in\left(0, +\infty\right),  \\[1.5em]
\mathcal{O}
\begin{pmatrix}
\log \left|\lambda-\eta_j\right| &1  \\[0.5em]
\log \left|\lambda-\eta_j\right|  &1
\end{pmatrix}
, & \mathrm{if} \,\,\, \beta_j=0,  \\[1.5em]
\mathcal{O}
\begin{pmatrix}
\left|\lambda-\eta_j\right|^{\beta_j/2} &\left|\lambda-\eta_j\right|^{\beta_j/2}\\[0.5em]
\left|\lambda-\eta_j\right|^{\beta_j/2} &\left|\lambda-\eta_j\right|^{\beta_j/2}
\end{pmatrix}
,  & \mathrm{if} \,\,\, \beta_j\in\left(-1, 0\right),
\end{cases}
\end{gather}
Similarly, near \( -\eta_j \), for \( j = 1, 2 \), the local behavior is:
\begin{gather}
T(\lambda; \x, 0)=
\begin{cases}
\mathcal{O}
\begin{pmatrix}
\left|\lambda+\eta_j\right|^{\beta_j/2} &\left|\lambda+\eta_j\right|^{-\beta_j/2}  \\[0.5em]
\left|\lambda+\eta_j\right|^{\beta_j/2} &\left|\lambda+\eta_j\right|^{-\beta_j/2}
\end{pmatrix}
, & \mathrm{if} \,\,\, \beta_j\in\left(0, +\infty\right),  \\[1.5em]
\mathcal{O}
\begin{pmatrix}
1& \log \left|\lambda+\eta_j\right|   \\[0.5em]
1& \log \left|\lambda+\eta_j\right|
\end{pmatrix}
, & \mathrm{if} \,\,\, \beta_j=0,  \\[1.5em]
\mathcal{O}
\begin{pmatrix}
\left|\lambda+\eta_j\right|^{\beta_j/2} &\left|\lambda+\eta_j\right|^{\beta_j/2}\\[0.5em]
\left|\lambda+\eta_j\right|^{\beta_j/2} &\left|\lambda+\eta_j\right|^{\beta_j/2}
\end{pmatrix}
,  & \mathrm{if} \,\,\, \beta_j\in\left(-1, 0\right),
\end{cases}
\end{gather}
The local behavior near the singularities \( \pm\eta_0 \) differs for the two types of reflection coefficients. For the first type, \( r = r_0 \) with \( \beta_0 \neq 0 \), the behavior is:
\begin{gather}
T(\lambda; \x, 0)=
\begin{cases}
\mathcal{O}
\begin{pmatrix}
 \left|\lambda\mp\eta_0\right|^{\mp\beta_0/2}&\left|\lambda\mp\eta_0\right|^{\pm\beta_0/2}\\[0.5em]
  \left|\lambda\mp\eta_0\right|^{\mp\beta_0/2}&\left|\lambda\mp\eta_0\right|^{\pm\beta_0/2}
\end{pmatrix}
, & \mathrm{if} \,\,\,\beta_0\in\left(0, +\infty\right), \\[1.5em]
\mathcal{O}
\begin{pmatrix}
 \left|\lambda\mp\eta_0\right|^{\beta_0/2}&\left|\lambda\mp\eta_0\right|^{\beta_0/2}\\[0.5em]
  \left|\lambda\mp\eta_0\right|^{\beta_0/2}&\left|\lambda\mp\eta_0\right|^{\beta_0/2}
\end{pmatrix}
, & \mathrm{if} \,\,\,\beta_0\in\left(-1, 0\right),
\end{cases}
\end{gather}
For the second type, \( r = r_c \), the local behavior near \( \pm\eta_0 \) is:
\begin{gather}
T(\lambda; \x, 0)=
\begin{cases}
\mathcal{O}
\begin{pmatrix}
\log\left|\lambda-\eta_0\right| &1\\[0.5em]
\log\left|\lambda-\eta_0\right| &1
\end{pmatrix}
, & \mathrm{as} \,\,\,\lambda\to\eta_0,\\[1.5em]
\mathcal{O}
\begin{pmatrix}
1 & \log\left|\lambda+\eta_0\right| \\[0.5em]
1 & \log\left|\lambda+\eta_0\right|
\end{pmatrix}
, & \mathrm{as} \,\,\,\lambda\to-\eta_0.
\end{cases}
\end{gather}

\begin{figure}[!t]
\centering
\includegraphics[scale=0.32]{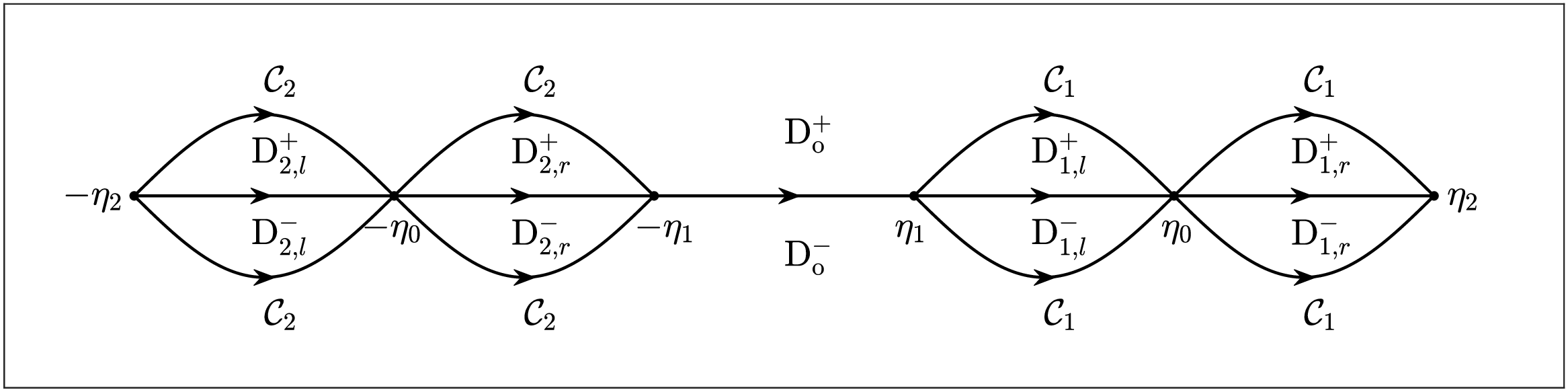}
\caption{Contour deformation by opening lenses for $r=r_0, r_c$ with $\beta_0\ne 0$}
\label{Deformation-0}
\end{figure}

\subsection{Riemann-Hilbert problem for $S(\lambda; \x, 0)$}

The next step is to deform the Riemann-Hilbert problem \( Y(\lambda; \x, 0) \) to a new Riemann-Hilbert problem \( S(\lambda; \x, 0) \). Initially, we handle the general reflection coefficient \( r = r_0, r_c \) with \( \beta_0 \neq 0 \). The contour deformation for this Riemann-Hilbert problem, accomplished by opening lenses, is illustrated in Figure \ref{Deformation-0}.

 The domains are defined as follows:
\( \mathrm{D}^+_{1, l} \) and \( \mathrm{D}^-_{1, l} \) are lenses above and below the interval \( (\eta_1, \eta_0) \), respectively.
 \( \mathrm{D}^+_{1, r} \) and \( \mathrm{D}^-_{1, r} \) are lenses above and below \( (\eta_0, \eta_2) \), respectively.
\( \mathrm{D}^+_{2, l} \) and \( \mathrm{D}^-_{2, l} \) are lenses above and below \( (-\eta_2, -\eta_0) \), respectively.
\( \mathrm{D}^+_{2, r} \) and \( \mathrm{D}^-_{2, r} \) are lenses above and below \( (-\eta_0, -\eta_1) \), respectively.
We define the composite regions:
\bee
 \mathrm{D}_{j, l} = \mathrm{D}_{j, l}^+ \cup \mathrm{D}_{j, l}^-,\qquad
\mathrm{D}_{j, r} = \mathrm{D}_{j, r}^+ \cup \mathrm{D}_{j, r}^-,\qquad
\mathrm{D}_j^{\pm} = \mathrm{D}_{j, l}^{\pm} \cup \mathrm{D}_{1, r}^{\pm},\quad j=1,2,
%\mathrm{D}_1^- = \mathrm{D}_{1, l}^- \cup \mathrm{D}_{1, r}^-,
%\mathrm{D}_{2, l} = \mathrm{D}_{2, l}^+ \cup \mathrm{D}_{2, l}^-,
% \mathrm{D}_{2, r} = \mathrm{D}_{2, r}^+ \cup \mathrm{D}_{2, r}^-,
%\mathrm{D}_2^+ = \mathrm{D}_{2, l}^+ \cup \mathrm{D}_{2, r}^+,
% \mathrm{D}_2^- = \mathrm{D}_{2, l}^- \cup \mathrm{D}_{2, r}^-.
 \ene
Thus, the regions \( \mathrm{D}_1 \) and \( \mathrm{D}_2 \) are given by \( \mathrm{D}_1 = \mathrm{D}_{1, l} \cup \mathrm{D}_{1, r} \) and \( \mathrm{D}_2 = \mathrm{D}_{2, l} \cup \mathrm{D}_{2, r} \).

For the special case where the generalized reflection coefficient is \( r = r_0 \) and \( \beta_0 = 0 \), the contour deformation is shown in Figure \ref{Deformation-00}. The regions \( \mathrm{D}_1^+ \) and \( \mathrm{D}_1^- \) are lenses above and below \( (\eta_1, \eta_2) \), respectively, while \( \mathrm{D}_2^+ \) and \( \mathrm{D}_2^- \) are lenses above and below \( (-\eta_2, -\eta_1) \), respectively. The regions \( \mathrm{D}_1 \) and \( \mathrm{D}_2 \) are given by \( \mathrm{D}_1 = \mathrm{D}_1^+ \cup \mathrm{D}_1^- \) and \( \mathrm{D}_2 = \mathrm{D}_2^+ \cup \mathrm{D}_2^- \).

The exterior region \( \mathrm{D}_\mathrm{o} \) is defined as:
\bee
\mathrm{D}_\mathrm{o} = \mathbb{C} \setminus \overline{\mathrm{D}_1 \cup \mathrm{D}_2 \cup (-\eta_1, \eta_1)},\quad
 \mathrm{D}_\mathrm{o} = \mathrm{D}_\mathrm{o}^+ \cup \mathrm{D}_\mathrm{o}^- \cup (\eta_2, +\infty) \cup (-\infty, -\eta_2),
 \ene
 where \( \mathrm{D}_\mathrm{o}^+ \) and \( \mathrm{D}_\mathrm{o}^- \) refer to the portions of the exterior region in the upper and lower half-planes, respectively.

The \( 2 \times 2 \)-matrix-valued function \( S(\lambda; \x, 0) \) is defined as follows:
\begin{gather}\label{T-To-S}
S(\lambda; \x, 0)=\begin{cases}
T(\lambda; \x, 0)\mathcal{U}^{\x p_0}_{f_0}\left[-\mathrm{i}r^{-1}\right]^{-1}, & \mathrm{for}\,\,\, \lambda\in\mathrm{D}_1^+,  \\[0.5em]
T(\lambda; \x, 0)\mathcal{U}^{\x p_0}_{f_0}\left[-\mathrm{i}r^{-1}\right], & \mathrm{for}\,\,\, \lambda\in\mathrm{D}_1^-,   \\[0.5em]
T(\lambda; \x, 0)\mathcal{L}^{\x p_0}_{f_0}\left[-\mathrm{i}r^{-1}\right]^{-1}, & \mathrm{for}\,\,\,  \lambda\in\mathrm{D}_2^+,  \\[0.5em]
T(\lambda; \x, 0)\mathcal{L}^{\x p_0}_{f_0}\left[-\mathrm{i}r^{-1}\right], & \mathrm{for}\,\,\,  \lambda\in\mathrm{D}_2^-,  \\[0.5em]
T(\lambda; \x, 0), & \mathrm{for}\,\,\,  \lambda\in\mathrm{D}_\mathrm{o}.
\end{cases}
\end{gather}

\begin{RH}
The function \( S(\lambda; \x, 0) \) satisfies the following Riemann-Hilbert problem:
\begin{itemize}
\item{} \( S(\lambda; \x, 0) \) is analytic in \( \lambda \in \mathbb{C} \setminus \left( [-\eta_2, \eta_2] \cup \mathcal{C}_1 \cup \mathcal{C}_2 \right) \);

\item{} \( S(\lambda; \x, 0) \) normalizes to the identity matrix \( \mathbb{I}_2 \) as \( \lambda \to \infty \);

\item{} For \( \lambda \in (-\eta_2, \eta_2) \cup \mathcal{C}_1 \cup \mathcal{C}_2 \setminus \{ \pm \eta_1, \pm \eta_0 \} \), \( S(\lambda; \x, 0) \) admits continuous boundary values \( S_+ \) and \( S_- \), which are related by the following jump conditions:
For \( r = r_0, r_c \) with \( \beta_0 \neq 0 \):
\begin{gather}
S_+(\lambda; \x, 0)=S_-(\lambda; \x, 0)
\begin{cases}
\mathcal{U}^{\x p_0}_{f_0}\left[-\mathrm{i}r^{-1}\right], &\mathrm{for}\,\,\, \lambda\in\mathcal{C}_1\setminus\left\{\eta_1, \eta_2, \eta_0\right\},\\[0.5em]
\mathcal{L}^{\x p_0}_{f_0}\left[-\mathrm{i}r^{-1}\right], &\mathrm{for}\,\,\, \lambda\in\mathcal{C}_2\setminus\left\{-\eta_1, -\eta_2, -\eta_0\right\},\\[0.5em]
\mathrm{i}\sigma_1, &\mathrm{for}\,\,\, \lambda\in\left(\eta_1, \eta_2\right)\cup\left(-\eta_2, -\eta_1\right)\setminus \{\pm\eta_0\},\\[0.5em]
\mathrm{e}^{\Delta_1^0\sigma_3}, &\mathrm{for}\,\,\, \lambda\in\left(-\eta_1, \eta_1\right);
\end{cases}
\end{gather}
For the case \( r = r_0 \) with \( \beta_0 = 0 \):
\begin{gather}
S_+(\lambda; \x, 0)=S_-(\lambda; \x, 0)
\begin{cases}
\mathcal{U}^{\x p_0}_{f_0}\left[-\mathrm{i}r^{-1}\right], &\mathrm{for}\,\,\, \lambda\in\mathcal{C}_1\setminus\left\{\eta_1, \eta_2\right\},\\[0.5em]
\mathcal{L}^{\x p_0}_{f_0}\left[-\mathrm{i}r^{-1}\right], &\mathrm{for}\,\,\, \lambda\in\mathcal{C}_2\setminus\left\{-\eta_1, -\eta_2\right\},\\[0.5em]
\mathrm{i}\sigma_1, &\mathrm{for}\,\,\, \lambda\in\left(\eta_1, \eta_2\right)\cup\left(-\eta_2, -\eta_1\right),\\[0.5em]
\mathrm{e}^{\Delta_1^0\sigma_3}, &\mathrm{for}\,\,\, \lambda\in\left(-\eta_1, \eta_1\right).
\end{cases}
\end{gather}
\end{itemize}
\end{RH}

For the endpoints \( \pm \eta_j \) with \( j = 1, 2 \), the local behavior of \( S(\lambda; \x, 0) \) differs between the regions \( \mathrm{D}_\mathrm{o} \) and \( \mathrm{D}_1 \cup \mathrm{D}_2 \).
If \( \beta_j \in (0, +\infty) \):
\begin{gather}
S(\lambda; \x, 0)=
\begin{cases}
\mathcal{O}
\begin{pmatrix}
 \left|\lambda\mp\eta_j\right|^{\mp\beta_j/2}&\left|\lambda\mp\eta_j\right|^{\pm\beta_j/2}\\
  \left|\lambda\mp\eta_j\right|^{\mp\beta_j/2}&\left|\lambda\mp\eta_j\right|^{\pm\beta_j/2}
\end{pmatrix}
, & \mathrm{as} \,\,\,\lambda\in\mathrm{D}_{\mathrm{o}}\to\pm\eta_j, \\[1.5em]
\mathcal{O}
\begin{pmatrix}
 \left|\lambda\mp\eta_j\right|^{-\beta_j/2}&\left|\lambda\mp\eta_j\right|^{-\beta_j/2}\\
 \left|\lambda\mp\eta_j\right|^{-\beta_j/2}&\left|\lambda\mp\eta_j\right|^{-\beta_j/2}
\end{pmatrix}
, & \mathrm{as} \,\,\,\lambda\in\mathrm{D}_1\cup\mathrm{D}_2\to\pm\eta_j.
\end{cases}
\end{gather}
If \( \beta_j = 0 \):
\begin{gather}
S(\lambda; \x, 0)=
\begin{cases}
\mathcal{O}
\begin{pmatrix}
\log\left|\lambda-\eta_j\right| &1\\
\log\left|\lambda-\eta_j\right| &1
\end{pmatrix}
, & \mathrm{as} \,\,\,\lambda\in\mathrm{D}_{\mathrm{o}}\to\eta_j,\\[1.5em]
\mathcal{O}
\begin{pmatrix}
1 & \log\left|\lambda+\eta_j\right| \\
1 & \log\left|\lambda+\eta_j\right|
\end{pmatrix}
, & \mathrm{as} \,\,\,\lambda\in\mathrm{D}_{\mathrm{o}}\to-\eta_j, \\[1.5em]
\mathcal{O}
\begin{pmatrix}
\log\left|\lambda\mp\eta_j\right| & \log\left|\lambda\mp\eta_j\right| \\
\log\left|\lambda\mp\eta_j\right| & \log\left|\lambda\mp\eta_j\right|
\end{pmatrix}
, & \mathrm{as} \,\,\,\lambda\in\mathrm{D}_1\cup\mathrm{D}_2\to \pm\eta_j.
\end{cases}
\end{gather}
If \( \beta_j \in (-1, 0) \):
\begin{gather}
S(\lambda; \x, 0)=
\mathcal{O}
\begin{pmatrix}
 \left|\lambda\mp\eta_j\right|^{\beta_j/2}&\left|\lambda\mp\eta_j\right|^{\beta_j/2}  \\
 \left|\lambda\mp\eta_j\right|^{\beta_j/2}&\left|\lambda\mp\eta_j\right|^{\beta_j/2}
\end{pmatrix}, \quad
\mathrm{as} \,\,\,\lambda\in\mathrm{D}_1\cup\mathrm{D}_2\cup \mathrm{D}_{\mathrm{o}}\to\pm\eta_j.
\end{gather}

\begin{figure}[!t]
\centering
\includegraphics[scale=0.32]{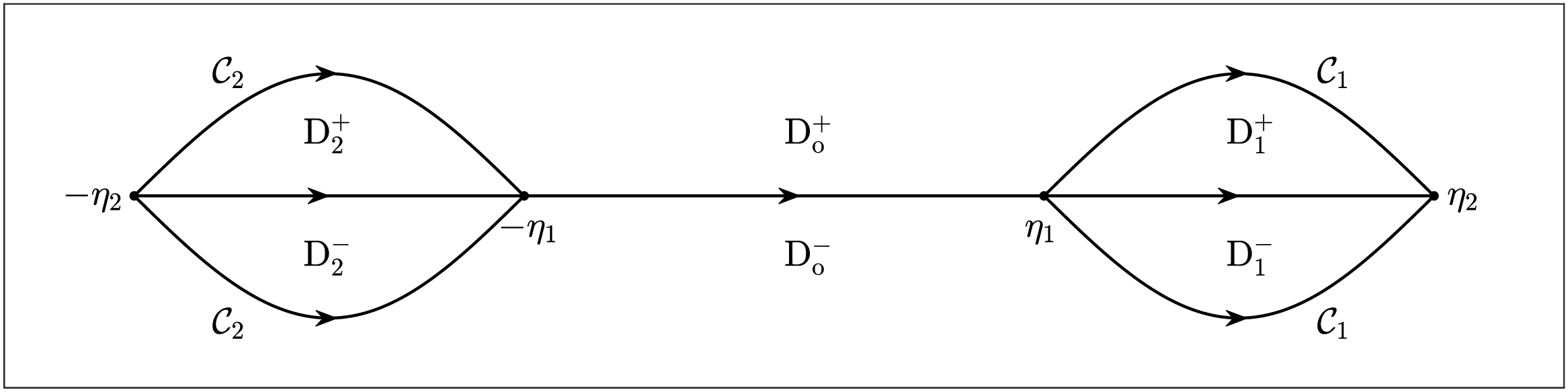}
\caption{Contour deformation by opening lenses for $r=r_0$ with $\beta_0=0$. }
\label{Deformation-00}
\end{figure}

For the singularities \( \pm \eta_0 \), the local behavior of \( S(\lambda; \x, 0) \) depends on the reflection coefficient.
For \( r = r_0 \) with \( \beta_0 \neq 0 \):
If \( \beta_0 \in (0, +\infty) \):
\begin{gather}
S(\lambda; \x, 0)=
\begin{cases}
\mathcal{O}
\begin{pmatrix}
 \left|\lambda\mp\eta_0\right|^{\mp\beta_0/2}&\left|\lambda\mp\eta_0\right|^{\pm\beta_0/2}\\
  \left|\lambda\mp\eta_0\right|^{\mp\beta_0/2}&\left|\lambda\mp\eta_0\right|^{\pm\beta_0/2}
\end{pmatrix}
, & \mathrm{as} \,\,\,\lambda\in\mathrm{D}_{\mathrm{o}}\to\pm\eta_0, \\[1.5em]
\mathcal{O}
\begin{pmatrix}
 \left|\lambda\mp\eta_0\right|^{-\beta_0/2}&\left|\lambda\mp\eta_0\right|^{-\beta_0/2}\\
 \left|\lambda\mp\eta_0\right|^{-\beta_0/2}&\left|\lambda\mp\eta_0\right|^{-\beta_0/2}
\end{pmatrix}
, & \mathrm{as} \,\,\,\lambda\in\mathrm{D}_1\cup\mathrm{D}_2\to\pm\eta_0,
\end{cases}
\end{gather}
and if $\beta_0\in\left(-1, 0\right)$:
\begin{gather}
S(\lambda; \x, 0)=
\mathcal{O}
\begin{pmatrix}
 \left|\lambda\mp\eta_0\right|^{\beta_0/2}&\left|\lambda\mp\eta_0\right|^{\beta_0/2}  \\
 \left|\lambda\mp\eta_0\right|^{\beta_0/2}&\left|\lambda\mp\eta_0\right|^{\beta_0/2}
\end{pmatrix}, \quad
\mathrm{as} \,\,\,\lambda\in\mathrm{D}_1\cup\mathrm{D}_2\cup \mathrm{D}_{\mathrm{o}}\to\pm\eta_0.
\end{gather}
For \( r = r_c \), the local behavior is:
 \begin{gather}
S(\lambda; \x, 0)=
\begin{cases}
\mathcal{O}
\begin{pmatrix}
\log\left|\lambda-\eta_0\right| &1\\
\log\left|\lambda-\eta_0\right| &1
\end{pmatrix}
, & \mathrm{as} \,\,\,\lambda\in\mathrm{D}_{\mathrm{o}}\to\eta_0,\\[1.5em]
\mathcal{O}
\begin{pmatrix}
1 & \log\left|\lambda+\eta_0\right| \\
1 & \log\left|\lambda+\eta_0\right|
\end{pmatrix}
, & \mathrm{as} \,\,\,\lambda\in\mathrm{D}_{\mathrm{o}}\to-\eta_0, \\[1.5em]
\mathcal{O}
\begin{pmatrix}
\log\left|\lambda\mp\eta_0\right| & \log\left|\lambda\mp\eta_0\right| \\
\log\left|\lambda\mp\eta_0\right| & \log\left|\lambda\mp\eta_0\right|
\end{pmatrix}
, & \mathrm{as} \,\,\,\lambda\in\mathrm{D}_1\cup\mathrm{D}_2\to \pm\eta_0
\end{cases}
\end{gather}

\subsection{Riemann-Hilbert problem for $E(\lambda; \x, 0)$}
We define the error matrix \( E(\lambda; \x, 0) \) as follows:
\begin{gather}\label{E0}
E(\lambda; \x, 0) = S(\lambda; \x, 0) P(\lambda; \x, 0)^{-1},
\end{gather}
where \( P(\lambda; \x, 0) \) denotes the global parametrix. This parametrix is structured according to the nature of the reflection coefficient \( r = r_0, r_c \) with \( \beta_0 \neq 0 \):
\begin{gather}
P(\lambda; \x, 0) =
\begin{cases}
P^\infty(\lambda; \x, 0), & \mathrm{if }\,\,\, \lambda \in \mathbb{C} \setminus \overline{B(\pm \eta_2, \pm \eta_0, \pm \eta_1)}, \\[1em]
P^{\eta_j}(\lambda; \x, 0), & \mathrm{if }\,\,\, \lambda \in B(\eta_j),\,\, j=0,1,2, \\[1em]
%P^{\eta_0}(\lambda; \x, 0), & \mathrm{if }\,\,\, \lambda \in B(\eta_0), \\
%P^{\eta_1}(\lambda; \x, 0), & \mathrm{if }\,\,\, \lambda \in B(\eta_1), \\
\sigma_2 P^{\eta_j}(-\lambda; \x, 0) \sigma_2, & \mathrm{if }\,\,\, \lambda \in B(-\eta_j),\,\, j=0,1,2.
%\sigma_2 P^{\eta_0}(-\lambda; \x, 0) \sigma_2, & \mathrm{if }\,\,\, \lambda \in B(-\eta_0), \\
%\sigma_2 P^{\eta_1}(-\lambda; \x, 0) \sigma_2, & \mathrm{if }\,\,\, \lambda \in B(-\eta_1).
\end{cases}
\end{gather}
Here, the set \( B(\pm \eta_2, \pm \eta_0, \pm \alpha) \) is defined as the union of disks:
\[
B(\pm \eta_2, \pm \eta_0, \pm \alpha) = B(\eta_2) \cup B(-\eta_2) \cup B(\eta_0) \cup B(-\eta_0) \cup B(\alpha) \cup B(-\alpha).
\]

In the specific scenario where \( r = r_0 \) and \( \beta_0 = 0 \), the parametrix takes the form:
\begin{gather}
P(\lambda; \x, 0) =
\begin{cases}
P^\infty(\lambda; \x, 0), & \mathrm{if }\,\,\, \lambda \in \mathbb{C} \setminus \overline{B(\pm \eta_2, \pm \eta_1)}, \\[1em]
P^{\eta_j}(\lambda; \x, 0), & \mathrm{if }\,\,\, \lambda \in B(\eta_2), \,\, j=1,2,\\[1em]
%P^{\eta_1}(\lambda; \x, 0), & \mathrm{if }\,\,\, \lambda \in B(\eta_1), \\
\sigma_2 P^{\eta_j}(-\lambda; \x, 0) \sigma_2, & \mathrm{if }\,\,\, \lambda \in B(-\eta_j),\,\, j=1,2.
%\sigma_2 P^{\eta_1}(-\lambda; \x, 0) \sigma_2, & \mathrm{if }\,\,\, \lambda \in B(-\eta_1).
\end{cases}
\end{gather}
The set \( B(\pm \eta_2, \pm \alpha) \) is similarly defined:
\[
B(\pm \eta_2, \pm \alpha) = B(\eta_2) \cup B(-\eta_2) \cup B(\alpha) \cup B(-\alpha).
\]

The outer parametrix \( P^\infty \) is expressed in terms of the theta functions as follows:
\begin{gather}
\begin{aligned}
P^\infty_{1, 1}(\lambda; \x, 0) &= \frac{\delta_1 + \delta_1^{-1}}{2} \frac{\vartheta_3\left(w_1 + \frac{1}{4} + \frac{\Delta_1^0}{2\pi \mathrm{i}}; \tau_1\right)}{\vartheta_3\left(w_1 + \frac{1}{4}; \tau_1\right)} \frac{\vartheta_3(0; \tau_1)}{\vartheta_3\left(\frac{\Delta_1^0}{2\pi \mathrm{i}}; \tau_1\right)}, \\[1em]
P^\infty_{1, 2}(\lambda; \x, 0) &= \frac{\delta_1 - \delta_1^{-1}}{2} \frac{\vartheta_3\left(w_1 - \frac{1}{4} - \frac{\Delta_1^0}{2\pi \mathrm{i}}; \tau_1\right)}{\vartheta_3\left(w_1 - \frac{1}{4}; \tau_1\right)} \frac{\vartheta_3(0; \tau_1)}{\vartheta_3\left(\frac{\Delta_1^0}{2\pi \mathrm{i}}; \tau_1\right)}, \\[1em]
P^\infty_{2, 1}(\lambda; \x, 0) &= \frac{\delta_1 - \delta_1^{-1}}{2} \frac{\vartheta_3\left(w_1 - \frac{1}{4} + \frac{\Delta_1^0}{2\pi \mathrm{i}}; \tau_1\right)}{\vartheta_3\left(w_1 - \frac{1}{4}; \tau_1\right)} \frac{\vartheta_3(0; \tau_1)}{\vartheta_3\left(\frac{\Delta_1^0}{2\pi \mathrm{i}}; \tau_1\right)}, \\[1em]
P^\infty_{2, 2}(\lambda; \x, 0) &= \frac{\delta_1 + \delta_1^{-1}}{2} \frac{\vartheta_3\left(w_1 + \frac{1}{4} - \frac{\Delta_1^0}{2\pi \mathrm{i}}; \tau_1\right)}{\vartheta_3\left(w_1 + \frac{1}{4}; \tau_1\right)} \frac{\vartheta_3(0; \tau_1)}{\vartheta_3\left(\frac{\Delta_1^0}{2\pi \mathrm{i}}; \tau_1\right)}.
\end{aligned}
\end{gather}
In this formulation, \( \delta_1 = \delta_1(\lambda) \) represents a branch of the expression
\begin{gather}
 \left(\frac{(\lambda + \eta_1)(\lambda - \eta_2)}{(\lambda + \eta_2)(\lambda - \eta_1)}\right)^{1/4},
\end{gather}
which is defined with a branch cut along the intervals \( [\eta_1, \eta_2] \cup [-\eta_2, -\eta_1] \). This branch is normalized such that \( \delta_1 = 1 + \mathcal{O}(\lambda^{-1}) \) as \( \lambda \to \infty \). Furthermore, \( w_1 = w_1(\lambda) \) is defined by the integral:
\begin{gather}
w_1 = -\frac{\eta_2}{4eK(m_1)} \int_{\eta_2}^\lambda \frac{\mathrm{d}\zeta}{R_0(\zeta)}.
\end{gather}

The conformal mapping $\mathcal{F}_0^{\eta_2}=p_0^2$ is defined in the vicinity of $\lambda=\eta_2$, mapping the domains as follows: $\mathrm{D}_\mathrm{o}\cap B(\eta_2)$ to $\mathrm{D}_1^\zeta\cap B^\zeta(0)$, $\mathrm{D}_{1}^+\cap B(\eta_2)$ to $\mathrm{D}_2^\zeta\cap B^\zeta(0)$, and $\mathrm{D}_{1}^-\cap B(\eta_2)$ to $\mathrm{D}_3^\zeta\cap B^\zeta(0)$. The local parametrix $P^{\eta_2}(\lambda; \x, 0)$ for $\lambda\in\left(\mathrm{D}_1\cup\mathrm{D}_\mathrm{o}\right)\cap B(\eta_2)$ is given by
\begin{gather}
P^{\eta_2}(\lambda; \x, 0) = P^\infty(\lambda; \x, 0) A_0^{\eta_2} C \zeta_{\eta_2}^{-\sigma_3/4} M^{\mathrm{mB}}(\zeta_{\eta_2}; \beta_2) \mathrm{e}^{-\sqrt{\zeta_{\eta_2}}\sigma_3} \left(A_0^{\eta_2}\right)^{-1},
\end{gather}
where $\zeta_{\eta_2} = \x^2 \mathcal{F}_0^{\eta_2}$ and $A_0^{\eta_2} = \left(\mathrm{e}^{\pi \mathrm{i}/4}/f_0 d^{\eta_2}\right)^{\sigma_3} \sigma_2$. For the first generalized reflection coefficient, $r=r_0$, one has
\bee
d^{\eta_2} = (\lambda - \eta_1)^{\beta_1/2} (\lambda - \eta_2)^{\beta_2/2} \left|\lambda - \eta_0\right|^{\beta_0/2} \gamma(\lambda)^{1/2},
 \ene
 while for the second coefficient $r=r_c$, it is
 \bee d^{\eta_2} = (\lambda - \eta_1)^{\beta_1/2} (\lambda - \eta_2)^{\beta_2/2} \chi_c(\lambda)^{1/2} \gamma(\lambda)^{1/2}.
 \ene

In the neighborhood of $\lambda=\eta_1$, the conformal mapping $\mathcal{F}_0^{\eta_1}$ is defined as $\mathcal{F}_0^{\eta_1} = \left(p_0 \pm \Omega_1/2\right)^2$ for $\lambda \in B(\eta_1) \cap \mathbb{C}^\pm$. This mapping similarly organizes the domains: $\mathrm{D}_\mathrm{o} \cap B(\eta_1)$ to $\mathrm{D}_1^\zeta\cap B^\zeta(0)$, $\mathrm{D}_{1}^-\cap B(\eta_1)$ to $\mathrm{D}_2^\zeta\cap B^\zeta(0)$, and $\mathrm{D}_{1}^+\cap B(\eta_1)$ to $\mathrm{D}_3^\zeta\cap B^\zeta(0)$. The local parametrix in this context, $P^{\eta_1}(\lambda; \x, 0)$ for $\lambda\in\left(\mathrm{D}_{1}^\pm\cup \mathrm{D}_\mathrm{o}^\pm\right)\cap B(\eta_1)$, is given by
\begin{gather}
P^{\eta_1}(\lambda; \x, 0) = P^\infty(\lambda; \x, 0) A_{0\pm}^{\eta_1} C \zeta_{\eta_1}^{-\sigma_3/4} M^{\mathrm{mB}}(\zeta_{\eta_1}; \beta_1) \mathrm{e}^{-\sqrt{\zeta_{\eta_1}}\sigma_3} \left(A_{0\pm}^{\eta_1}\right)^{-1},
\end{gather}
where $\zeta_{\eta_1} = \x^2 \mathcal{F}_0^{\eta_1}$ and $A_{0\pm}^{\eta_1} = \left(\mathrm{e}^{\pi \mathrm{i}/4 \mp \x \Omega_1 /2}/f_0 d^{\eta_1}\right)^{\sigma_3} \sigma_1$. The function $d^{\eta_1}$ is characterized as
\bee
d^{\eta_1} = (\eta_1 - \lambda)^{\beta_1/2} (\eta_2 - \lambda)^{\beta_2/2} \left|\lambda - \eta_0\right|^{\beta_0/2} \gamma(\lambda)^{1/2}\ene for the first type,
and
\bee
d^{\eta_1} = (\eta_1 - \lambda)^{\beta_1/2} (\eta_2 - \lambda)^{\beta_2/2} \chi_c(\lambda)^{1/2} \gamma(\lambda)^{1/2}
 \ene
for the second type.

For the first type of  generalized reflection coefficient  $r=r_0$ with $\beta_0\ne 0$, the conformal mapping is defined as follows:
$\mathcal{F}_0^{\eta_0}=\mp\left(p_0-p_{0\pm}(\eta_0)\right)$ if $\lambda\in B(\eta_0)\cap\mathbb{C}^\pm$.
It maps $\mathrm{D}_{1, r}^+\cap B(\eta_0)$ to $\mathrm{D}_1^\zeta\cap B^\zeta(0)$,
$\mathrm{D}_{1, l}^+\cap B(\eta_0)$ to $\mathrm{D}_4^\zeta\cap B^\zeta(0)$,
 $\mathrm{D}_{1, l}^-\cap B(\eta_0)$ to $\mathrm{D}_5^\zeta\cap B^\zeta(0)$,
 and
 $\mathrm{D}_{1, r}^-\cap B(\eta_0)$ to $\mathrm{D}_8^\zeta\cap B^\zeta(0)$.
 Define the domains
 $\mathrm{D}_{\mathrm{o}, 1, r}^+=\left(\mathcal{F}_0^{\eta_0}\right)^{-1} \left(\mathrm{D}^\zeta_2\cap B^\zeta(0)\right)$,
 $\mathrm{D}_{\mathrm{o}, 1, l}^+=\left(\mathcal{F}_0^{\eta_0}\right)^{-1} \left(\mathrm{D}^\zeta_3\cap B^\zeta(0)\right)$,
 $\mathrm{D}_{\mathrm{o}, 1, l}^-=\left(\mathcal{F}_0^{\eta_0}\right)^{-1} \left(\mathrm{D}^\zeta_6\cap B^\zeta(0)\right)$,
and
 $\mathrm{D}_{\mathrm{o}, 1, r}^-=\left(\mathcal{F}_0^{\eta_0}\right)^{-1} \left(\mathrm{D}^\zeta_7\cap B^\zeta(0)\right)$.
The local parametrix $P^{\eta_0}(\lambda; \x, 0)$ in the neighborhood of $\lambda=\eta_0$ is constructed as follows:
 for $\lambda\in\left(\mathrm{D}_{1, r}^+\cup\mathrm{D}_{\mathrm{o}, 1, r}^+\right)\cap B(\eta_0)$,
\begin{gather}
P^{\eta_0}(\lambda; \x, 0)=P^\infty(\lambda; \x, 0) A_{0r+}^{\eta_0}\mathrm{e}^{\beta_0\pi\mathrm{i}\sigma_3/4}\left(-\mathrm{i}\sigma_2\right)M^{\mathrm{mb}}(\zeta_{\eta_0}; \beta_0)\mathrm{e}^{\zeta_{\eta_0}\sigma_3}\left(A_{0r+}^{\eta_0}\right)^{-1},
\end{gather}
with $A_{0r+}^{\eta_0}=\left(\mathrm{e}^{\pi\mathrm{i}/4+\x p_{0+}(\eta_0)}/f_0d_r^{\eta_0}\right)^{\sigma_3}\sigma_2$;
for $\lambda\in \left(\mathrm{D}_{1, l}^+\cup\mathrm{D}_{\mathrm{o}, 1, l}^+\right)\cap B(\eta_0)$, $P^{\eta_0}$ is formulated as
\begin{gather}
P^{\eta_0}(\lambda; \x, 0)=P^\infty (\lambda; \x, 0)A_{0l+}^{\eta_0}\mathrm{e}^{-\beta_0\pi\mathrm{i}\sigma_3/4}\left(-\mathrm{i}\sigma_2\right)M^{\mathrm{mb}}(\zeta_{\eta_0}; \beta_0)\mathrm{e}^{\zeta_{\eta_0}\sigma_3}\left(A_{0l+}^{\eta_0}\right)^{-1},
\end{gather}
with $A_{0l+}^{\eta_0}=\left(\mathrm{e}^{\pi\mathrm{i}/4+\x p_{0+}(\eta_0)}/f_0d_l^{\eta_0}\right)^{\sigma_3}\sigma_2$;
for $\lambda\in\left(\mathrm{D}_{1, l}^-\cup\mathrm{D}_{\mathrm{o}, 1, l}^-\right)\cap B(\eta_0)$,  $P^{\eta_0}$ is written as
\begin{gather}
P^{\eta_0}(\lambda; \x, 0)=P^\infty(\lambda; \x, 0) A_{0l-}^{\eta_0}\mathrm{e}^{\beta_0\pi\mathrm{i}\sigma_3/4}M^{\mathrm{mb}}(\zeta_{\eta_0}; \beta_0)\mathrm{e}^{-\zeta_{\eta_0}\sigma_3}\left(A_{0l-}^{\eta_0}\right)^{-1},
\end{gather}
with $A_{0l-}^{\eta_0}=\left(\mathrm{e}^{\pi\mathrm{i}/4+\x p_{0-}(\eta_0)}/f_0d_l^{\eta_0}\right)^{\sigma_3}\sigma_2$;
for $\lambda\in \left(\mathrm{D}_{1, r}^-\cup\mathrm{D}_{\mathrm{o}, 1, r}^-\right)\cap B(\eta_0)$, $P^{\eta_0}$ is expressed as
\begin{gather}
P^{\eta_0}(\lambda; \x, 0)=P^\infty(\lambda; \x, 0) A_{0r-}^{\eta_0} \mathrm{e}^{-\beta_0\pi\mathrm{i}\sigma_3/4}M^{\mathrm{mb}}(\zeta_{\eta_0}; \beta_0)\mathrm{e}^{-\zeta_{\eta_0}\sigma_3}\left(A_{0r-}^{\eta_0}\right)^{-1},
\end{gather}
with $A_{0r-}^{\eta_0}=\left(\mathrm{e}^{\pi\mathrm{i}/4+\x p_{0-}(\eta_0)}/f_0d_r^{\eta_0}\right)^{\sigma_3}\sigma_2$,
where $\zeta_{\eta_0}=-\x\mathcal{F}^{\eta_0}_0$.
The functions $d_l^{\eta_0}$ and $d_r^{\eta_0}$ are defined as follows:
$d_l^{\eta_0}=
(\lambda-\eta_1)^{\beta_1/2}(\eta_2-\lambda)^{\beta_2/2}(\lambda-\eta_0)^{\beta_0/2}\gamma(\lambda)^{1/2}$
and
$d_r^{\eta_0}=
(\lambda-\eta_1)^{\beta_1/2}(\eta_2-\lambda)^{\beta_2/2}(\eta_0-\lambda)^{\beta_0/2}\gamma(\lambda)^{1/2}$.

To address the second generalized reflection coefficient \( r = r_c \), we define the conformal mapping as follows: \( \mathcal{F}_0^{\eta_0} = \mp 2(p_0 - p_{0\pm}(\eta_0)) \) for \( \lambda \in B(\eta_0) \cap \mathbb{C}^\pm \). This mapping transforms regions as follows: it maps \( \mathrm{D}_{1, r}^+ \cap B(\eta_0) \) to \( \mathrm{D}_1^\zeta \cap B^\zeta(0) \), \( \mathrm{D}_\mathrm{o}^+ \cap B(\eta_0) \) to \( \mathrm{D}_2^\zeta \cap B^\zeta(0) \), \( \mathrm{D}_{1, l}^+ \cap B(\eta_0) \) to \( \mathrm{D}_3^\zeta \cap B^\zeta(0) \), \( \mathrm{D}_{1, l}^- \cap B(\eta_0) \) to \( \mathrm{D}_4^\zeta \cap B^\zeta(0) \), \( \mathrm{D}_\mathrm{o}^- \cap B(\eta_0) \) to \( \mathrm{D}_5^\zeta \cap B^\zeta(0) \), and \( \mathrm{D}_{1, r}^- \cap B(\eta_0) \) to \( \mathrm{D}_6^\zeta \cap B^\zeta(0) \).
In the vicinity of \( \lambda = \eta_0 \), the local parametrix is constructed as follows: for \( \lambda \in \left(\mathrm{D}_1^+ \cup \mathrm{D}_\mathrm{o}^+\right) \cap B(\eta_0) \), we define the local parametrix as
\begin{gather}
P^{\eta_0}(\lambda; \x, 0) = P^\infty (\lambda; \x, 0) A^{\eta_0}_{0+} \left(\zeta_{\eta_0}^{\kappa_0 \sigma_3} \mathrm{i} \sigma_2 \mathrm{e}^{\kappa_0 \pi \mathrm{i} \sigma_3}\right)^{-1} M^{\mathrm{CH}}(\zeta_{\eta_0}; \kappa_0) \mathrm{e}^{\zeta_{\eta_0} \sigma_3 / 2} \left(A^{\eta_0}_{0+}\right)^{-1};
\end{gather}
for \( \lambda \in \left(\mathrm{D}_1^- \cup \mathrm{D}_\mathrm{o}^-\right) \cap B(\eta_0) \), it is expressed as
\begin{gather}
P^{\eta_0}(\lambda; \x, 0) = P^\infty (\lambda; \x, 0) A^{\eta_0}_{0-} \mathrm{e}^{-\kappa_0 \pi \mathrm{i} \sigma_3} M^{\mathrm{CH}}(\zeta_{\eta_0}; \kappa_0) \mathrm{e}^{-\zeta_{\eta_0} \sigma_3 / 2} \left(A^{\eta_0}_{0-}\right)^{-1},
\end{gather}
where \( A^{\eta_0}_{0\pm} = \left(\mathrm{e}^{\pi \mathrm{i}/4 + \x p_{0\pm}(\eta_0)}/f_0 d^{\eta_0}\right)^{\sigma_3} \sigma_2 \), \( \zeta_{\eta_0} = -\x \mathcal{F}_0^{\eta_0} \), \( \kappa_0 = \left(\mathrm{i}/\pi\right) \log c \in \mathrm{i} \mathbb{R} \), and \( d^{\eta_0} = \left(\lambda - \eta_1\right)^{\beta_1/2} \left(\eta_2 - \lambda\right)^{\beta_2/2} c^{1/2} \gamma(\lambda)^{1/2} \).

\begin{figure}[!t]
\centering
\includegraphics[scale=0.32]{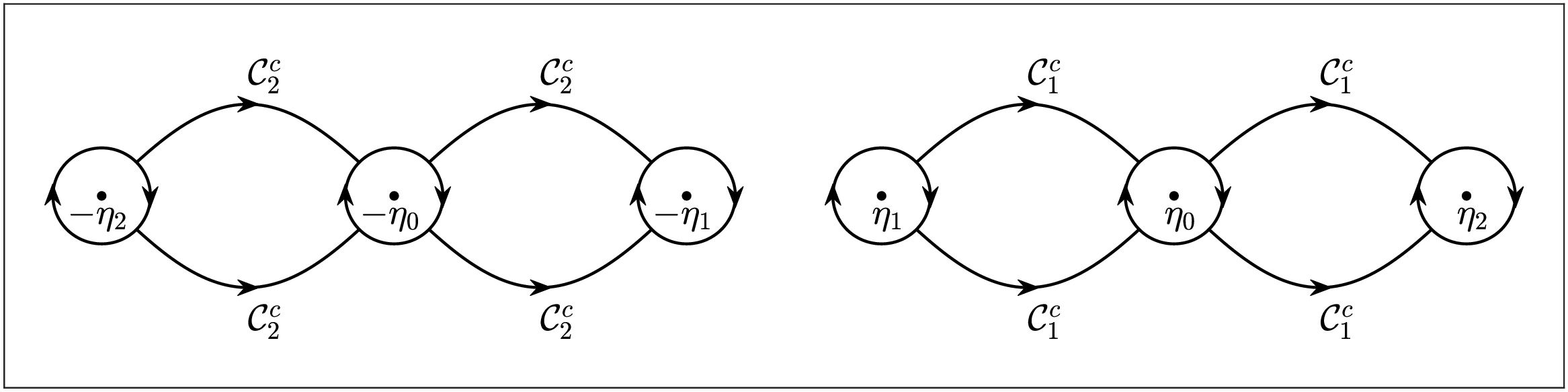}
\caption{Jump contours of the error matrix $E(\lambda; \x, 0)$ for $r=r_0, r_c$ with $\beta_0\ne 0$.}
\label{ErrorE0}
\end{figure}

For the general reflection coefficients \( r = r_0, r_c \) with \( \beta_0 \neq 0 \), the error matrix \( E(\lambda; \x, 0) \) is governed by a Riemann-Hilbert problem, as illustrated in Figure \ref{ErrorE0}. Specifically, \( E(\lambda; \x, 0) \) is analytic for \( \lambda \in \mathbb{C} \setminus \left(B(\pm\eta_2, \pm\eta_0, \pm\eta_1) \cup \mathcal{C}_1^c \cup \mathcal{C}_2^c\right) \) and approaches the identity matrix \( \mathbb{I}_2 \) at infinity. For \( \lambda \) in the boundaries \( B(\pm\eta_2, \pm\eta_0, \pm\eta_1) \cup \mathcal{C}_1^c \cup \mathcal{C}_2^c \), \( E(\lambda; \x, 0) \) has continuous boundary values related by the jump conditions given by
\begin{gather}\label{E-jump}
E_+(\lambda; \x, 0) = E_-(\lambda; \x, 0) V_0^E,
\end{gather}
with the jump matrices defined as follows:
\begin{gather}
V_0^E =
\begin{cases}
P^\infty (\lambda; \x, 0) \mathcal{U}^{\x p_0}_{f_0}\left[-ir^{-1}\right] \left(P^\infty (\lambda; \x, 0)\right)^{-1}, & \mathrm{for }\,\,\, \lambda \in \mathcal{C}^c_1,  \\[0.5em]
P^\infty (\lambda; \x, 0) \mathcal{L}^{\x p_0}_{f_0}\left[-ir^{-1}\right] \left(P^\infty (\lambda; \x, 0)\right)^{-1}, & \mathrm{for }\,\,\, \lambda \in \mathcal{C}^c_2,  \\[0.5em]
P^{\eta_2}(\lambda; \x, 0) \left(P^\infty(\lambda; \x, 0)\right)^{-1}, &\mathrm{for }\,\,\, \lambda \in \partial B(\eta_2),  \\[0.5em]
P^{\eta_0}(\lambda; \x, 0) \left(P^\infty(\lambda; \x, 0)\right)^{-1}, & \mathrm{for }\,\,\, \lambda \in \partial B(\eta_0),  \\[0.5em]
P^{\eta_1}(\lambda; \x, 0) \left(P^\infty(\lambda; \x, 0)\right)^{-1}, & \mathrm{for }\,\,\, \lambda \in \partial B(\eta_1),   \\[0.5em]
\sigma_2 P^{\eta_2}(-\lambda; \x, 0) \left(P^\infty\left(-\lambda; \x, 0\right)\right)^{-1} \sigma_2, & \mathrm{for }\,\,\,\lambda \in \partial B(-\eta_2), \\[0.5em]
\sigma_2 P^{\eta_0}(-\lambda; \x, 0) \left(P^\infty\left(-\lambda; \x, 0\right)\right)^{-1} \sigma_2, & \mathrm{for }\,\,\,\lambda \in \partial B(-\eta_0), \\[0.5em]
\sigma_2 P^{\eta_1}(-\lambda; \x, 0) \left(P^\infty\left(-\lambda; \x, 0\right)\right)^{-1} \sigma_2, & \mathrm{for }\,\,\, \lambda \in \partial B(-\eta_1).
\end{cases}
\end{gather}

In the special case of the first type of generalized reflection coefficient \( r = r_0 \) with \( \beta_0 = 0 \), a modified Riemann-Hilbert problem is presented, as shown in Figure \ref{ErrorE01}. Here, \( E(\lambda; \x, 0) \) remains analytic for \( \lambda \in \mathbb{C} \setminus \left(B(\pm\eta_2, \pm\eta_1) \cup \mathcal{C}_1^c \cup \mathcal{C}_2^c\right) \) and normalizes to \( \mathbb{I}_2 \) at infinity. For \( \lambda \in B(\pm\eta_2, \pm\eta_1) \cup \mathcal{C}_1^c \cup \mathcal{C}_2^c \), \( E(\lambda; \x, 0) \) exhibits continuous boundary values governed by the same jump conditions in \eqref{E-jump}, but with the following jump matrices:
\begin{gather}
V_0^E =
\begin{cases}
P^\infty (\lambda; \x, 0) \mathcal{U}^{\x p_0}_{f_0}\left[-ir^{-1}\right] \left(P^\infty (\lambda; \x, 0)\right)^{-1}, & \mathrm{for }\,\,\, \lambda \in \mathcal{C}^c_1,  \\[0.5em]
P^\infty (\lambda; \x, 0) \mathcal{L}^{\x p_0}_{f_0}\left[-ir^{-1}\right] \left(P^\infty (\lambda; \x, 0)\right)^{-1}, & \mathrm{for }\,\,\, \lambda \in \mathcal{C}^c_2,  \\[0.5em]
P^{\eta_2}(\lambda; \x, 0) \left(P^\infty(\lambda; \x, 0)\right)^{-1}, & \mathrm{for }\,\,\, \lambda \in \partial B(\eta_2),  \\[0.5em]
P^{\eta_1}(\lambda; \x, 0) \left(P^\infty(\lambda; \x, 0)\right)^{-1}, & \mathrm{for }\,\,\, \lambda \in \partial B(\eta_1),   \\[0.5em]
\sigma_2 P^{\eta_2}(-\lambda; \x, 0) \left(P^\infty\left(-\lambda; \x, 0\right)\right)^{-1} \sigma_2, & \mathrm{for }\,\,\, \lambda \in \partial B(-\eta_2), \\[0.5em]
\sigma_2 P^{\eta_0}(-\lambda; \x, 0) \left(P^\infty\left(-\lambda; \x, 0\right)\right)^{-1} \sigma_2, & \mathrm{for }\,\,\, \lambda \in \partial B(-\eta_0).
\end{cases}
\end{gather}

\begin{figure}[!t]
\centering
\includegraphics[scale=0.32]{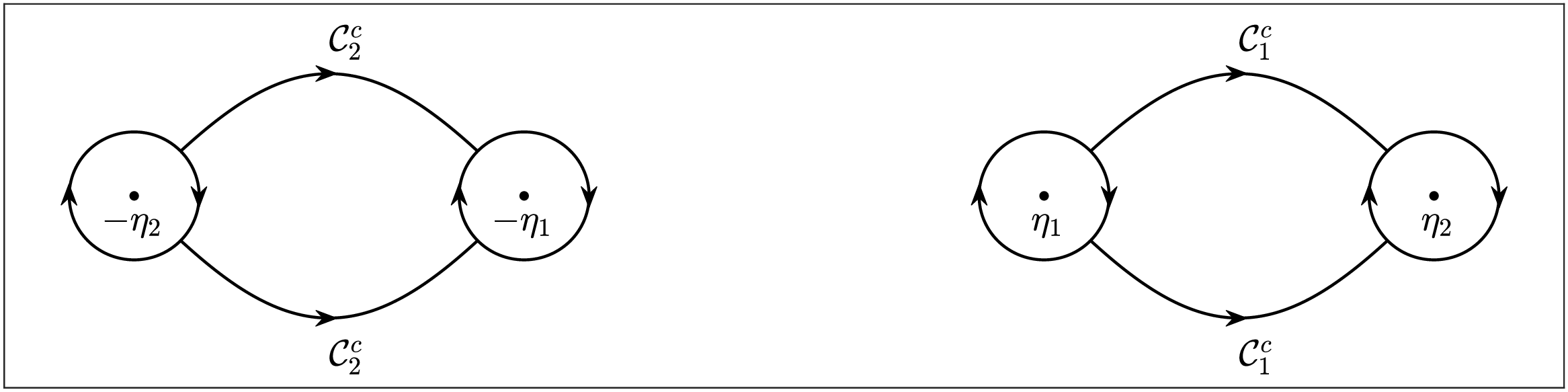}
\caption{Jump contours of the error matrix $E(\lambda; \x, 0)$ for $r=r_0$ with $\beta_0=0$}
\label{ErrorE01}
\end{figure}

Near each self-intersection point, the matrix \( E \) exhibits the local behavior:
\begin{gather}
E(\lambda; \x, 0) = \mathcal{O}
\begin{pmatrix}
1 & 1  \\
1 & 1
\end{pmatrix}
\end{gather}
as \( \lambda \) approaches each self-intersection point.

\begin{proposition}[Small norm estimate]
Let \( \beta_0, \beta_1, \beta_2 > -1 \). In the limit as \( \x \to -\infty \), the jump matrices \( V_0^E \) exhibit the following small norm estimates:
\begin{gather}
\begin{gathered}
\left\| V_0^E - \mathbb{I}_2 \right\|_{L^1 \cap L^2 \cap L^\infty \left( \mathcal{C}_1^c \cup \mathcal{C}_2^c \right)} = \mathcal{O}\left( \mathrm{e}^{-\mu_0 |\x|} \right), \\[0.5em]
\left\| V_0^E - \mathbb{I}_2 \right\|_{L^1 \cap L^2 \cap L^\infty \left( B\left( \pm \eta_2, \pm \eta_1, \pm \eta_0 \right) \right)} = \mathcal{O}\left( |\x|^{-1} \right),
\end{gathered}
\end{gather}
which implies
\begin{gather}\label{E0-Estimate}
E(0; \x, 0) = \mathbb{I}_2 + \mathcal{O}\left( |\x|^{-1} \right), \quad E^{[1]}(\x, 0) = \mathcal{O}\left( |\x|^{-1} \right), \quad \mathrm{as }\,\,\,\x \to -\infty,
\end{gather}
where \( E^{[1]}(\x, 0) = \lim_{\lambda \to 0} \lambda^{-1} \left( E(\lambda; \x, 0) - \mathbb{I}_2 \right) \).
\end{proposition}

\begin{proof}
The proof follows standard procedures and is therefore omitted here.
\end{proof}

\subsection{Large-\( \x \) Asymptotic of the Initial Values \( u(x, 0) \) as \( \x \to -\infty \)}
From \eqref{potential-formula}, the initial value \( u(x, 0) \) is derived from \( Y(\lambda; \x, 0) \):
\begin{gather}
\begin{array}{rl}
u(x, 0) =& \d \lim_{\lambda \to 0} \lambda^{-1} \left( Y(0; \x, 0)^{-1} Y(\lambda; \x, 0) \right)_{1, 2}, \\[1em]
x = & \x + \d\lim_{\lambda \to 0} \lambda^{-1} \left( \left( Y(0; \x, 0)^{-1} Y(\lambda; \x, 0) \right)_{1, 1} - 1 \right).
\end{array}
\end{gather}
Utilizing the transformations \( Y(\lambda; \x, 0) \mapsto T(\lambda; \x, 0) \) from \eqref{Conjugation-0}, \( T(\lambda; \x, 0) \mapsto S(\lambda; \x, 0) \) from \eqref{T-To-S}, and \( S(\lambda; \x, 0) \mapsto E(\lambda; \x, 0) \) from \eqref{E0}, along with the estimate in \eqref{E0-Estimate}, we find:
\begin{gather}
\begin{array}{rl}
u(x, 0) =& f_{0\pm}^{-2}(0) \mathrm{e}^{2\x g_{0\pm}(0)} \left( P_{\pm}^\infty(0; \x, 0)^{-1} P_{\pm}^{\infty, [1]}(\x, 0) \right)_{1, 2} + \mathcal{O}\left( |\x|^{-1} \right), \\[1em]
x =& \x + \left( P_{\pm}^\infty(0; \x, 0)^{-1} P_{\pm}^{\infty, [1]}(\x, 0) \right)_{1, 1} + f_{0\pm}^{-1}(0) f_{0\pm}^{[1]} - \x g_{0\pm}^{[1]} + \mathcal{O}\left( |\x|^{-1} \right),
\end{array}
\end{gather}
where
\begin{gather}
\begin{array}{rl}
P_{\pm}^{\infty, [1]}(\x, 0) =& \d\lim_{\lambda \in \mathbb{C}^\pm \to 0} \lambda^{-1} \left( P^\infty(\lambda; \x, 0) - P_\pm^\infty(0; \x, 0) \right), \\[1em]
f_{0\pm}^{[1]} =& \d\lim_{\lambda \in \mathbb{C}^\pm \to 0} \lambda^{-1} \left( f_0(\lambda) - f_{0\pm}(0) \right), \\[1em]
g_{0\pm}^{[1]} =& \d\lim_{\lambda \in \mathbb{C}^\pm \to 0} \lambda^{-1} \left( g_0(\lambda) - g_{0\pm}(0) \right).
\end{array}
\end{gather}
A straightforward calculation yields \eqref{large-x-left} in Theorem \ref{large-x}.

\section{Long-time asymptotic for the soliton gas $u(x, t)$}

In this section, we derive the long-time asymptotics of the soliton gas in the limit as \( t \to +\infty \).
For the case where \( \xi > -\eta_2^{-2} \), a standard small norm argument leads  to \eqref{large-right}. When \( \xi < -\eta_2^{-2} \), we can similarly consider two scenarios: one with \( r = r_0, r_c \) and \( \beta_0 \neq 0 \), and the other with \( r = r_0 \) and \( \beta_0 = 0 \). The derivation of long-time asymptotics for the latter case is more straightforward than for the former, and thus we shall omit the details for simplicity, leaving this case for the reader to explore.

We focus on deriving the long-time asymptotics for the soliton gas $u(x, t)$ in the regions \( \xi \in \left(-\infty, \xi_\mathrm{crit}\right) \cup \left(\xi_\mathrm{crit}, \xi_0\right) \cup \left(\xi_0, -\eta_2^{-2}\right) \).
To ensure that \( E(\lambda; \x, t) \) normalizes to the identity matrix \( \mathbb{I}_2 \) at infinity, and that its jump matrices decay uniformly and exponentially to \( \mathbb{I}_2 \), we will apply a series of transformations: \( Y(\lambda; \x, t) \mapsto T(\lambda; \x, t) \mapsto S(\lambda; \x, t) \mapsto E(\lambda; \x, t) \).
Following the Deift-Zhou steepest descent method, we proceed to perform triangular decompositions to facilitate contour deformation:
\begin{gather}
\begin{gathered}
\mathcal{L}^{t\theta}\left[\mathrm{i}r\right] = \mathcal{U}^{t\theta}\left[-\mathrm{i}r^{-1}\right] \left(\mathcal{L}^{t\theta}\left[\mathrm{i}r\right] + \mathcal{U}^{t\theta}\left[\mathrm{i}r^{-1}\right] - 2\,\mathbb{I}_2\right) \mathcal{U}^{t\theta}\left[-\mathrm{i}r^{-1}\right], \\[0.5em]
\mathcal{U}^{t\theta}\left[\mathrm{i}r\right] = \mathcal{L}^{t\theta}\left[-\mathrm{i}r^{-1}\right] \left(\mathcal{L}^{t\theta}\left[\mathrm{i}r^{-1}\right] + \mathcal{U}^{t\theta}\left[\mathrm{i}r\right] - 2\,\mathbb{I}_2\right) \mathcal{L}^{t\theta}\left[-\mathrm{i}r^{-1}\right].
\end{gathered}
\end{gather}
However, based on the sign of \( \Re\left(\theta\right) \) illustrated in Figure \ref{Sign}, we observe that exponential decay fails on the corresponding lenses. Consequently, prior to contour deformation, a conjugation operation is required, achieved by introducing an appropriate \( g \)-function. Additionally, a scalar \( f \)-function is introduced to ensure constant jump matrices.

Local behaviors play a crucial role in the construction of local parametrices, as discussed in the previous section. In this section, for the sake of clarity and simplicity, we will omit the detailed examination of the local behaviors of \( T \), \( S \), and \( E \).

\begin{figure}[!t]
\centering
\includegraphics[scale=0.32]{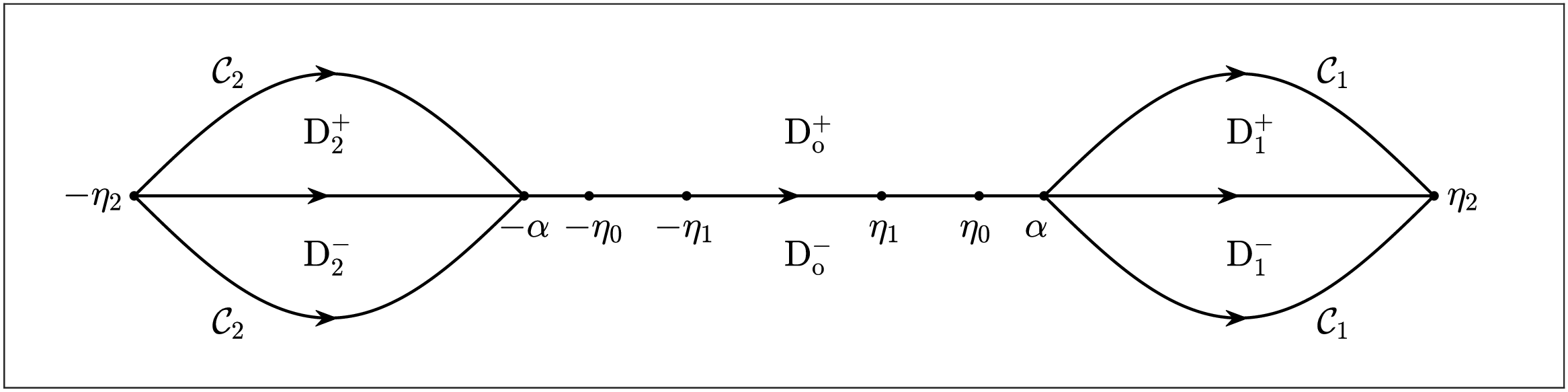}
\caption{Contour deformation in the region $\xi\in\left(\xi_0, -\eta_2^{-2}\right)$.}
\label{Deformation-2}
\end{figure}

\subsection{Riemann-Hilbert problem for $T(\lambda; \x, t)$}

The function \( f = f(\lambda; \xi) \) is analytic in the variable \( \lambda \) for \( \lambda \in \mathbb{C} \setminus [-\eta_2, \eta_2] \) and approaches the limit 1 as \( \lambda \to \infty \). The continuous boundary values, denoted as \( f_\pm \), exhibit the following relationships:
\begin{gather}
\begin{aligned}
& f_+(\lambda; \xi) f_-(\lambda; \xi) = r, && \mathrm{for}\,\,\, \lambda \in (\alpha, \eta_2), \\[0.5em]
& f_+(\lambda; \xi) f_-(\lambda; \xi) = r^{-1}, &&\mathrm{for}\,\,\, \lambda \in (-\eta_2, -\alpha), \\[0.5em]
& f_+^{-1}(\lambda; \xi) f_-(\lambda; \xi) = \mathrm{e}^{\Omega^\alpha \phi^\alpha}, && \mathrm{for}\,\,\, \lambda \in (-\alpha, \alpha),
\end{aligned}
\end{gather}
when \( \xi \in (\xi_0, -\eta_2^{-2}) \). Conversely, for \( \xi < \xi_0 \), the boundary values relate as follows:
\begin{gather}
\begin{aligned}
&f_+(\lambda; \xi)f_-(\lambda; \xi) =r, && \mathrm{for}\,\,\, \lambda\in\left(\alpha, \eta_0\right)\cup\left(\eta_0, \eta_2\right),   \\[0.5em]
&f_+(\lambda; \xi)f_-(\lambda; \xi) =r^{-1}&& \mathrm{for}\,\,\, \lambda\in\left(-\eta_2, -\eta_0\right)\cup\left(-\eta_0, -\alpha\right),  \\[0.5em]
&f_+^{-1}(\lambda; \xi)f_-(\lambda; \xi)=\mathrm{e}^{\Omega^\alpha\phi^\alpha}, &&\mathrm{for}\,\,\, \lambda\in\left(-\alpha, \alpha\right),
\end{aligned}
\end{gather}
The \( f \)-function can be expressed through the logarithm of \( R(\lambda; \xi) \) utilizing Plemelj's formula as follows:
\begin{gather}
f(\lambda; \xi)=\exp\left\{\frac{R(\lambda; \xi)}{\pi \mathrm{i}}\left(\int_{\alpha}^{\eta_2}\frac{\log r(s)}{R_{+}\left(s; \xi\right)}\frac{\lambda}{s^2-\lambda^2}\mathrm{d}s
-\int_{0}^{\alpha}\frac{\Omega^\alpha\phi^\alpha}{R(s; \xi)}\frac{\lambda}{s^2-\lambda^2}\mathrm{d}s\right)\right\}.
\end{gather}
With this groundwork established, we define the following conjugation:
\begin{gather}\label{Conjugation}
T(\lambda; \x, t)=Y(\lambda; \x, t)\mathrm{e}^{tg(\lambda; \xi)\sigma_3}f(\lambda; \xi)^{-\sigma_3},
\end{gather}	

\begin{RH} \( T(\lambda; \xi, t) \) is a \( 2 \times 2 \) matrix-valued function that fulfills a Riemann-Hilbert problem:

 \begin{itemize}

 \item{} The function \( T(\lambda; \xi, t) \) is analytic for \( \lambda \in \mathbb{C} \setminus [-\eta_2, \eta_2] \);

 \item{} It normalizes to the identity matrix \( \mathbb{I}_2 \) as \( \lambda \to \infty \);

 \item{} For  $\lambda\in\left(-\eta_2, \eta_2\right)\setminus\left\{\pm\eta_1, \pm\eta_0, \pm\alpha\right\}$, $T(\lambda; \x, t)$ admits continuous boundary values, and they are related by the following jump conditions
\begin{gather}
T_+=T_-
\begin{cases}
\mathcal{U}^{tp_-}_{f_-}\left[-\mathrm{i}r^{-1}\right]\left(\mathrm{i}\sigma_1\right)\mathcal{U}^{tp_+}_{f_+}\left[-\mathrm{i}r^{-1}\right], &\mathrm{for}\,\,\, \lambda\in\left(\alpha, \eta_2\right),\\[0.5em]
\mathcal{L}^{tp_-}_{f_-}\left[-\mathrm{i}r^{-1}\right]\left(\mathrm{i}\sigma_1\right)\mathcal{L}^{tp_+}_{f_+}\left[-\mathrm{i}r^{-1}\right], &\mathrm{for}\,\,\, \lambda\in\left(-\eta_2, -\alpha\right),\\[0.5em]
f_-^{\sigma_3}\mathrm{e}^{tp_-\sigma_3}\mathcal{L}\left[\mathrm{i}r\right]\mathrm{e}^{-tp_+\sigma_3}f_+^{-\sigma_3}, &\mathrm{for}\,\,\, \lambda\in\left(\eta_1, \eta_0\right)\cup\left(\eta_0, \alpha\right),\\[0.5em]
f_-^{\sigma_3}\mathrm{e}^{tp_-\sigma_3}\mathcal{U}\left[\mathrm{i}r\right]\mathrm{e}^{-tp_+\sigma_3}f_+^{-\sigma_3}, &\mathrm{for}\,\,\, \lambda\in\left(-\alpha, -\eta_0\right)\cup\left(-\eta_0, -\eta_1\right),\\[0.5em]
\mathrm{e}^{\Delta^\alpha\sigma_3}, &\mathrm{for}\,\,\, \lambda\in\left(-\eta_1, \eta_1\right).
\end{cases}
\end{gather}
if $\xi\in\left(\xi_0, -\eta_2^{-2}\right)$; if $\xi\in\left(\xi_\mathrm{crit}, \xi_0\right)$, the continuous boundary values are related by
\begin{gather}
T_+=T_-
\begin{cases}
\mathcal{U}^{tp_-}_{f_-}\left[-\mathrm{i}r^{-1}\right]\left(\mathrm{i}\sigma_1\right)\mathcal{U}^{tp_+}_{f_+}\left[-\mathrm{i}r^{-1}\right], &\mathrm{for}\,\,\, \lambda\in\left(\alpha, \eta_0\right)\cup\left(\eta_0, \eta_2\right),\\[0.5em]
\mathcal{L}^{tp_-}_{f_-}\left[-\mathrm{i}r^{-1}\right]\left(\mathrm{i}\sigma_1\right)\mathcal{L}^{tp_+}_{f_+}\left[-\mathrm{i}r^{-1}\right], &\mathrm{for}\,\,\, \lambda\in\left(-\eta_2, -\eta_0\right)\cup\left(-\eta_0, -\alpha\right),\\[0.5em]
f_-^{\sigma_3}\mathrm{e}^{tp_-\sigma_3}\mathcal{L}\left[\mathrm{i}r\right]\mathrm{e}^{-tp_+\sigma_3}f_+^{-\sigma_3}, &\mathrm{for}\,\,\, \lambda\in\left(\eta_1, \alpha\right),\\[0.5em]
f_-^{\sigma_3}\mathrm{e}^{tp_-\sigma_3}\mathcal{U}\left[\mathrm{i}r\right]\mathrm{e}^{-tp_+\sigma_3}f_+^{-\sigma_3}, &\mathrm{for}\,\,\, \lambda\in\left(-\alpha,  -\eta_1\right),\\[0.5em]
\mathrm{e}^{\Delta^\alpha\sigma_3}, &\mathrm{for}\,\,\, \lambda\in\left(-\eta_1, \eta_1\right);
\end{cases}
\end{gather}
and if $\xi<\xi_\mathrm{crit}$, the continuous boundary values are related by
\begin{gather}
T_+=T_-
\begin{cases}
\mathcal{U}^{tp_-}_{f_-}\left[-\mathrm{i}r^{-1}\right]\left(\mathrm{i}\sigma_1\right)\mathcal{U}^{tp_+}_{f_+}\left[-\mathrm{i}r^{-1}\right], &\mathrm{for}\,\,\, \lambda\in\left(\eta_1, \eta_0\right)\cup\left(\eta_0, \eta_2\right),\\[0.5em]
\mathcal{L}^{tp_-}_{f_-}\left[-\mathrm{i}r^{-1}\right]\left(\mathrm{i}\sigma_1\right)\mathcal{L}^{tp_+}_{f_+}\left[-\mathrm{i}r^{-1}\right], &\mathrm{for}\,\,\, \lambda\in\left(-\eta_2, -\eta_0\right)\cup\left(-\eta_0, -\eta_1\right),\\[0.5em]
\mathrm{e}^{\Delta_1\sigma_3}, &\mathrm{for}\,\,\, \lambda\in\left(-\eta_1, \eta_1\right);
\end{cases}
\end{gather}
\end{itemize}
\end{RH}

\subsection{Riemann-Hilbert problem for $S(\lambda; \x, t)$}
The subsequent step involves deforming the Riemann-Hilbert problem \( Y(\lambda; \xi, t) \) to derive a new Riemann-Hilbert problem \( S(\lambda; \xi, t) \).
We begin by considering the case where \( \xi \in (\xi_0, -\eta_2^{-2}) \). The contour deformation is achieved by the introduction of lens-shaped regions, as illustrated in Figure \ref{Deformation-2}. In this configuration, the domains \( \mathrm{D}^+_1 \) and \( \mathrm{D}^-_1 \) represent lenses located above and below the interval \( (\alpha, \eta_2) \), respectively. Similarly, the domains \( \mathrm{D}^+_2 \) and \( \mathrm{D}^-_2 \) correspond to lenses situated above and below the interval \( (-\eta_2, -\alpha) \).
We define the combined regions as \( \mathrm{D}_1 = \mathrm{D}^+_1 \cup \mathrm{D}^-_1 \) and \( \mathrm{D}_2 = \mathrm{D}^+_2 \cup \mathrm{D}^-_2 \).

For the case \( \xi \in (\xi_\mathrm{crit}, \xi_0) \), the contour deformation of the Riemann-Hilbert problem is depicted in Figure \ref{Deformation-3}. Here, we identify the domains \( \mathrm{D}^+_{1, l} \) and \( \mathrm{D}^-_{1, l} \) as lenses above and below the interval \( (\alpha, \eta_0) \), respectively, while \( \mathrm{D}^+_{1, r} \) and \( \mathrm{D}^-_{1, r} \) are lenses above and below \( (\eta_0, \eta_2) \). The domains \( \mathrm{D}^+_{2, l} \) and \( \mathrm{D}^-_{2, l} \) correspond to lenses above and below \( (-\eta_2, -\eta_0) \), whereas \( \mathrm{D}^+_{2, r} \) and \( \mathrm{D}^-_{2, r} \) are positioned above and below \( (-\eta_0, -\alpha) \).
For convenience, we introduce the following definitions: 
\bee
\mathrm{D}_{j, l}=\mathrm{D}_{j, l}^+\cup \mathrm{D}_{j, l}^-,\quad
\mathrm{D}_{j, r}=\mathrm{D}_{j, r}^+\cup \mathrm{D}_{j, r}^-,\quad
\mathrm{D}_j^{\pm}=\mathrm{D}_{j, l}^{\pm}\cup \mathrm{D}_{j, r}^{\pm},\quad 
%$\mathrm{D}_1^-=\mathrm{D}_{1, l}^-\cup \mathrm{D}_{1, r}^-$,
%$\mathrm{D}_{2, l}=\mathrm{D}_{2, l}^+\cup \mathrm{D}_{2, l}^-$,
%$\mathrm{D}_{2, r}=\mathrm{D}_{2, r}^+\cup \mathrm{D}_{2, r}^-$,
%$\mathrm{D}_2^+=\mathrm{D}_{2, l}^+\cup \mathrm{D}_{2, r}^+$,
%$\mathrm{D}_2^-=\mathrm{D}_{2, l}^-\cup \mathrm{D}_{2, r}^-$,
\mathrm{D}_j=\mathrm{D}_{j, l}\cup \mathrm{D}_{j, r},\quad j=1,2.
%and $\mathrm{D}_2=\mathrm{D}_{2, l}\cup \mathrm{D}_{2, r}$.
\ene
In the scenario where \( \xi < \xi_\mathrm{crit} \), the contour deformation follows the same configuration as depicted in Figure \ref{Deformation-0}.

For the aforementioned three regions, we denote the domain outside the lenses as \( \mathrm{D}_\mathrm{o} \), defined by
$\mathrm{D}_\mathrm{o} = \mathbb{C} \setminus \overline{\mathrm{D}_1 \cup \mathrm{D}_2 \cup (-\eta_1, \eta_1)}$,
which can be expressed as the union of regions:
$\mathrm{D}_\mathrm{o} = \mathrm{D}_\mathrm{o}^+ \cup \mathrm{D}_\mathrm{o}^- \cup (\eta_2, +\infty) \cup (-\infty, -\eta_2),$
where \( \mathrm{D}_\mathrm{o}^+ \) and \( \mathrm{D}_\mathrm{o}^- \) represent the portions in the upper and lower half-planes, respectively.

The \( 2 \times 2 \) matrix-valued function \( S(\lambda; \xi, t) \) is defined as follows:
\begin{gather}\label{T-To-S-t}
S(\lambda; \x, t)=T(\lambda; \x, t)\begin{cases}
\mathcal{U}^{tp}_{f}\left[-\mathrm{i}r^{-1}\right]^{-1}, & \mathrm{for}\,\,\, \lambda\in\mathrm{D}_1^+,  \\[0.5em]
\mathcal{U}^{tp}_{f}\left[-\mathrm{i}r^{-1}\right], & \mathrm{for}\,\,\, \lambda\in\mathrm{D}_1^-,   \\[0.5em]
\mathcal{L}^{tp}_{f}\left[-\mathrm{i}r^{-1}\right]^{-1}, & \mathrm{for}\,\,\,  \lambda\in\mathrm{D}_2^+,  \\[0.5em]
\mathcal{L}^{tp}_{f}\left[-\mathrm{i}r^{-1}\right], & \mathrm{for}\,\,\,  \lambda\in\mathrm{D}_2^-,  \\[0.5em]
I, & \mathrm{for}\,\,\,  \lambda\in\mathrm{D}_\mathrm{o}.
\end{cases}
\end{gather}

\begin{RH}  \( S(\lambda; \xi, t) \) solves the following Riemann-Hilbert problem:

 \begin{itemize}

 \item{} It is analytic in \( \lambda \) for \( \lambda \in \mathbb{C} \setminus \left( \left[-\eta_2, \eta_2\right] \cup \mathcal{C}_1 \cup \mathcal{C}_2 \right) \);

 \item{}It normalizes to the identity matrix \( \mathbb{I}_2 \) as \( \lambda \to \infty \);

 \item{} For \( \lambda \in \left(-\eta_2, \eta_2\right) \cup \mathcal{C}_1 \cup \mathcal{C}_2 \setminus \left\{ \pm\eta_1, \pm\eta_0 \right\} \), \( S(\lambda; \xi, 0) \) admits continuous boundary values denoted by \( S_+(\lambda; \xi, t) \) and \( S_-(\lambda; \xi, t) \). These boundary values are related by the following jump conditions:
If \( \xi \in \left(\xi_0, -\eta_2^{-2}\right) \), the jump condition is:
\begin{gather}
S_+=S_-
\begin{cases}
\mathcal{U}^{tp}_{f}\left[-\mathrm{i}r^{-1}\right], &\mathrm{for}\,\,\, \lambda\in\mathcal{C}_1\setminus\left\{\alpha, \eta_2\right\},\\[0.5em]
\mathcal{L}^{tp}_{f}\left[-\mathrm{i}r^{-1}\right], &\mathrm{for}\,\,\, \lambda\in\mathcal{C}_2\setminus\left\{-\alpha, -\eta_2\right\},\\[0.5em]
\mathrm{i}\sigma_1, &\mathrm{for}\,\,\, \lambda\in\left(\alpha, \eta_2\right)\cup\left(-\eta_2, -\alpha\right),\\[0.5em]
f_-^{\sigma_3}\mathrm{e}^{tp_-\sigma_3}\mathcal{L}\left[-\mathrm{i}r\right]\mathrm{e}^{-tp_+\sigma_3}f_+^{-\sigma_3}, &\mathrm{for}\,\,\, \lambda\in\left(\eta_1, \eta_0\right)\cup\left(\eta_0, \alpha\right),\\[0.5em]
f_-^{\sigma_3}\mathrm{e}^{tp_-\sigma_3}\mathcal{U}\left[\mathrm{i}r\right]\mathrm{e}^{-tp_+\sigma_3}f_+^{-\sigma_3}, &\mathrm{for}\,\,\, \lambda\in\left(-\alpha, -\eta_0\right)\cup\left(-\eta_0, -\eta_1\right),\\[0.5em]
\mathrm{e}^{\Delta^\alpha\sigma_3}, &\mathrm{for}\,\,\, \lambda\in\left(-\eta_1, \eta_1\right);
\end{cases}
\end{gather}
If \( \xi \in \left(\xi_\mathrm{crit}, \xi_0\right) \), the jump condition is:
\begin{gather}
S_+=S_-
\begin{cases}
\mathcal{U}^{tp}_{f}\left[-\mathrm{i}r^{-1}\right], &\mathrm{for}\,\,\, \lambda\in\mathcal{C}_1\setminus\left\{\eta_0, \alpha, \eta_2\right\},\\[0.5em]
\mathcal{L}^{tp}_{f}\left[-\mathrm{i}r^{-1}\right], &\mathrm{for}\,\,\, \lambda\in\mathcal{C}_2\setminus\left\{-\eta_0, -\alpha, -\eta_2\right\},\\[0.5em]
\mathrm{i}\sigma_1, &\mathrm{for}\,\,\, \lambda\in\left(\alpha, \eta_2\right)\cup\left(-\eta_2, -\alpha \right)\setminus \{\pm\eta_0\},\\[0.5em]
f_-^{\sigma_3}\mathrm{e}^{tp_-\sigma_3}\mathcal{L}\left[-\mathrm{i}r\right]\mathrm{e}^{-tp_+\sigma_3}f_+^{-\sigma_3}, &\mathrm{for}\,\,\, \lambda\in\left(\eta_1, \alpha\right),\\[0.5em]
f_-^{\sigma_3}\mathrm{e}^{tp_-\sigma_3}\mathcal{U}\left[\mathrm{i}r\right]\mathrm{e}^{-tp_+\sigma_3}f_+^{-\sigma_3}, &\mathrm{for}\,\,\, \lambda\in\left(-\alpha, -\eta_1\right),\\[0.5em]
\mathrm{e}^{\Delta^\alpha\sigma_3}, &\mathrm{for}\,\,\, \lambda\in\left(-\eta_1, \eta_1\right);
\end{cases}
\end{gather}
If \( \xi < \xi_\mathrm{crit} \), the jump condition is:
\begin{gather}
S_+=S_-
\begin{cases}
\mathcal{U}^{tp}_{f}\left[-\mathrm{i}r^{-1}\right], &\mathrm{for}\,\,\, \lambda\in\mathcal{C}_1\setminus\left\{\eta_0, \eta_1, \eta_2\right\},\\[0.5em]
\mathcal{L}^{tp}_{f}\left[-\mathrm{i}r^{-1}\right], &\mathrm{for}\,\,\, \lambda\in\mathcal{C}_2\setminus\left\{-\eta_0, -\eta_1, -\eta_2\right\},\\[0.5em]
\mathrm{i}\sigma_1, &\mathrm{for}\,\,\, \lambda\in\left(\eta_1, \eta_2\right)\cup\left(-\eta_2, -\eta_1\right)\setminus \{\pm\eta_0\},\\[0.5em]
\mathrm{e}^{\Delta_1\sigma_3}, &\mathrm{for}\,\,\, \lambda\in\left(-\eta_1, \eta_1\right).
\end{cases}
\end{gather}
\end{itemize}
\end{RH}

\subsection{Riemann-Hilbert problem for $E(\lambda; \x, t)$}

We define the error matrix \( E(\lambda; \x, t) \) as follows:
\begin{gather}\label{E}
E(\lambda; \x, t)=S(\lambda; \x, t)P(\lambda; \x, t)^{-1},
\end{gather}
where \( P(\lambda; \x, t) \) represents the global parametrix, which varies across distinct regions. In the interval \( \xi \in (\xi_0, -\eta_2^{-2}) \), the formulation of the global parametrix \( P(\lambda; \x, t) \) is given by:
\begin{gather}
P(\lambda; \x, t)=
\begin{cases}
P^\infty\left(\lambda; \x, t\right), &\mathrm{for}\,\,\,\lambda\in\mathbb{C}\setminus\overline{B\left(\pm\eta_2, \pm\alpha\right)},  \\
P^{\eta_2}\left(\lambda; \x, t\right), & \mathrm{for}\,\,\,\lambda\in B\left(\eta_2\right),  \\
P^{\alpha}\left(\lambda; \x, t\right), & \mathrm{for}\,\,\,\lambda\in B\left(\alpha\right),  \\
\sigma_2P^{\eta_2}\left(-\lambda; \x, t\right)\sigma_2, & \mathrm{for}\,\,\,\lambda\in B\left(-\eta_2\right),  \\
\sigma_2P^{\alpha}\left(-\lambda; \x, t\right)\sigma_2, & \mathrm{for}\,\,\,\lambda\in B\left(-\alpha\right).
\end{cases}
\end{gather}
The outer parametrix, denoted \( P^\infty(\lambda; \x, t) = \left( P_{i,j}^\infty(\lambda; \x, t) \right)_{2 \times 2} \), is described as follows:
%\begin{gather}\label{out-alpha}
\begin{align}
&P^\infty_{1, 1}(\lambda; \x, t)=\frac{\delta_\alpha+\delta_\alpha^{-1}}{2}
\frac{\vartheta_3\left(w_\alpha+\frac{1}{4}+\frac{\Delta^\alpha}{2\pi\mathrm{i}}; \tau_\alpha\right)}{\vartheta_3\left(w_\alpha+\frac{1}{4}; \tau_\alpha\right)}
\frac{\vartheta_3\left(0; \tau_\alpha\right)}{\vartheta_3\left(\frac{\Delta^\alpha}{2\pi\mathrm{i}}; \tau_\alpha\right)}, \no\\
&P^\infty_{1, 2}(\lambda; \x, t)=\frac{\delta_\alpha-\delta_\alpha^{-1}}{2}
\frac{\vartheta_3\left(w_\alpha-\frac{1}{4}-\frac{\Delta^\alpha}{2\pi\mathrm{i}}; \tau_\alpha\right)}{\vartheta_3\left(w_\alpha-\frac{1}{4}; \tau_\alpha\right)}
\frac{\vartheta_3\left(0; \tau_\alpha\right)}{\vartheta_3\left(\frac{\Delta^\alpha}{2\pi\mathrm{i}}; \tau_\alpha\right)},\no\\
&P^\infty_{2, 1}(\lambda; \x, t)=\frac{\delta_\alpha-\delta_\alpha^{-1}}{2}
\frac{\vartheta_3\left(w_\alpha-\frac{1}{4}+\frac{\Delta^\alpha}{2\pi\mathrm{i}}; \tau_\alpha\right)}{\vartheta_3\left(w_\alpha-\frac{1}{4}; \tau_\alpha\right)}
\frac{\vartheta_3\left(0; \tau_\alpha\right)}{\vartheta_3\left(\frac{\Delta^\alpha}{2\pi\mathrm{i}}; \tau_\alpha\right)}, \no\\
&P^\infty_{2, 2}(\lambda; \x, t)=\frac{\delta_\alpha+\delta_\alpha^{-1}}{2}
\frac{\vartheta_3\left(w_\alpha+\frac{1}{4}-\frac{\Delta^\alpha}{2\pi\mathrm{i}}; \tau_\alpha\right)}{\vartheta_3\left(w_\alpha+\frac{1}{4}; \tau_\alpha\right)}
\frac{\vartheta_3\left(0; \tau_\alpha\right)}{\vartheta_3\left(\frac{\Delta^\alpha}{2\pi\mathrm{i}}; \tau_\alpha\right)},
\label{out-alpha}
\end{align}
%\end{gather}
where \( \delta_\alpha = \delta_\alpha(\lambda) \) is defined as a branch of
\[
\left( \frac{(\lambda + \alpha)(\lambda - \eta_2)}{(\lambda + \eta_2)(\lambda - \alpha)} \right)^{1/4}
\]
with branch cuts along \( [\alpha, \eta_2] \cup [-\eta_2, -\alpha] \) and normalized to \( \delta_\alpha = 1 + \mathcal{O}(\lambda^{-1}) \) as \( \lambda \to \infty \). Additionally, \( w_\alpha = w_\alpha(\lambda) \) is defined by
\begin{gather}
w_\alpha=-\frac{\eta_2}{4eK\left(m_\alpha\right)}\int_{\eta_2}^\lambda \frac{\mathrm{d}\zeta}{R(\zeta; \xi)}.
\end{gather}
For the local parametrix \( P^{\eta_2}(\lambda; \x, t) \), valid for \( \lambda \in (\mathrm{D}_1 \cup \mathrm{D}_\mathrm{o}) \cap B(\eta_2) \), we have:
\begin{gather} \label{parametrix-eta2}
P^{\eta_2}(\lambda; \x, t) = P^\infty(\lambda; \x, t) A^{\eta_2} C \zeta_{\eta_2}^{-\sigma_3/4} M^{\mathrm{mB}}(\zeta_{\eta_2}; \beta_2) e^{-\sqrt{\zeta_{\eta_2}} \sigma_3} (A^{\eta_2})^{-1},
\end{gather}
where $\zeta_{\eta_2}=t^2\mathcal{F}^{\eta_2}$, $A^{\eta_2}=\left(\mathrm{e}^{\pi\mathrm{i}/4/}/fd^{\eta_2}\right)^{\sigma_3}\sigma_2$, and $\mathcal{F}^{\eta_2}=p(\lambda; \xi)^2$ serves as the conformal mapping.
Additionally, for \( \lambda \in (\mathrm{D}_\mathrm{o}^\pm \cup \mathrm{D}_1^\pm) \cap B(\alpha) \), the local parametrix \( P^{\alpha}(\lambda; \x, t) \) can be expressed as:
\begin{gather} \label{parametrix-alpha}
P^{\alpha}(\lambda; \x, t) = P^\infty(\lambda; \x, t) A_{\pm}^{\alpha} C \zeta_{\alpha}^{-\sigma_3/4} M^{\mathrm{Ai}}(\zeta_{\alpha}) e^{\zeta_\alpha^{3/2} \sigma_3} (A_{\pm}^{\alpha})^{-1},
\end{gather}
with \( \zeta_{\alpha} = t^{2/3} \mathcal{F}_\pm^{\alpha} \), where \( \mathcal{F}^\alpha = \left( p \pm \frac{\Omega_\xi^\alpha}{2} \right)^{2/3} \), \( \Omega_\xi^\alpha = \frac{\Omega^\alpha(\xi + \frac{1}{\alpha \eta_2})}{4} \), and \( A_{\pm}^{\alpha} = \left( \frac{e^{\pi i/4 \mp t \Omega_\xi^\alpha /2}}{f \sqrt{r}} \right)^{\sigma_3} \sigma_1 \).

\begin{figure}[!t]
\centering
\includegraphics[scale=0.32]{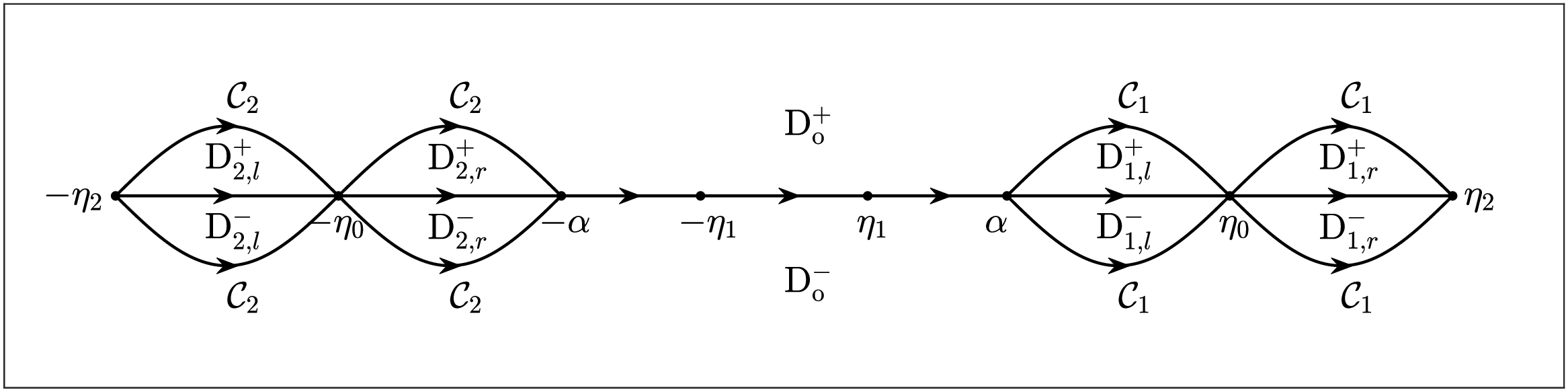}
\caption{Contour deformation in the region $\xi\in\left(\xi_\mathrm{crit}, \xi_0\right)$. }
\label{Deformation-3}
\end{figure}

As illustrated in Figure \ref{ErrorE1}, the error matrix \( E(\lambda; \x, t) \) satisfies the following Riemann-Hilbert problem. Specifically, \( E(\lambda; \x, t) \) is analytic in the region \( \lambda \in \mathbb{C} \setminus \mathcal{C}_1^c \cup \mathcal{C}_2^c \cup \partial B(\pm\alpha, \pm\eta_2) \cup [\eta_1, \alpha] \cup [-\alpha, -\eta_1] \) and is normalized to the identity matrix \( \mathbb{I}_2 \). The jump condition is expressed as:
\begin{gather}\label{E-Jump}
E_+(\lambda; \x, t) = E_-(\lambda; \x, t) V^E,
\end{gather}
where the jump matrix \( V^E \) is defined piecewise:
\begin{gather}
V^E=
\begin{cases}
P^\infty (\lambda)\mathcal{U}^{tp}_{f}\left[-ir^{-1}\right] \left(P^\infty (\lambda)\right)^{-1}, &\mathrm{for} \,\,\, \lambda\in\mathcal{C}^c_1,  \\
P^\infty (\lambda)\mathcal{L}^{tp}_{f}\left[-ir^{-1}\right] \left(P^\infty (\lambda)\right)^{-1}, &\mathrm{for} \,\,\, \lambda\in\mathcal{C}^c_2,  \\
P^{\eta_2}(\lambda)\left(P^\infty(\lambda)\right)^{-1}, &\mathrm{for} \,\,\, \lambda\in\partial B\left(\eta_2\right),  \\
P^{\alpha}(\lambda)\left(P^\infty(\lambda)\right)^{-1}, &\mathrm{for} \,\,\, \lambda\in\partial B\left(\alpha\right),   \\
\sigma_2P^{\eta_2}(-\lambda)\left(P^\infty\left(-\lambda\right)\right)^{-1}\sigma_2, &\mathrm{for} \,\,\, \lambda\in\partial B\left(-\eta_2\right), \\
\sigma_2P^{\alpha}\left(-\lambda\right)\left(P^\infty\left(-\lambda\right)\right)^{-1}\sigma_2, &\mathrm{for} \,\,\, \lambda\in\partial B\left(-\alpha\right),\\
\mathbb{I}_2+V_1^E, &\mathrm{for} \,\,\, \lambda\in\left(\eta_1, \eta_0\right)\cup\left(\eta_0, \alpha\right)\setminus \overline{B(\alpha)}\\
\mathbb{I}_2+V_2^E, &\mathrm{for} \,\,\, \lambda\in\left(-\alpha, -\eta_0\right)\cup\left(-\eta_0, -\eta_1\right)\setminus \overline{B(-\alpha)},
\end{cases}
\end{gather}
Here, \( V_1^E \) and \( V_2^E \) are given by:
\begin{gather}
\begin{aligned}
V^E_1 &= P^\infty_- f_-^{\sigma_3} \mathrm{e}^{tp_-\sigma_3} \left(\mathcal{L} \left[\mathrm{i}r\right] - \mathbb{I}_2\right) \mathrm{e}^{-tp_+\sigma_3} f_+^{-\sigma_3} \left(P^\infty_+\right)^{-1}, \\
V^E_2 &= P^\infty_- f_-^{\sigma_3} \mathrm{e}^{tp_-\sigma_3} \left(\mathcal{U} \left[\mathrm{i}r\right] - \mathbb{I}_2\right) \mathrm{e}^{-tp_+\sigma_3} f_+^{-\sigma_3} \left(P^\infty_+\right)^{-1}.
\end{aligned}
\end{gather}

\begin{proposition}[Small Norm Estimate in the Region \( \xi \in (\xi_0, -\eta_2^{-2}) \)]
For parameters \( \beta_2 > -1 \), \( \beta_1 \ge 0 \), and \( \beta_0 \ge 0 \), the jump matrices \( V^E \) exhibit the following small norm estimates:
\begin{gather}
\begin{aligned}
&\left\|V^E-\mathbb{I}_2\right\|_{L^1\cap L^2\cap L^\infty \left(\mathcal{C}^c\right)}=\mathcal{O}\left(\mathrm{e}^{-\mu t}\right), && \mathrm{as}\,\,\, t\to +\infty,  \\
&\left\|V^E-\mathbb{I}_2\right\|_{L^1\cap L^2\cap L^\infty \left(B\left(\pm\eta_2, \pm\alpha\right)\right)}=\mathcal{O}\left(t^{-1}\right), && \mathrm{as}\,\,\, t\to +\infty,
\end{aligned}
\end{gather}
These estimates yield the following result:
\begin{gather}
E(0; \x, t) = \mathbb{I}_2 + \mathcal{O}\left(t^{-1}\right), \quad E^{[1]}(\x, t) = \mathcal{O}\left(t^{-1}\right), \quad \mathrm{as}\,\,\, t \to +\infty,
\end{gather}
where \( E^{[1]}(\x, t) = \lim_{\lambda \to 0} \lambda^{-1} \left(E(\lambda; \x, 0) - \mathbb{I}_2\right) \) and \( \mathcal{C}^c = \mathcal{C}_1^c \cup \mathcal{C}_2^c \cup [\eta_1, \alpha] \cup [-\alpha, -\eta_1] \setminus B(\pm\alpha) \).
\end{proposition}

\begin{proof}
The proof follows from Proposition \ref{p-sign} along with the symmetry relation:
\begin{gather}
p(\lambda; \xi) = p(\lambda^*; \xi)^* = -p(-\lambda; \xi), \quad \mathrm{for}\,\,\,  \lambda \in \mathbb{C} \setminus [-\eta_2, \eta_2],
\end{gather}
which implies \( \Re(p) < 0 \) for \( \lambda \in \mathcal{C}_1 \setminus \{\alpha, \eta_2\} \), \( \Re(p) > 0 \) for \( \lambda \in \mathcal{C}_2 \setminus \{-\alpha, -\eta_2\} \), \( p_+ + p_- > 0 \) for \( \lambda \in [\eta_1, \alpha) \), and \( p_+ + p_- < 0 \) for \( \lambda \in (-\alpha, -\eta_1] \). These conditions lead to estimates for the jump matrix along \( \mathcal{C}^c \).
The estimate in \( B\left(\pm\eta_2, \pm\alpha\right) \) follows from:
\begin{gather}
\begin{aligned}
&P^{\eta_2}(\lambda)\left(P^\infty(\lambda)\right)^{-1} = \mathbb{I}_2 + \mathcal{O}\left(t^{-1}\right), && \mathrm{for}\,\,\, \lambda \in \partial B(\eta_2), \\
&P^{\alpha}(\lambda)\left(P^\infty(\lambda)\right)^{-1} = \mathbb{I}_2 + \mathcal{O}\left(t^{-1}\right), && \mathrm{for}\,\,\, \lambda \in \partial B(\alpha).
\end{aligned}
\end{gather}
This completes the proof.
\end{proof}

\begin{figure}[!t]
\centering
\includegraphics[scale=0.35]{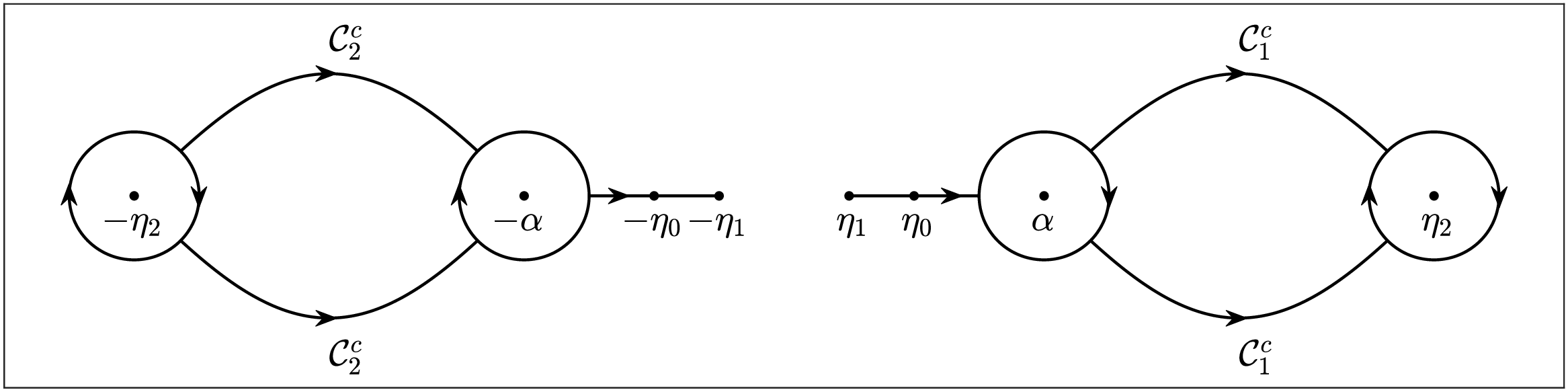}
\caption{Jump contour of the error matrix $E(\lambda; \x, t)$ in the region $\xi_0<\xi<-\eta_2^{-2}$}
\label{ErrorE1}
\end{figure}

In the interval \(\xi \in (\xi_\mathrm{crit}, \xi_0)\), we define the global parametrix \(P(\lambda; \x, t)\) as follows:
\begin{gather}
P(\lambda; \x, t)=
\begin{cases}
P^\infty\left(\lambda; \x, t\right), &\mathrm{for}\,\,\,\lambda\in\mathbb{C}\setminus\overline{B\left(\pm\eta_2, \pm\alpha, \pm \eta_0\right)},  \\
P^{\eta_2}\left(\lambda; \x, t\right), & \mathrm{for}\,\,\,\lambda\in B\left(\eta_2\right),  \\
P^{\alpha}\left(\lambda; \x, t\right), & \mathrm{for}\,\,\,\lambda\in B\left(\alpha\right),  \\
P^{\eta_0}\left(\lambda; \x, t\right), & \mathrm{for}\,\,\,\lambda\in B\left(\eta_0\right),  \\
\sigma_2P^{\eta_2}\left(-\lambda; \x, t\right)\sigma_2, & \mathrm{for}\,\,\,\lambda\in B\left(-\eta_2\right),  \\
\sigma_2P^{\alpha}\left(-\lambda; \x, t\right)\sigma_2, & \mathrm{for}\,\,\,\lambda\in B\left(-\alpha\right),\\
\sigma_2P^{\eta_0}\left(-\lambda; \x, t\right)\sigma_2, & \mathrm{for}\,\,\,\lambda\in B\left(-\eta_0\right),
\end{cases}
\end{gather}
The outer parametrix \(P^\infty(\lambda; \x, t)\) is specified in \eqref{out-alpha}, while the local parametrices \(P^{\eta_2}(\lambda; \x, t)\) and \(P^{\alpha}(\lambda; \x, t)\) are detailed in \eqref{parametrix-eta2} and \eqref{parametrix-alpha}, respectively.

We now present the local parametrix \(P^{\eta_0}(\lambda; \x, t)\). For the first type of generalized reflection coefficient \(r = r_0\) with \(\beta_0 \neq 0\), the conformal mapping is defined as \(\mathcal{F}^{\eta_0} = \pm(p - p_\pm(\eta_0))\) for \(\lambda \in B(\eta_0) \cap \mathbb{C}^\pm\). Following the asymptotic behavior of the initial value at large \(\x\), we establish the domains:
$\mathrm{D}_{\mathrm{o}, 1, r}^+=\left(\mathcal{F}^{\eta_0}\right)^{-1} \left(\mathrm{D}^\zeta_2\cap B^\zeta(0)\right)$,
 $\mathrm{D}_{\mathrm{o}, 1, l}^+=\left(\mathcal{F}^{\eta_0}\right)^{-1} \left(\mathrm{D}^\zeta_3\cap B^\zeta(0)\right)$,
 $\mathrm{D}_{\mathrm{o}, 1, l}^-=\left(\mathcal{F}^{\eta_0}\right)^{-1} \left(\mathrm{D}^\zeta_6\cap B^\zeta(0)\right)$,
and
 $\mathrm{D}_{\mathrm{o}, 1, r}^-=\left(\mathcal{F}^{\eta_0}\right)^{-1} \left(\mathrm{D}^\zeta_7\cap B^\zeta(0)\right)$.
In the vicinity of \(\lambda = \eta_0\), the local parametrix \(P^{\eta_0}(\lambda; \x, t)\) is constructed as follows:
 \begin{itemize}

 \item{} For \(\lambda \in \left(\mathrm{D}_{1, r}^+ \cup \mathrm{D}_{\mathrm{o}, 1, r}^+\right) \cap B(\eta_0)\),
\begin{gather}\label{eta0-1}
P^{\eta_0}(\lambda; \x, t) = P^\infty(\lambda; \x, t) A_{r+}^{\eta_0} \mathrm{e}^{\beta_0\pi \mathrm{i} \sigma_3/4} \left(-\mathrm{i} \sigma_2\right) M^{\mathrm{mb}}(\zeta_{\eta_0}; \beta_0) \mathrm{e}^{\zeta_{\eta_0} \sigma_3} \left(A_{r+}^{\eta_0}\right)^{-1},
\end{gather}
with \(A_{r+}^{\eta_0} = \left(\mathrm{e}^{\pi \mathrm{i}/4 + tp_{+}(\eta_0)}/fd_r^{\eta_0}\right)^{\sigma_3} \sigma_2\).

 \item{}  For \(\lambda \in \left(\mathrm{D}_{1, l}^+ \cup \mathrm{D}_{\mathrm{o}, 1, l}^+\right) \cap B(\eta_0)\), we have:
\begin{gather}
P^{\eta_0}(\lambda; \x, t) = P^\infty(\lambda; \x, t) A_{l+}^{\eta_0} \mathrm{e}^{-\beta_0 \pi \mathrm{i} \sigma_3/4} \left(-\mathrm{i} \sigma_2\right) M^{\mathrm{mb}}(\zeta_{\eta_0}; \beta_0) \mathrm{e}^{\zeta_{\eta_0} \sigma_3} \left(A_{l+}^{\eta_0}\right)^{-1},
\end{gather}
with \(A_{l+}^{\eta_0} = \left(\mathrm{e}^{\pi \mathrm{i}/4 + tp_{+}(\eta_0)}/fd_l^{\eta_0}\right)^{\sigma_3} \sigma_2\).

 \item{} For \(\lambda \in \left(\mathrm{D}_{1, l}^- \cup \mathrm{D}_{\mathrm{o}, 1, l}^-\right) \cap B(\eta_0)\), we express \(P^{\eta_0}\) as:
\begin{gather}
P^{\eta_0}(\lambda; \x, t) = P^\infty(\lambda; \x, t) A_{l-}^{\eta_0} \mathrm{e}^{\beta_0 \pi \mathrm{i} \sigma_3/4} M^{\mathrm{mb}}(\zeta_{\eta_0}; \beta_0) \mathrm{e}^{-\zeta_{\eta_0} \sigma_3} \left(A_{l-}^{\eta_0}\right)^{-1},
\end{gather}
with \(A_{l-}^{\eta_0} = \left(\mathrm{e}^{\pi \mathrm{i}/4 + tp_{-}(\eta_0)}/fd_l^{\eta_0}\right)^{\sigma_3} \sigma_2\).

 \item{} For \(\lambda \in \left(\mathrm{D}_{1, r}^- \cup \mathrm{D}_{\mathrm{o}, 1, r}^-\right) \cap B(\eta_0)\), we obtain:
\begin{gather}
P^{\eta_0}(\lambda; \x, t) = P^\infty(\lambda; \x, t) A_{r-}^{\eta_0} \mathrm{e}^{-\beta_0 \pi \mathrm{i} \sigma_3/4} M^{\mathrm{mb}}(\zeta_{\eta_0}; \beta_0) \mathrm{e}^{-\zeta_{\eta_0} \sigma_3} \left(A_{r-}^{\eta_0}\right)^{-1},
\end{gather}
with \(A_{r-}^{\eta_0} = \left(\mathrm{e}^{\pi \mathrm{i}/4 + tp_{-}(\eta_0)}/fd_r^{\eta_0}\right)^{\sigma_3} \sigma_2\) and \(\zeta_{\eta_0} = t \mathcal{F}^{\eta_0}\).

\end{itemize}

For the second type of generalized reflection coefficient \(r = r_c\), we define the conformal mapping as \(\mathcal{F}_0^{\eta_0} = \pm 2(p - p_{\pm}(\eta_0))\) for \(\lambda \in B(\eta_0) \cap \mathbb{C}^\pm\). The local parametrix in the vicinity of \(\lambda = \eta_0\) is constructed as follows:
 \begin{itemize}

 \item{} For \(\lambda \in \left(\mathrm{D}_1^+ \cup \mathrm{D}_\mathrm{o}^+\right) \cap B(\eta_0)\), it takes the form
\begin{gather}
P^{\eta_0}(\lambda; \x, t) = P^\infty(\lambda; \x, t) A^{\eta_0}_{+} \left(\zeta_{\eta_0}^{\kappa_0 \sigma_3} \mathrm{i} \sigma_2 \mathrm{e}^{\kappa_0 \pi \mathrm{i} \sigma_3}\right)^{-1} M^{\mathrm{CH}}(\zeta_{\eta_0}; \kappa_0) \mathrm{e}^{\zeta_{\eta_0} \sigma_3/2} \left(A^{\eta_0}_{+}\right)^{-1}.
\end{gather}

\item{} For \(\lambda \in \left(\mathrm{D}_1^- \cup \mathrm{D}_\mathrm{o}^-\right) \cap B(\eta_0)\), the local parametrix is expressed as
\begin{gather}\label{eta0-2}
P^{\eta_0}(\lambda; \x, t) = P^\infty(\lambda; \x, t) A^{\eta_0}_{-} \mathrm{e}^{-\kappa_0 \pi \mathrm{i} \sigma_3} M^{\mathrm{CH}}(\zeta_{\eta_0}; \kappa_0) \mathrm{e}^{-\zeta_{\eta_0} \sigma_3/2} \left(A^{\eta_0}_{-}\right)^{-1}.
\end{gather}
In the above two formulas, \(A^{\eta_0}_{\pm} = \left(\mathrm{e}^{\pi \mathrm{i}/4 + tp_{\pm}(\eta_0)}/fd^{\eta_0}\right)^{\sigma_3} \sigma_2\) and \(\zeta_{\eta_0} = t \mathcal{F}^{\eta_0}\).

\end{itemize}

\begin{RH} As illustrated in Figure \ref{ErrorE2}, the error matrix \(E(\lambda; \x, t)\) is governed by the following Riemann-Hilbert problem:

\begin{itemize}

\item{} \(E(\lambda; \x, t)\) is analytic for \(\lambda \in \mathbb{C} \setminus \mathcal{C}_1^c \cup \mathcal{C}_2^c \cup \partial B(\pm\alpha, \pm\eta_2) \cup [\eta_1, \alpha] \cup [-\alpha, -\eta_1]\);

\item{} It normalizes to the identity matrix \(\mathbb{I}_2\);

 \item{} The jump condition is detailed in \eqref{E-Jump}, with the jump matrix given by
\begin{gather}
V^E=
\begin{cases}
P^\infty (\lambda)\mathcal{U}^{tp}_{f}\left[-ir^{-1}\right] \left(P^\infty (\lambda)\right)^{-1}, &\mathrm{for} \,\,\, \lambda\in\mathcal{C}^c_1,  \\[0.5em]
P^\infty (\lambda)\mathcal{L}^{tp}_{f}\left[-ir^{-1}\right] \left(P^\infty (\lambda)\right)^{-1}, &\mathrm{for} \,\,\, \lambda\in\mathcal{C}^c_2,  \\[0.5em]
P^{\eta_j}(\lambda)\left(P^\infty(\lambda)\right)^{-1}, &\mathrm{for} \,\,\, \lambda\in\partial B\left(\eta_2\right), \,\, j=0,2, \\[0.5em]
%P^{\eta_0}(\lambda)\left(P^\infty(\lambda)\right)^{-1}, &\mathrm{for} \,\,\, \lambda\in\partial B\left(\eta_0\right),  \\[0.5em]
P^{\alpha}(\lambda)\left(P^\infty(\lambda)\right)^{-1}, &\mathrm{for} \,\,\, \lambda\in\partial B\left(\alpha\right),   \\[0.5em]
\sigma_2P^{\eta_j}(-\lambda)\left(P^\infty\left(-\lambda\right)\right)^{-1}\sigma_2, &\mathrm{for} \,\,\, \lambda\in\partial B\left(-\eta_2\right), \,\, j=0,2, \\[0.5em]
%\sigma_2P^{\eta_0}(-\lambda)\left(P^\infty\left(-\lambda\right)\right)^{-1}\sigma_2, &\mathrm{for} \,\,\, \lambda\in\partial B\left(-\eta_0\right), \\[0.5em]
\sigma_2P^{\alpha}\left(-\lambda\right)\left(P^\infty\left(-\lambda\right)\right)^{-1}\sigma_2, &\mathrm{for} \,\,\, \lambda\in\partial B\left(-\alpha\right),\\[0.5em]
\mathbb{I}_2+V_1^E, &\mathrm{for} \,\,\, \lambda\in\left(\eta_1, \alpha\right)\setminus \overline{B(\alpha)},\\[0.5em]
\mathbb{I}_2+V_2^E, &\mathrm{for} \,\,\, \lambda\in\left(-\alpha,  -\eta_1\right)\setminus \overline{B(-\alpha)}.
\end{cases}
\end{gather}
\end{itemize}
\end{RH}

\begin{proposition}[Small Norm Estimate in the Region \(\xi \in (\xi_\mathrm{crit}, \xi_0)\)]
For parameters satisfying \(\beta_2 > -1\), \(\beta_0 > -1\), and \(\beta_1 \geq 0\), the jump matrices \(V^E\) exhibit the following small norm estimates:
\begin{gather}
\begin{aligned}
&\left\|V^E-\mathbb{I}_2\right\|_{L^1\cap L^2\cap L^\infty \left(\mathcal{C}^c\right)}=\mathcal{O}\left(\mathrm{e}^{-\mu t}\right), && \mathrm{as}\,\,\, t\to +\infty  \\[1em]
&\left\|V^E-\mathbb{I}_2\right\|_{L^1\cap L^2\cap L^\infty \left(B\left(\pm\eta_2, \pm\eta_0, \pm\alpha\right)\right)}=\mathcal{O}\left(t^{-1}\right), && \mathrm{as}\,\,\, t\to +\infty
\end{aligned}
\end{gather}
This leads to the estimates
\begin{gather}
E(0; \x, t)=\mathbb{I}_2+\mathcal{O}\left(t^{-1}\right), \quad E^{[1]}(\x, t)=\mathcal{O}\left(t^{-1}\right), \quad \mathrm{as}\,\,\, t\to +\infty,
\end{gather}
where \(E^{[1]}(\x, t) = \lim_{\lambda \to 0} \lambda^{-1} \left(E(\lambda; \x, t) - \mathbb{I}_2\right)\), and \(\mathcal{C}^c = \mathcal{C}_1^c \cup \mathcal{C}_2^c \cup [\eta_1, \alpha] \cup [-\alpha, -\eta_1] \setminus B(\pm\alpha)\).
\end{proposition}
\begin{proof}
By invoking Proposition \ref{p-sign} alongside the symmetry relation
\begin{gather}
p(\lambda; \xi)=p(\lambda^*; \xi)^*=-p(-\lambda; \xi), \quad \mathrm{for}\,\,\, \lambda\in\mathbb{C}\setminus\left[-\eta_2, \eta_2\right],
\end{gather}
we derive that \(\Re(p) < 0\) for \(\lambda \in \mathcal{C}_1 \setminus \{\alpha, \eta_0, \eta_2\}\), \(\Re(p) > 0\) for \(\lambda \in \mathcal{C}_2 \setminus \{-\alpha, -\eta_0, -\eta_2\}\), \(p_+ + p_- > 0\) for \(\lambda \in [\eta_1, \alpha)\), and \(p_+ + p_- < 0\) for \(\lambda \in (-\alpha, -\eta_1]\). These conditions yield estimates for the jump matrix along \(\mathcal{C}^c\). The estimates on \(B(\pm\eta_2, \pm\eta_0, \pm\alpha)\) arise from the relations:
\begin{gather}
\begin{aligned}
&P^{\eta_2}(\lambda)\left(P^\infty(\lambda)\right)^{-1}=\mathbb{I}_2+\mathcal{O}\left(t^{-1}\right), && \mathrm{for}\,\,\, \lambda\in \partial B(\eta_2),\\[0.5em]
&P^{\alpha}(\lambda)\left(P^\infty(\lambda)\right)^{-1}=\mathbb{I}_2+\mathcal{O}\left(t^{-1}\right), && \mathrm{for}\,\,\, \lambda\in \partial B(\alpha), \\[0.5em]
&P^{\eta_0}(\lambda)\left(P^\infty(\lambda)\right)^{-1}=\mathbb{I}_2+\mathcal{O}\left(t^{-1}\right), && \mathrm{for}\,\,\, \lambda\in \partial B(\eta_0),
\end{aligned}
\end{gather}
Thus, we complete the proof.
\end{proof}

\begin{figure}[!t]
\centering
\includegraphics[scale=0.35]{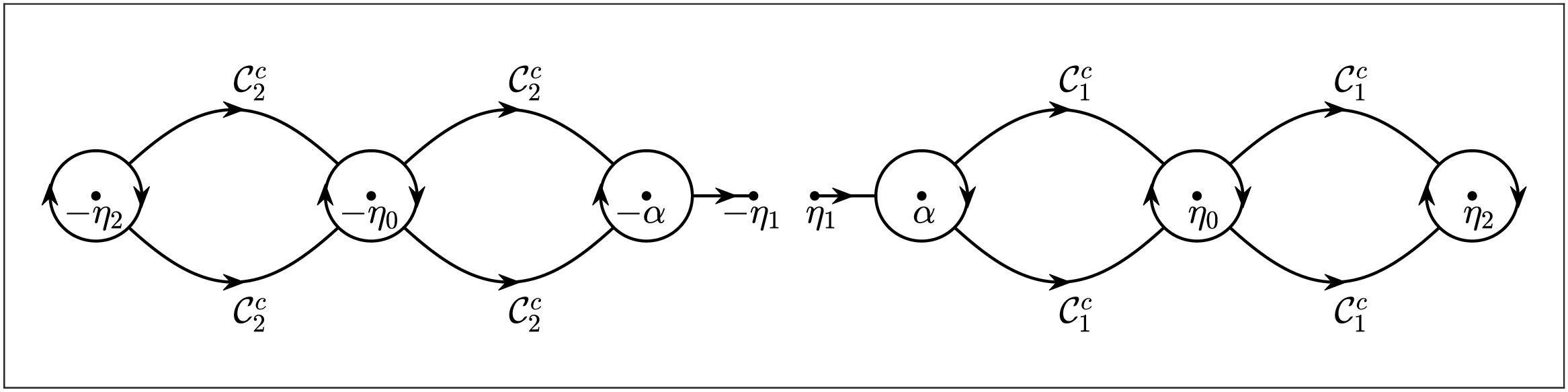}
\caption{Jump contour of the error matrix $E(\lambda; \x, t)$ in the region $\xi_\mathrm{crit}<\xi<\xi_0$}
\label{ErrorE2}
\end{figure}

In the region \(\xi < \xi_\mathrm{crit}\), the global parametrix \(P(\lambda; \x, t)\) is constructed as follows:
\begin{gather}
P(\lambda; \x, t)=
\begin{cases}
P^\infty\left(\lambda; \x, t\right), &\mathrm{for}\,\,\,\lambda\in\mathbb{C}\setminus\overline{B\left(\pm\eta_2, \pm\eta_1, \pm \eta_0\right)},  \\[0.5em]
P^{\eta_j}\left(\lambda; \x, t\right), & \mathrm{for}\,\,\,\lambda\in B\left(\eta_j\right),\,\, j=0,1,2, \\[0.5em]
%P^{\eta_1}\left(\lambda; \x, t\right), & \mathrm{for}\,\,\,\lambda\in B\left(\eta_1\right),  \\[0.5em]
%P^{\eta_0}\left(\lambda; \x, t\right), & \mathrm{for}\,\,\,\lambda\in B\left(\eta_0\right),  \\[0.5em]
\sigma_2P^{\eta_j}\left(-\lambda; \x, t\right)\sigma_2, & \mathrm{for}\,\,\,\lambda\in B\left(-\eta_j\right), \,\, j=0,1,2. %\\[0.5em]
%\sigma_2P^{\eta_1}\left(-\lambda; \x, t\right)\sigma_2, & \mathrm{for}\,\,\,\lambda\in B\left(-\eta_1\right),\\[0.5em]
%\sigma_2P^{\eta_0}\left(-\lambda; \x, t\right)\sigma_2, & \mathrm{for}\,\,\,\lambda\in B\left(-\eta_0\right).
\end{cases}
\end{gather}
The outer parametrix \(P^\infty(\lambda; \x, t) = \left(P_{i,j}^\infty(\lambda; \x, t)\right)_{2\times 2}\) is given by:
\begin{gather}
\begin{aligned}
P^\infty_{1, 1}(\lambda; \x, t)&=\frac{\delta_1+\delta_1^{-1}}{2}
\frac{\vartheta_3\left(w_1+\frac{1}{4}+\frac{\Delta_1}{2\pi\mathrm{i}}; \tau_1\right)}{\vartheta_3\left(w_1+\frac{1}{4}; \tau_1\right)}
\frac{\vartheta_3\left(0; \tau_1\right)}{\vartheta_3\left(\frac{\Delta_1}{2\pi\mathrm{i}}; \tau_1\right)}, \\[0.5em]
P^\infty_{1, 2}(\lambda; \x, t)&=\frac{\delta_1-\delta_1^{-1}}{2}
\frac{\vartheta_3\left(w_1-\frac{1}{4}-\frac{\Delta_1}{2\pi\mathrm{i}}; \tau_1\right)}{\vartheta_3\left(w_1-\frac{1}{4}; \tau_1\right)}
\frac{\vartheta_3\left(0; \tau_1\right)}{\vartheta_3\left(\frac{\Delta_1}{2\pi\mathrm{i}}; \tau_1\right)},\\[0.5em]
P^\infty_{2, 1}(\lambda; \x, t)&=\frac{\delta_1-\delta_1^{-1}}{2}
\frac{\vartheta_3\left(w_1-\frac{1}{4}+\frac{\Delta_1}{2\pi\mathrm{i}}; \tau_1\right)}{\vartheta_3\left(w_1-\frac{1}{4}; \tau_1\right)}
\frac{\vartheta_3\left(0; \tau_1\right)}{\vartheta_3\left(\frac{\Delta_1}{2\pi\mathrm{i}}; \tau_1\right)}, \\[0.5em]
P^\infty_{2, 2}(\lambda; \x, t)&=\frac{\delta_1+\delta_1^{-1}}{2}
\frac{\vartheta_3\left(w_1+\frac{1}{4}-\frac{\Delta_1}{2\pi\mathrm{i}}; \tau_1\right)}{\vartheta_3\left(w_1+\frac{1}{4}; \tau_1\right)}
\frac{\vartheta_3\left(0; \tau_1\right)}{\vartheta_3\left(\frac{\Delta_1}{2\pi\mathrm{i}}; \tau_1\right)}.
\end{aligned}
\end{gather}
For \(\lambda \in (\mathrm{D}_{1}^\pm \cup \mathrm{D}_\mathrm{o}^\pm) \cap B(\eta_1)\), the local parametrix \(P^{\eta_1}(\lambda; \x, t)\) is defined as:
\begin{gather}
P^{\eta_1}(\lambda; \x, t)=P^\infty(\lambda; \x, t)A_{\pm}^{\eta_1}C\zeta_{\eta_1}^{-\sigma_3/4}M^{\mathrm{mB}}(\zeta_{\eta_1}; \beta_1)\mathrm{e}^{-\sqrt{\zeta_{\eta_1}}\sigma_3}\left(A_{\pm}^{\eta_1}\right)^{-1},
\end{gather}
where $\zeta_{\eta_1}=t^2\mathcal{F}^{\eta_1}$, $A_{\pm}^{\eta_1}=\left(\mathrm{e}^{\pi\mathrm{i}/4\mp t\hat{\Omega}_1 /2}/fd^{\eta_1}\right)^{\sigma_3}\sigma_1$,
$\mathcal{F}^{\eta_1}=\left(p\pm\hat{\Omega}_1/2\right)^2$ for $\lambda\in B(\eta_1)\cap\mathbb{C}^\pm$, and $\hat{\Omega}_1=\Omega_1\left(\xi+1/\eta_1\eta_2\right)/4$
The local parametrices \(P^{\eta_2}(\lambda; \x, t)\) and \(P^{\eta_0}(\lambda; \x, t)\) are defined in the respective equations \eqref{parametrix-eta2} and (\ref{eta0-1})-(\ref{eta0-2}).

The Riemann-Hilbert problem for the error matrix \(E(\lambda; \x, t)\) shares the same jump contour as shown in Figure \ref{ErrorE0}. The matrix \(E(\lambda; \x, t)\) is analytic for \(\lambda \in \mathbb{C} \setminus \mathcal{C}_1^c \cup \mathcal{C}_2^c \cup \partial B(\pm\eta_1, \pm\eta_2, \pm\eta_0)\) and normalizes to the identity matrix \(\mathbb{I}_2\). The jump condition is given in \eqref{E-Jump} with the jump matrix defined as follows:
\begin{gather}
V^E=
\begin{cases}
P^\infty (\lambda)\mathcal{U}^{tp}_{f}\left[-ir^{-1}\right] \left(P^\infty (\lambda)\right)^{-1}, &\mathrm{for} \,\,\, \lambda\in\mathcal{C}^c_1,  \\[0.5em]
P^\infty (\lambda)\mathcal{L}^{tp}_{f}\left[-ir^{-1}\right] \left(P^\infty (\lambda)\right)^{-1}, &\mathrm{for} \,\,\, \lambda\in\mathcal{C}^c_2,  \\[0.5em]
P^{\eta_j}(\lambda)\left(P^\infty(\lambda)\right)^{-1}, &\mathrm{for} \,\,\, \lambda\in\partial B\left(\eta_j\right), \,\, j=0,1,2,  \\[0.5em]
%P^{\eta_0}(\lambda)\left(P^\infty(\lambda)\right)^{-1}, &\mathrm{for} \,\,\, \lambda\in\partial B\left(\eta_0\right),  \\[0.5em]
%P^{\eta_1}(\lambda)\left(P^\infty(\lambda)\right)^{-1}, &\mathrm{for} \,\,\, \lambda\in\partial B\left(\eta_1\right),   \\[0.5em]
\sigma_2P^{\eta_j}(-\lambda)\left(P^\infty\left(-\lambda\right)\right)^{-1}\sigma_2, &\mathrm{for} \,\,\, \lambda\in\partial B\left(-\eta_j\right), \,\, j=0,1,2.
%\sigma_2P^{\eta_0}(-\lambda)\left(P^\infty\left(-\lambda\right)\right)^{-1}\sigma_2, &\mathrm{for} \,\,\, \lambda\in\partial B\left(-\eta_0\right), \\[0.5em]
%\sigma_2P^{\eta_1}\left(-\lambda\right)\left(P^\infty\left(-\lambda\right)\right)^{-1}\sigma_2, &\mathrm{for} \,\,\, \lambda\in\partial B\left(-\eta_1\right).
\end{cases}
\end{gather}

\begin{proposition}[Estimate of Small Norm in the Region $\xi < \xi_{\mathrm{crit}}$]
Let $\beta_2, \beta_1, \beta_0 > -1$. The jump matrices $V^E$ satisfy the following small norm estimates:
\begin{gather}
\begin{aligned}
&\left\|V^E-\mathbb{I}_2\right\|_{L^1\cap L^2\cap L^\infty \left(\mathcal{C}_1^c\cup\mathcal{C}_2^c\right)}=\mathcal{O}\left(\mathrm{e}^{-\mu t}\right), && \mathrm{as}\,\,\, t\to +\infty  \\[1em]
&\left\|V^E-\mathbb{I}_2\right\|_{L^1\cap L^2\cap L^\infty \left(B\left(\pm\eta_2, \pm\eta_0, \pm\eta_1 \right)\right)}=\mathcal{O}\left(t^{-1}\right), && \mathrm{as}\,\,\, t\to +\infty
\end{aligned}
\end{gather}
Consequently, we obtain:
\begin{gather}
E(0; \x, t)=\mathbb{I}_2+\mathcal{O}\left(t^{-1}\right), \quad E^{[1]}(\x, t)=\mathcal{O}\left(t^{-1}\right), \quad \text{as}\,\,\, t\to +\infty,
\end{gather}
\end{proposition}

\begin{proof}
This result follows from Proposition \ref{p-sign-1} and the symmetry relation:
\begin{gather}
p(\lambda; \xi) = p(\lambda^*; \xi)^* = -p(-\lambda; \xi), \quad \mathrm{for}\,\,\, \lambda \in \mathbb{C} \setminus \left[-\eta_2, \eta_2\right].
\end{gather}
Here, we observe that $\Re(p) < 0$ for $\lambda \in \mathcal{C}_1 \setminus \left\{\eta_1, \eta_0, \eta_2\right\}$ and $\Re(p) > 0$ for $\lambda \in \mathcal{C}_2 \setminus \left\{-\eta_1, -\alpha, -\eta_2\right\}$, leading to the estimates for the jump matrix on $\mathcal{C}_1^c \cup \mathcal{C}_2^c$. The estimates on the boundary $B\left(\pm\eta_2, \pm\eta_0, \pm\eta_1\right)$ are derived from:
\begin{gather}
\begin{aligned}
&P^{\eta_2}(\lambda)\left(P^\infty(\lambda)\right)^{-1} = \mathbb{I}_2 + \mathcal{O}\left(t^{-1}\right), && \mathrm{for}\,\,\, \lambda \in \partial B(\eta_2), \\[1em]
&P^{\eta_0}(\lambda)\left(P^\infty(\lambda)\right)^{-1} = \mathbb{I}_2 + \mathcal{O}\left(t^{-1}\right), && \mathrm{for}\,\,\, \lambda \in \partial B(\eta_0), \\[1em]
&P^{\eta_1}(\lambda)\left(P^\infty(\lambda)\right)^{-1} = \mathbb{I}_2 + \mathcal{O}\left(t^{-1}\right), && \mathrm{for}\,\,\, \lambda \in \partial B(\eta_1).
\end{aligned}
\end{gather}
Thus, the proof is completed.
\end{proof}

By recalling the transformations \( Y(\lambda; \x, t) \mapsto T(\lambda; \x, t) \mapsto S(\lambda; \x, t) \mapsto E(\lambda; \x, t) \) along with the propositions regarding small norm estimates, we conclude:
\begin{gather}
\begin{array}{rl}
u(x, t) =& f_{\pm}^{-2}(0; \xi) \mathrm{e}^{2t g_{\pm}(0; \xi)} \left(P_{\pm}^\infty(0; \x, t)^{-1} P_{\pm}^{\infty, [1]}(\x, t)\right)_{1, 2} + \mathcal{O}\left(|\x|^{-1}\right), \\[1em]
x =& \x + \left(P_{\pm}^\infty(0; \x, t)^{-1} P_{\pm}^{\infty, [1]}(\x, t)\right)_{1, 1} + f_{\pm}^{-1}(0; \xi) f_{\pm}^{[1]}(\xi) - t g_{\pm}^{[1]}(\xi) + \mathcal{O}\left(|\x|^{-1}\right),
\end{array}
\end{gather}
where
\begin{gather}
\begin{array}{rl}
P_{\pm}^{\infty, [1]}(\x, t) = & \d\lim_{\lambda \in \mathbb{C}^\pm \to 0} \lambda^{-1}\left(P^\infty(\lambda; \x, t) - P_\pm^\infty(0; \x, t)\right), \\[1em]
f_{\pm}^{[1]}(\xi) =& \d\lim_{\lambda \in \mathbb{C}^\pm \to 0} \lambda^{-1}\left(f(\lambda; \xi) - f_{\pm}(0; \xi)\right), \\[1em]
g_{\pm}^{[1]}(\xi) =& \d\lim_{\lambda \in \mathbb{C}^\pm \to 0} \lambda^{-1}\left(g(\lambda; \xi) - g_{\pm}(0; \xi)\right).
\end{array}
\end{gather}
A straightforward computation yields \eqref{large-middle} and \eqref{large-left}.

\section{Conclusions and discussions}

In conclusion, we have investigated  the asymptotic characteristics of short-pulse soliton gases, framed within the Riemann-Hilbert problem through the lens of two distinct generalized reflection coefficients. Specifically, we analyze the first reflection coefficient $r_0(\lambda)$
%\[r_0 = (\lambda - \eta_1)^{\beta_1}(\eta_2 - \lambda)^{\beta_2}\left|\lambda - \eta_0\right|^{\beta_0}\gamma(\lambda)\]
and the second reflection coefficient $r_c(\lambda)$.
%\[r_c = (\lambda - \eta_1)^{\beta_1}(\eta_2 - \lambda)^{\beta_2}\chi_c(\lambda).\]
A pivotal aspect of our investigation is the construction of the \(g\)-function. The short pulse can be described as a negative flow of the WKI system, exhibiting a phase \(\theta\) that significantly diverges from those observed in the KdV and modified KdV equations. To address the singularity present at the origin, we must develop an appropriate \(g\)-function.

Additionally, we focus on the construction of local parametrices. For our analysis, we employ the Airy parametrix alongside the first type of modified Bessel parametrix to appropriately capture the local behavior near the endpoints. Notably, the local parametrix around the singularity \(\eta_0\) varies depending on the type of reflection coefficient employed. For the first type of generalized reflection coefficient, we utilize the second type of modified Bessel parametrix, while for the second type, we implement the confluent hypergeometric parametrix.

The long-time asymptotic behavior of SP soliton gases defined by these generalized forms, \(r_0\) and \(r_c\), can be determined by adhering to the methodologies presented in this paper, while also incorporating local parametrices around each singularity \(\eta_{0, j}\). Specifically, we employ the second type of modified Bessel parametrix for the first case \(r_0\) and the confluent hypergeometric parametrix for the second case \(r_c\).

Nevertheless, several challenges persist in this research domain. As highlighted by Girotti {\it et al}~\cite{21}, establishing rigorous asymptotics for soliton gases influenced by two nontrivial reflection coefficients remains a complex undertaking. Moreover, the limit process when discrete spectra cluster in disconnected segments of the imaginary axis introduces further complications. For additional significant open problems, we refer to the review by Suret {\it et al}~\cite{20}, which encapsulates various fundamental inquiries prompted by intriguing theoretical and experimental challenges.

\v

%\section*{Acknowledgements}

\addcontentsline{toc}{section}{Founding}

\v\v\noindent{\bf Founding} \v

\noindent This work of G.Z. was supported by the National Natural Science Foundation of China (Grant No. 12201615).
This work of Z.Y. was supported by the National Natural Science Foundation of China (Grant No.11925108).





%\begin{thebibliography}{00}


%\bibitem{}

%\end{thebibliography}

%\bibliographystyle{elsarticle-num}
%\bibliographystyle{elsarticle-num-names}
%\bibliography{short-pulse}

\addcontentsline{toc}{section}{References}

\begin{thebibliography}{100}\setlength{\itemsep}{-0.8mm}


\bibitem{soliton-book91} M. J. Ablowitz, P. A. Clarkson, {\it Solitons, Nonlinear Evolution Equations and Inverse Scattering} (Cambridge University Press, Cambridge, 1991).

\bibitem{37} M. J. Ablowitz, D. J. Kaup, A. C. Newell, H. Segur, Method for solving the sine-Gordon equation,
Phys. Rev. Lett. 30 (1973) 1262–1264. %doi:10.1103/physrevlett.30.1262.

\bibitem{32} N. Akhmediev, A. Ankiewicz, J. M. Soto-Crespo, Rogue waves and rational solutions of the nonlinear
Schr\"odinger equation, Phys. Rev. E 80 (2009) 026601. %doi:10.1103/physreve.80.026601.


\bibitem{15} J. Baik, P. A. Deift, K. Johansson, On the distribution of the length of the longest increasing subsequence
of random permutations, J. Amer. Math. Soc. 12 (1999) 1119–1178. %doi:10.1090/S0894-0347-99-00307-0.

\bibitem{53} F. Balogh, M. Bertola, S.-Y. Lee, K. D. T.-R. Mclaughlin, Strong asymptotics of the orthogonal
polynomials with respect to a measure supported on the plane, Commun. Pure Appl. Math. 68 (2015)
112–172. %doi:https://doi.org/10.1002/cpa.21541.

\bibitem{19} M. Bertola, T. Grava, G. Orsatti, Soliton shielding of the focusing nonlinear Schr\"odinger equation,
Phys. Rev. Lett. 130 (2023) 127201. %doi:10.1103/PhysRevLett.130.127201, pRL.



\bibitem{39} D. Bilman, R. Buckingham, Large-order asymptotics for multiple-pole solitons of the focusing nonlinear
Schr\"odinger equation, J. Nonlinear Sci. 29 (2019) 2185–2229. % doi:10.1007/s00332-019-09542-7.


\bibitem{30}D. Bilman, L. Ling, P. D. Miller, Extreme superposition: Rogue waves of infinite order and the
Painlev\'e-III hierarchy, Duke Math. J. 169 (2020) 671-760. %doi:10.1215/00127094-2019-0066.


\bibitem{31} D. Bilman, P. D. Miller, A robust inverse scattering transform for the focusing nonlinear Schr\"odinger
equation, Comm. Pure Appl. Math. 72 (2019) 1722–1805. %doi:10.1002/cpa.21819.

\bibitem{38} G. Biondini, G. Kova$\check{c}$i$\check{c}$, Inverse scattering transform for the focusing nonlinear
  Schr\"odinger equation with nonzero boundary conditions, J. Math. Phys. 55 (2014) 031506. %doi:10.1063/1.4868483.

\bibitem{28} G. Biondini, S. Li, D. Mantzavinos, Long-time asymptotics for the focusing nonlinear Schr\"odinger
equation with nonzero boundary conditions in the presence of a discrete spectrum, Commun. Math.
Phys. 382 (2021) 1495–1577. %doi:10.1007/s00220-021-03968-5.


\bibitem{27} G. Biondini, D. Mantzavinos, Long-time asymptotics for the focusing nonlinear Schr\"odinger equation with nonzero boundary conditions at infinity and asymptotic stage of modulational instability,
Comm. Pure Appl. Math. 70 (2017) 2300–2365. % doi:10.1002/cpa.21701.



\bibitem{58} A. I. Bobenko, A. Its, The asymptotic behavior of the discrete holomorphic map 
$Z^a$ via the Riemann–Hilbert method, Duke Math. J. 165 (2016) 2607–2682. %doi:10.1215/00127094-3620012.

\bibitem{57} A. Bogatskiy, T. Claeys, A. Its, Hankel determinant and orthogonal polynomials for a Gaussian weight
with a discontinuity at the edge, Commun. Math. Phys. 347 (2016) 127–162. % doi:10.1007/s00220-016-2691-y.


\bibitem{62} A.  Boutet de Monvel, A. Its, V. Kotlyarov, Long-time asymptotics for the focusing NLS equation
with time-periodic boundary condition on the half-line, Commun. Math. Phys. 290 (2009) 479–522.
%doi:10.1007/s00220-009-0848-7.

\bibitem{22} A. Boutet de Monvel, D. Shepelsky, L. Zielinski, The short pulse equation by a Riemann–Hilbert
approach, Lett. Math. Phys. 107 (2017) 1345–1373. %doi:10.1007/s11005-017-0945-z.


\bibitem{Brun1} J. C. Brunelli, The short pulse equation hierarchy, J. Math. Phys. 46 (2005) 123507.

\bibitem{Brun2} J. C. Brunelli, The bi-Hamiltonian structure of the short pulse equation, Phys. Lett. A 353 (2006) 475.




\bibitem{25} R. Buckingham, S. Venakides, Long-time asymptotics of the nonlinear Schr\"odinger equation shock
problem, Commun. Pure Appl. Math. 60 (2007) 1349–1414. % doi:10.1002/cpa.20179.

\bibitem{9} V. B. Bulchandani, On classical integrability of the hydrodynamics of quantum integrable systems, J.
Phys. A: Math. Theor. 50 (2017) 435203. %doi:10.1088/1751-8121/aa8c62.


\bibitem{cau-1} G. M. Coclite, L. di Ruvo, Well-posedness results for the short pulse equation, Z. Angew. Math. Phys. 66 (2015)
1529–1557.

\bibitem{5} T. Congy, G. A. El, G. Roberti, Soliton gas in bidirectional dispersive hydrodynamics, Phys. Rev. E 103
(2021) 042201. %doi:10.1103/physreve.103.042201.


\bibitem{46} T. O. Conway, P. A. Deift, Asymptotics of polynomials orthogonal with respect to a logarithmic weight,
SIGMA  14 (2018) 056. %doi:10.3842/sigma.2018.056.


\bibitem{56} A. Dea$\tilde{n}$o, A. B. J. Kuijlaars, P. Rom\'an, Asymptotics of matrix valued orthogonal polynomials on [-1,
1], Adv. Math. 423 (2023) 109043. %doi:https://doi.org/10.1016/j.aim.2023.109043.

\bibitem{50} P. A. Deift, D. Gioev, T. Kriecherbauer, M. Vanlessen, Universality for orthogonal and symplectic
Laguerre-type ensembles, J. Stat. Phys. 129 (2007) 949–1053. %doi:10.1007/s10955-007-9325-x.


\bibitem{29} P. A. Deift, S. Kamvissis, T. Kriecherbauer, X. Zhou, The Toda rarefaction problem, Comm. Pure
Appl. Math. 49 (1996) 35–83. %doi:10.1002/(sici)1097-0312(199601)49:1¡35::aid-cpa2¿3.0.co;2-8.

\bibitem{24} P. A. Deift, T. Kriecherbauer, K. T.-R. McLaughlin, S. Venakides, X. Zhou, Uniform asymptotics
for polynomials orthogonal with respect to varying exponential weights and applications to universality questions in random matrix theory, Comm. Pure Appl. Math. 52 (1999) 1335–1425.
%doi:10.1002/(sici)1097-0312(199911)52:11¡1335::aid-cpa1¿3.0.co;2-1.


\bibitem{43} P. A. Deift, T. Kriecherbauer, K. T. McLaughlin, S. Venakides, X. Zhou, Strong asymptotics of orthogonal
polynomials with respect to exponential weights, Comm. Pure Appl. Math. 52 (1999) 1491–1552.

\bibitem{47} P. A. Deift, M. Piorkowski, Recurrence coefficients for orthogonal polynomials with a logarithmic weight
function, SIGMA 20 (2024) 004. %doi:10.3842/sigma.2024.004.

\bibitem{16} P. A. Deift, S. Venakides, X. Zhou, New results in small dispersion KdV by an extension of the steepest descent method for Riemann-Hilbert problems, Int. Math. Res. Not. 1997 (1997) 285–299.
%doi:10.1155/s1073792897000214.

\bibitem{14} P. A. Deift, X. Zhou, A steepest descent method for oscillatory Riemann–Hilbert problems. Asymptotics
for the MKdV equation, Ann. Math. 137 (1993) 295–368. %doi:10.2307/2946540.


\bibitem{17} P. A. Deift, X. Zhou, Asymptotics for the Painlev\'e II equation, Comm. Pure Appl. Math. 48 (1995)
277–337. %doi:10.1002/cpa.3160480304.


\bibitem{13} S. Dyachenko, D. Zakharov, V. E. Zakharov, Primitive potentials and bounded solutions of the KdV
equation, Physica D 333 (2016) 148–156. %doi:10.1016/j.physd.2016.04.002.

\bibitem{23} I. Egorova, Z. Gladka, V. Kotlyarov, G. Teschl, Long-time asymptotics for the Korteweg–de Vries equation with step-like initial data, Nonlinearity 26 (2013) 1839–1864. %doi:10.1088/0951-7715/26/7/1839.



\bibitem{2} G. A. El, The thermodynamic limit of the Whitham equations, Phys. Lett. A 311 (2003) 374–383.
%doi:10.1016/s0375-9601(03)00515-2.

\bibitem{3} G. A. El, A. M. Kamchatnov, Kinetic equation for a dense soliton gas, Phys. Rev. Lett. 95 (2005)
204101. % doi:10.1103/PhysRevLett.95.204101.

\bibitem{7} G. A. El, A. M. Kamchatnov, M. V. Pavlov, S. A. Zykov, Kinetic equation for a soliton gas and its hydrodynamic reductions, J. Nonlinear Sci. 21 (2010) 151–191. %doi:10.1007/s00332-010-9080-z.


\bibitem{4} G. A. El, A. Tovbis, Spectral theory of soliton and breather gases for the focusing nonlinear Schr\"odinger
equation, Phys. Rev. E 101 (2020) 052207. %doi:10.1103/physreve.101.052207.


\bibitem{gsp2} B. F. Feng, Complex short pulse and coupled complex short pulse equations. Physica D 297 (2015) 62–75.

\bibitem{gsp8} B. F. Feng, L. Ling, and Z. Zhu, Defocusing complex short-pulse equation and its multi-dark-soliton solution,
Phys. Rev. E 93 (2016) 052227.

\bibitem{dis} B. F. Feng, K. Maruno, Y. Ohta, Integrable discretizations of the short pulse equation, J. Phys. A: Math. Theor. 43 (2010) 085203.

\bibitem{6}  E. Ferapontov, M. V. Pavlov, Kinetic equation for soliton gas: Integrable reductions, J. Nonlinear Sci.
32 (2022) 26.

\bibitem{52} G. Filipuk, W. Van Assche, L. Zhang, The recurrence coefficients of semi-classical Laguerre polynomials and the fourth Painlev\'e equation, J. Phys. A: Math. Theor. 45 (2012) 205201.% doi:10.1088/1751-8113/45/20/205201.

\bibitem{44} A. S. Fokas, A. R. Its, A. V. Kitaev, Discrete Painlev\'e equations and their appearance in quantum
gravity, Commun. Math. Phys. 142 (1991) 313–344. % doi:10.1007/bf02102066.

\bibitem{45} A. S. Fokas, A. R. Its, A. V. Kitaev, The isomonodromy approach to matric models in 2D quantum
gravity, Commun. Math. Phys. 147 (1992) 395–430. %doi:10.1007/bf02096594.

\bibitem{mch-1} B. Fuchssteiner, A.S. Fokas, Symplectic structures, their Bäcklund transformations and hereditary symmetries, Physica D 4 (1981) 47-66.


\bibitem{GGKM} C. S. Gardner, J. M. Greene, M. D. Kruskal, R. M. Miura, Method for solving the Korteweg-de Vries
equation, Phys. Rev. Lett. 19 (1967) 1095–1097.


\bibitem{12} A. A. Gelash, D. S. Agafontsev, Strongly interacting soliton gas and formation of rogue waves, Phys.
Rev. E 98 (2018) 042210. %doi:10.1103/physreve.98.042210.

\bibitem{10} A. A. Gelash, D. S. Agafontsev, V. E. Zakharov, G. A. El, S. Randoux, P. Suret, Bound state soliton gas dynamics underlying the spontaneous modulational instability, Phys. Rev. Lett. 123 (2019) 234102.
%doi:10.1103/physrevlett.123.234102.

\bibitem{21} M. Girotti, T. Grava, R. Jenkins, K. D. T. R. McLaughlin, Rigorous asymptotics of a KdV soliton gas,
Commun. Math. Phys. 384 (2021) 733–784. %doi:10.1007/s00220-021-03942-1.


\bibitem{18} M. Girotti, T. Grava, R. Jenkins, K. T. McLaughlin, A. Minakov, Soliton versus the gas: Fredholm
determinants, analysis, and the rapid oscillations behind the kinetic equation, Comm. Pure Appl.
Math. 76 (2023) 3233–3299. %doi:https://doi.org/10.1002/cpa.22106.


\bibitem{26} T. Grava, A. Minakov, On the long-time asymptotic behavior of the modified Korteweg–de Vries equation with step-like initial data, SIAM J. Math. Anal. 52 (2020) 5892–5993. % doi:10.1137/19m1279964.

\bibitem{Gu} C. Gu, H. Hu, Z. Zhou, {\it Darboux Transformations in Integrable Systems
Theory and their Applications to Geometry} (Springer, New York, 2005).


\bibitem{33} B. Guo, L. Ling, Q. P. Liu, Nonlinear Schr\"odinger equation: Generalized Darboux transformation and
rogue wave solutions, Phys. Rev. E 85 (2012) 026607. %doi:10.1103/physreve.85.026607.


\bibitem{Hase} A. Hasegawa, Y. Kodama, {\it  Solitons in Optical Communications} (Oxford University Press, Oxford, 1995).

%\bibitem{hirota} R. Hirota, Exact solution of the Korteweg-de Vries equation for multiple collisions of solitons, Phys.
%Rev. Lett. 27 (1971) 1192–1194.

\bibitem{35} R. Hirota, Exact solution of the Korteweg-de Vries equation for multiple collisions of solitons, Phys.
Rev. Lett. 27 (1971) 1192–1194. %doi:10.1103/physrevlett.27.1192.


\bibitem{gsp} A. N. W. Hone, V. Novikov, J. P. Wang, Generalizations of the short pulse equation, Lett. Math. Phys. 108 (2018) 927-947.

\bibitem{op1} A. Its, I. Krasovsky, Hankel determinant and orthogonal polynomials for the Gaussian weight with a jump. In: Integrable systems and random matrices. Vol. 458, Contemp. Math. Amer. Math. Soc., Providence, RI (2008) 215–247.

%\bibitem{61} A. Its, I. Krasovsky, Hankel determinant and orthogonal polynomials for the Gaussian weight with a
%jump, Contemp. Math. (2008) 215–247. %doi:10.1090/conm/458/08938.


\bibitem{ko} K. Konno, H. Oono, New coupled integrable dispersionless equations, J. Phys. Soc. Jpn. 63 (1994) 377.

\bibitem{63} V. Kotlyarov, A. Minakov, Riemann–Hilbert problem to the modified Korteveg–de Vries equation:
Long-time dynamics of the steplike initial data, J. Math. Phys. 51 (2010)  093506. %doi:10.1063/1.3470505


\bibitem{48} T. Kriecherbauer, K. T.-R. McLaughlin, Strong asymptotics of polynomials orthogonal with respect
to Freud weights, Int. Math. Res. Not. 1999 (1999) 299–333. %doi:10.1155/s1073792899000161.

\bibitem{gsp3} V. K. Kuetche, S. Youssoufa, T. C. Kofane, Ultrashort optical waveguide excitations in uniaxial silica fibers: elastic collision scenarios, Phys. Rev. E 90 (2014) 063203.

\bibitem{49} A. B. J. Kuijlaars, K. T. R. McLaughlin, Asymptotic zero behavior of Laguerre polynomials with
negative parameter, Constr. Approx. 20 (2004) 497–523. % doi:10.1007/s00365-003-0536-3.

\bibitem{op2} A. B. J. Kuijlaars, K. D. T.-R. McLaughlin, W. Van Assche, M. Vanlessen, The Riemann–Hilbert approach to strong asymptotics for orthogonal polynomials on [-1, 1], Adv. Math. 188 (2004) 337–398.

%\bibitem{55} A. B. J. Kuijlaars, K. T. R. McLaughlin, W. Van Assche, M. Vanlessen, The Riemann–Hilbert approach to strong asymptotics for orthogonal polynomials on [-1, 1], Adv. Math. 188 (2004) 337–398.
%doi:https://doi.org/10.1016/j.aim.2003.08.015.


\bibitem{8} A. Kuijlaars, A. Tovbis, On minimal energy solutions to certain classes of integral equations related to
soliton gases for integrable systems, Nonlinearity 34 (2021) 7227–7254. %doi:10.1088/1361-6544/ac20a5.

\bibitem{41} P. D. Lax, C. D. Levermore, The small dispersion limit of the Korteweg-de Vries equation. I,
Comm. Pure Appl. Math. 36 (1983) 253–290. % doi:10.1002/cpa.3160360302.

\bibitem{lenells17} J. Lenells, The nonlinear steepest descent method for Riemann–Hilbert problems of low regularity, Indiana Univ. Math. 66 (2017) 1287–1332

\bibitem{li23} Z.Q. Li, S.-F. Tian, J.-J. Yang, On the asymptotic stability of N-soliton solution for the
short pulse equation with weighted Sobolev initial data, J. Differ. Equ. 377 (2023) 121–187.


\bibitem{wb} Y. Liu, D. Pelinovsky, A. Sakovich, Wave breaking in the short-pulse equation. Dyn. PDE 6 (2009) 291–310.

\bibitem{40} G. D. Lyng, P. D. Miller, The N-soliton of the focusing nonlinear Schr\"odinger equation for N large,
Comm. Pure Appl. Math. 60 (2006) 951–1026. %doi:10.1002/cpa.20162.

\bibitem{BT} H. Mao, Q. Liu, The short pulse equation: B\"acklund transformations and applications, Stud. Appl. Math. 145
(2020) 791–811.


\bibitem{dbar1} K.T.R. McLaughlin, P.D. Miller, The $\bar\partial$ steepest descent method and the asymptotic behavior of polynomials orthogonal on the unit circle with fixed and exponentially varying non-analytic weights, Int. Math. Res. Not. (2006)
48673.

\bibitem{dbar2} K.T.R. McLaughlin, P.D. Miller, The $\bar\partial$ steepest descent method for orthogonal polynomials on the real line with varying weights, Int. Math. Res. Not. (2008) 075.

\bibitem{Mat-1} Y. Matsuno, Multiloop solutions and multibreather solutions of the short pulse model equation, J. Phys. Soc. Jpn.
76 (2007) 084003.


\bibitem{Mat-2} Y. Matsuno, Periodic solutions of the short pulse model equation, J. Phys. Soc. Jpn. 49 (2008) 073508


\bibitem{gsp4} Y. Matsuno, Integrable multi-component generalization of a modified short pulse equation. J. Math. Phys. 57 (2016) 111507.

\bibitem{DT} V. Matveev, M. A. Salle, {\it Darboux Transformations and Solitons} (Springer, New York, 1991).


\bibitem{36} Y. Ohta, J. Yang, Rogue waves in the Davey-Stewartson I equation, Phys. Rev. E 86 (2012) 036604.
%doi:10.1103/physreve.86.036604.

\bibitem{mch-2} P.J. Olver, P. Rosenau, Tri-Hamiltonian duality between solitons and solitary-wave solutions having compact support, Phys. Rev. E 53 (1996) 1900-1906.


\bibitem{cau-2} D. Pelinovsky, A. Sakovich, Global well-posedness of the short-pulse and sine-Gordon equations in energy space,
Commun. Partial Differ. Equ. 35 (2010) 613–629.

\bibitem{mch-3} Z. Qiao, A new integrable equation with cuspons and W/M-shape-peaks solitons, J. Math. Phys. 47 (2006) 112701.


\bibitem{lax-SS} A. Sakovich, S. Sakovich, The short pulse equation is integrable, J. Phys. Soc. Jpn. 74 (2005) 239–241.

\bibitem{SS-1} A. Sakovich, S. Sakovich, Solitary wave solutions of the short pulse equation, J. Phys. A: Math. Gen. 39 (2006)
L361–L367.

\bibitem{gsp1} S. Sakovich, Transformation and integrability of a generalized short pulse equation, Commun. Nonlinear Sci. Numer. Simul. 39 (2016) 21–28.

 \bibitem{sp} T. Sch\"afer,  C.E. Wayne, Propagation of ultra-short optical pulses in cubic nonlinear media, Physica D 196 (2004) 90-105.


\bibitem{gsp5} S. F. Shen, B. Feng, Y. Ohta,  A modified complex short pulse equation of defocusing type, J. Nonlinear Math. Phys. 24 (2017)  195–209.


\bibitem{11} A. Slunyaev, E. Pelinovsky, Role of multiple soliton interactions in the generation of rogue
waves: The modified Korteweg–de Vries framework, Phys. Rev. Lett. 117 (2016) 214501.
%doi:10.1103/physrevlett.117.214501.


\bibitem{20} P. Suret, S. Randoux, A. Gelash, D. Agafontsev, B. Doyon, G. El, Soliton gas: Theory, numerics, and
experiments, Phys. Rev. E 109 (2024) 061001. %doi:10.1103/physreve.109.061001.


\bibitem{tang24} W. Tang, G.-F. Yu, S.-F. Shen, Rogue periodic waves of the short pulse equation and the coupled integrable dispersionless equation, Wave Motion, 124 (2024) 103234.


\bibitem{42} A. Tovbis, S. Venakides, X. Zhou, On semiclassical (zero dispersion limit) solutions of the focusing
nonlinear Schr\"odinger equation, Comm. Pure Appl. Math. 57 (2004) 877–985. %doi:10.1002/cpa.20024.


\bibitem{60} M. Vanlessen, Strong asymptotics of the recurrence coefficients of orthogonal polynomials associated
to the generalized Jacobi weight, J. Approx. Theory 125 (2003) 198–237. %doi:10.1016/j.jat.2003.11.005.

%\bibitem{51} M. Vanlessen, Strong asymptotics of Laguerre-type orthogonal polynomials and applications in random
%matrix theory, Constr. Approx. 25 (2007) 125–175. %doi:10.1007/s00365-005-0611-z.


\bibitem{op3} M. Vanlessen, Strong asymptotics of Laguerre-type orthogonal polynomials and applications in random matrix theory, Constr. Approx. 25 (2007) 125–175.


\bibitem{gsp7} K. Wang, X. Geng, M. Chen, Riemann–Hilbert approach and long-time asymptotics of the positive flow short-pulse equation, Physica D  439 (2022) 133383.


\bibitem{59} G. B. Whitham, {\it Linear and Nonlinear Waves} (Wiley, 1999). %doi:10.1002/9781118032954.

\bibitem{54} R. Wong, W. Zhang, Uniform asymptotics for Jacobi polynomials with varying large negative param-
eters—a Riemann-Hilbert approach, Trans. Amer. Math. Soc. 358 (2006) 2663–2694.


\bibitem{xu18} J. Xu, Long-time asymptotics for the short pulse equation, J. Differ. Equ. 265 (2018) 3494–3532.


\bibitem{gsp6} B. Yang, J. Yang, Transformations between nonlocal and local integrable equations. Stud. Appl. Math. 140 (2018) 178–201.

\bibitem{fan21} Y. Yang, E. Fan, Soliton resolution for the short-pulse equation, J. Differ. Equ. 280 (2021) 644–689.


\bibitem{soliton} N. J. Zabusky, M. D. Kruskal, Interaction of ``Solitons" in a collisionless plasma and the recurrence of initial states, Phys. Rev. Lett. 15 (1965) 240.

\bibitem{1}  V. E. Zakharov, Kinetic equation for solitons, Sov. Phys. -JETP 33 (1971) 538–541.

\bibitem{z1973} V. E. Zakharov, A. Shabat, Interaction between solitons in a stable medium, Sov. Phys.-JETP 37 (1973) 823–828.


\bibitem{34} G. Zhang, L. Ling, Z. Yan, Multi-component nonlinear Schr\"odinger equations with nonzero boundary conditions: Higher-order vector Peregrine solitons and asymptotic estimates, J. Nonlinear Sci. 31
(2021) 81.


\end{thebibliography}

\end{document}